%% file: main.tex
\documentclass[12pt]{article}
\usepackage[letterpaper, margin=1in]{geometry}
\usepackage{amsmath, amssymb, bm, amsthm}  
\usepackage{graphicx}
\usepackage{euscript}
\usepackage{lmodern}       
\usepackage[T1]{fontenc} 
\usepackage[linesnumbered,ruled,vlined]{algorithm2e}
\usepackage{relsize}
\usepackage{natbib}
\usepackage[hyphens]{url}
\usepackage[breaklinks]{hyperref}  
\usepackage{setspace}
\usepackage{ragged2e}  
\usepackage{enumitem}
\usepackage{microtype}  
\usepackage[font=footnotesize]{caption} 
\usepackage{subcaption}
\usepackage{cancel}
\usepackage{booktabs}
\usepackage{multirow}
\usepackage{makecell}
 \usepackage{longtable, array}
\usepackage{easyReview}
\usepackage{float}

\setreviewsoff

\newcommand{\Var}{\text{Var}}

\def\EE{{\mathbb E}}

\def\RR{{\mathbb R}}

\def\nn{{\EuScript N}}

\newcommand{\indep}{\perp \!\!\! \perp}

\newcommand*{\V}[1]{\mathbf{#1}}  

\newtheorem{theorem}{Theorem}
\newtheorem{assumption}{Assumption}
\newtheorem{proposition}{Proposition}
\newtheorem{lemma}{Lemma}
\newtheorem{definition}{Definition}

\newtheorem{manualconditioninner}{Condition}
\newenvironment{manualcondition}[1]{%
	\renewcommand\themanualconditioninner{#1}%
	\manualconditioninner
}{\endmanualconditioninner}

\title{Heterogeneity-robust granular instruments\thanks{Email: \href{mailto:ericqian@jhu.edu}{\texttt{ericqian@jhu.edu}}.  I am grateful for the support and guidance from my third year paper advisor Mikkel Plagborg-M\o ller. I thank Mark Aguiar, Matias Cattaneo, Bo Honor\'e, John Grigsby, Michal Koles\'ar, Ulrich M\"{u}ller, Richard Rogerson, Karthik Sastry, Mark Watson, and participants from various workshops for their helpful comments.}}
\author{Eric Qian  \\
	Johns Hopkins University\\
}

\date{\today}
	
\begin{document}
		\onehalfspacing
		\maketitle
		\justifying
		
		\begin{abstract}
			\noindent Granular instrumental variables (GIV) has experienced sharp growth in empirical macro-finance. The methodology’s rise showcases granularity’s potential for identification across many economic environments, like the estimation of spillovers and demand systems. I propose a new estimator—called robust granular instrumental variables (RGIV)—that enables studying unit-level heterogeneity in spillovers. Unlike existing methods that assume heterogeneity is a function of observables, RGIV leaves heterogeneity unrestricted. In contrast to the baseline GIV estimator, RGIV allows for unknown shock variances and does not require skewness in the size distribution. \alert{I find evidence of unit-level heterogeneity in applications to sovereign yield spillovers and the inelastic markets hypothesis.}
		\end{abstract}
\clearpage\newpage		
\section{Introduction}

In many macroeconomic settings, researchers wish to estimate the spillover of an idiosyncratic shock to one unit onto other units. Examples include estimating the contagion effects of financial market distress, aggregate demand externalities on individual consumption, strategic complementarities in price setting, and the price elasticity of demand of the stock market \citep{Allen2009,Auclert2023,Alvarez2023,Gabaix2023}. Identifying plausibly exogenous variation in these settings however is notoriously difficult and contributes to the continued challenges of conducting empirical work.

\cite{Gabaix2024} tackles this difficulty with a novel technique, called granular instrumental variables (GIV).   It considers environments where an individual unit's outcome is partially determined by the size-weighted outcome. Thus, idiosyncratic shocks to a single unit spill over to all other units in equilibrium. To overcome the resulting endogeneity bias, their instrument for the size-weighted outcome is constructed as the difference between size- and equal-weighted outcomes. Their instrument is ``granular'' in that idiosyncratic shocks to large units are the primary source of identifying variation as they disproportionately contribute to movements in the size-weighted outcome variable. The broad adoption of GIV in applied macroeconomics and finance is indicative of its usefulness \citep{Chodorow-Reich2021a,Adrian2022,Camanho2022,Gabaix2023}.

The baseline GIV estimator of \cite{Gabaix2024}, however, requires strong assumptions to satisfy the instrumental variables relevance condition and exclusion restriction and its extensions require further conditions. The baseline estimator requires homogeneous spillovers across units, homogeneous shock variances,  skewed unit size, and idiosyncratic shocks (uncorrelated shocks between units after accounting for a \emph{known} factor structure).   In particular, there is typically no ex-ante rationale for there to be homogeneous spillovers and shock variances across units, contributing to the gap between theory and practice. Furthermore, I show that there is no general guarantee that the GIV  spillover estimand admits a positive-weighted average of unit-specific spillovers. To address these concerns, \cite{Gabaix2024} generalizes their procedure to allow for heterogeneity that is a function of observables, but the empirically important question of how to best tackle unobserved heterogeneity remains open.

This paper's main contribution is to establish global identification for a GIV model with unit-specific spillovers and unknown shock variances. I propose an estimator, called robust granular instrumental variables (RGIV),  that is robust in the sense that it is applicable to a wider set of environments than the baseline \cite{Gabaix2024} GIV estimator.  Ultimately, RGIV brings the study of unit-level heterogeneity---a feature that is of substantial importance in other areas of macroeconomics{\interfootnotelinepenalty10000 \footnote{For instance, take household-level differences in the marginal propensity to consume akin to \cite{Fagereng2021,Lewis2021,Fuster2021}.}}---to the estimation of spillovers. 

 RGIV uses internally estimated idiosyncratic shocks as instruments for the size-weighted outcome variable. Informally, \alert{identification can be understood} sequentially.  An initial guess of the first unit's spillover gives an estimate of the first unit's idiosyncratic shock. Using this estimated shock as an instrument for the size-weighted outcome variable, the remaining units' spillover coefficients and estimated idiosyncratic shocks can be computed. If one or more pairs of estimated idiosyncratic shocks are correlated, we can guess a new spillover for the first unit and repeat until the estimated idiosyncratic shocks are uncorrelated. \alert{The estimator I propose minimizes the GMM objective jointly}. RGIV makes use of granularity by exploiting the contribution of individual shocks (even to relatively small units) to movements in the size-weighted outcome variable. This approach differs from the baseline GIV estimator in that skewness in the size distribution of units and homogeneity in shock variances are not required.

I prove that there is an inevitable tradeoff in assumptions between allowing for general spillover heterogeneity and a general shock covariance structure. Like GIV, the RGIV estimator can easily accommodate cross-unit correlations that are due to observed covariates (both time-varying and unit-specific) or a known factor structure in the residuals. However, I prove that a model with an \textit{unknown} factor structure and unrestricted spillover heterogeneity is not identified. Hence, absent additional structure, researchers must restrict either spillover heterogeneity or the correlation structure of the shocks.

I develop tests that evaluate the appropriateness of (1) the RGIV framework and (2) the assumption of homogeneous spillovers featured in \cite{Gabaix2024}.  These hypothesis tests are applications of standard GMM results \citep{Newey1994}. The first test is the Sargan--Hansen test, where the null hypothesis of correct specification is rejected when one or more pairs of estimated idiosyncratic shocks are correlated. The second test is the distance metric test, where the null hypothesis is rejected when the constraint of spillover homogeneity binds.

I demonstrate the usefulness of RGIV in two empirical applications.

 First, building on the analysis of a working paper version of \cite{Gabaix2024}, I apply RGIV to sovereign yield spillovers in the Euro area and find strong evidence of spillover heterogeneity among Euro area members. I control for correlations induced by heterogeneous loadings on unobserved aggregate shocks by including a rich set of observed explanatory variables. In addition, I account for the well-documented correlation of shocks among Euro area countries by size-aggregating countries into larger  ``core'' and ``periphery'' blocks \citep{Bayoumi1992}. Under the preferred specification, I fail to reject the null hypothesis of correct specification. Meanwhile,  the null hypothesis of spillover coefficient homogeneity is strongly rejected. I find that the size-weighted spillovers of countries in the western-periphery block (Ireland, Portugal, and Spain) are approximately twice that of the core country block. 
 
 \alert{Second, based on the empirical work of \citet{Gabaix2023}, I apply RGIV to study the ``inelastic markets hypothesis,'' that the aggregate stock market has a low price elasticity of demand. After including a wide set of observed explanatory variables, I apply RGIV after size-aggregating by operational similarities. The RGIV estimates support the inelastic markets hypothesis, with a size-weighted price elasticity of demand of $-0.05$. I also reject the null hypothesis of between-block homogeneity. I find that demand is most responsive for the Retirement and Insurance block (elasticity of $-0.2$), relatively inelastic for both the Real Domestic and Foreign blocks (elasticities of $-0.08$ and near zero respectively) and positively-sloped for the Intermediaries and Funds block (elasticity of $+0.09$).}

In a simulation study, I show that RGIV confidence intervals have good finite sample coverage properties \alert{based on a} DGP taken from the \alert{sovereign yield spillovers} empirical application. GIV confidence intervals in contrast can undercover under spillover heterogeneity and when shock variances are taken to be unknown. \alert{I also show that coverage for RGIV worsens for a large number of moments and that GIV coverage worsens with insufficient skewness in the size distribution (relative to the sample length)}. \\

\noindent \textsc{Literature}{} {} {} My model generalizes the baseline environment of \cite{Gabaix2024} to allow for spillover heterogeneity, shock variance heterogeneity, and equal-sized units. For estimation, RGIV fully exploits second moment information through GMM and can also be adapted to the case in which correlations in shocks are induced by factors with known loadings. In contrast, the spillover coefficient homogeneity, skewed size, and known shock variance conditions are used to justify the validity of the GIV instrument. \cite{Gabaix2024} also proposes extensions that separately estimate spillover coefficient heterogeneity\footnote{\cite{Gabaix2024} Proposition 6 considers the case where unit-level spillover heterogeneity is determined by observables, Section D.9.1 proposes unit-specific ``leave-one-out'' granular instruments for when shock variances are known, and Section D.9.2 lists moment conditions for adapting GIV to heterogeneous spillovers though without the analysis of its econometric properties.} and unknown shock variances\footnote{\cite{Gabaix2024} Section D.3 appends shock variance moments after assuming homogeneous spillovers and global identification.}. In contrast, RGIV doesn't require skewness of the size distribution, heterogeneity to depend on observables, or for shock variances to be known; uncorrelated shocks is sufficient for jointly allowing unit-level spillovers and unknown shock variances. 

\alert{\cite{Baumeister2023a} applies the insights of GIV to a rich, structural model of the oil market (with dynamics, unobserved explanatory variables, and inventories) in a full information maximum likelihood framework. While they consider a specification with elasticity heterogeneity, they do not develop formal theoretical results on identification in this setting. This paper instead adapts the parsimonious framework of \cite{Gabaix2024} to develop theoretical results on elasticity heterogeneity---on identification, non-identification with latent factors, and limiting behavior---using a moment condition framework. This paper ultimately seeks to provide researchers with theoretical results that can be adapted to their own applications.}

\cite{Banafti2022} proposes a refinement to the GIV estimator, allowing for unobserved explanatory variables but require the number of economic units to be large and the size distribution to be skewed, unlike RGIV. Here, the principal challenge is to prevent the law of large numbers from averaging away the granularity of units.  The authors overcome this challenge by requiring the right tail of the unit size distribution to be highly skewed, with a Pareto tail index of less than 1.  In contrast, RGIV is consistent even when the size distribution is uniform.

GIV procedures are closely linked to models found in the spatial panel and peer effects literature \citep{Su2023,Aquaro2021,Chen2022,Manski1993}. GIV's spillover network structure can be viewed as a restricted form of spatial autocorrelation. Unlike these papers, I allow the number of units to be small (finite), show that the GIV spillovers under unit-level heterogeneity are globally identified, and estimate spillovers without the use of external instruments. More broadly, my procedure can also be viewed as a special case of a simultaneous equation model with covariance restrictions, which is closely related to impact matrix identification in the structural VAR setting \citep{Sims1980, Hausman1983a}. \\

\noindent \textsc{Outline}{}{}{} Section \ref{sec:illustration} illustrates the properties of GIV and RGIV in a simple three unit setting. Section \ref{sec:general} presents the main consistency and asymptotic normality results for the RGIV estimator. Section \ref{sec:testing} describes the RGIV specification test and parameter homogeneity tests. Section \ref{sec:extensions} extends RGIV to include observable explanatory variables and discusses non-identification under unobserved explanatory variables. Section \ref{sec:application} applies the RGIV estimator to investigate sovereign yield spillovers in the Euro area and the inelastic markets hypothesis. Section \ref{sec:simulation} examines finite-sample coverage accuracy of RGIV and GIV through simulations.  Section \ref{sec:discussion} concludes.

\section{Simple illustration}
\label{sec:illustration}

The remainder of this paper will consider the sovereign yield spillovers application as a running example.  An idiosyncratic shock to one Euro area country raises yields in that country, as well as those of other Euro area countries since losses from default are partially shared. Country-specific spillovers are permitted.

In this section, I use a simple three-country setting to illustrate the construction and properties of GIV and RGIV. Section \ref{sec:general} formalizes the discussion and generalizes to $n$ countries. \\

\noindent \textsc{Notation}{} {} {} Throughout this paper, I will follow the notation of \cite{Gabaix2024} for convenience. For a vector $X = (X_i)_{i=1,...,n}$ and size $S_i$ satisfying $\sum_{i=1}^n S_i = 1$, the equal-weighted sum $X_E$ and size-weighted sum $X_S$ are defined below:
\begin{align*}
	X_E \equiv \frac{1}{n} \sum_{i=1}^n X_i, \quad X_S \equiv \sum_{i=1}^n S_i X_i.
\end{align*}
Similarly, the equal-weighted and size-weighted cross-sectional averages for time series data $X_t = (X_{it})_{i=1,...,n}$ are defined as $X_{Et} \equiv \frac{1}{n} \sum_{i=1}^n X_{it}$ and $X_{St} \equiv \sum_{i=1}^n S_i X_{it}$ respectively.

\subsection{Baseline three country setting}
\label{ill:general_setting}

I will outline my ``baseline'' setting for the three country case using sovereign yield spillovers in the Euro area as a running example. For countries $i=1,2,3$, let $y_{it}$ be the yield spread relative to some comparison country. Yield spread growth is $r_{it} = \frac{ y_{it}- y_{i,t-1}}{y_{i,t-1}}$. Size $S_i$ is observed and corresponds to a country's ``debt at risk.''  Here, $S_i \in (0,1)$ is taken to be time-invariant and sums to 1.

Suppose the researcher is interested in estimating elasticity $\phi_i$, which will be called the ``spillover coefficient.'' Specifically, yield spread growth $r_{it}$ is determined by a country-specific spillover coefficient $\phi_i$, size-aggregated yield spread $r_{St}$, and idiosyncratic shock $u_{it}$:
\begin{align}
	\label{eq:illustration}
	r_{it} &= \phi _i r_{St} + u_{it}, \quad i=1,...,n, \quad 	\phi_S < 1.
	\end{align}
The size-weighted spillover coefficient $\phi_S$ is taken to be less than 1. Moreover, shocks $u_{it}$ are mutually  and serially independent, mean zero, and have country-specific variance $\EE(u_{it}^2) = \sigma^2_i > 0$. 

The propagation of an idiosyncratic shock depends on the size of the shock's origin country, the recipient country's spillover coefficient, and the size-weighted average spillover coefficient. To illustrate, consider a unit idiosyncratic shock to country 1 ($u_{1t} = 1$) holding the idiosyncratic shocks of countries 2 and 3 at zero ($u_{2t} = u_{3t} =0$). Then, the size-weighted yield spread growth $r_{St}$ increases by $S_1 \times \frac{1}{1-\phi_S}$, which is computed from taking a size-weighted average of Equation \ref{eq:illustration}:
\begin{align}
	\sum_{i=1}^3 S_i r_{it} &= \sum_{i=1}^3 S_i (\phi _i r_{St} + u_{it}) \implies r_{St} = \frac{u_{St}}{1 - \phi_S}. \label{eq:ill_multiplier}
\end{align}
 The increase in $r_{St}$ is larger when the size of country 1 is large and when countries are, on average, sensitive to spillovers (from multiplier $\frac{1}{1-\phi_S}$). Then, idiosyncratic shock $u_{1t}$ spills over to $r_{2t}$ and $r_{3t}$, giving rise to increases of $\phi_j r_{St} = \phi_j\frac{S_1}{1-\phi_S}$ for $j=2,3$. The total increase in $r_{1t}$ is $1+\phi_1 \frac{S_1}{1-\phi_S}$, a composition of the direct effect of idiosyncratic shock $u_{1t}$ and a spillover effect.

Granularity rules out the ordinary least squares regression of $r_{it}$ on $r_{St}$ as a method for estimating $\phi_i$. Equation \ref{eq:ill_multiplier} highlights granularity in the baseline setting, showing that idiosyncratic shocks are responsible for movements in the aggregated yield spread $r_{St}$ since $S_i>0$ and  $\sigma^2_i>0$. Therefore, regressing $r_{it}$ on $r_{St}$ to estimate $\phi_i$ as $T \to \infty$ suffers from endogeneity bias since $\EE(r_{St} u_{1t}) = \frac{S_1 \sigma^2_1}{1-\phi_S} \neq 0$. An alternative estimator is needed.

\subsection{Illustration of GIV}

In this section, I review the baseline GIV estimator of \cite{Gabaix2024} applied to this simple setting and discuss the conditions for its validity.  I show that a skewed size distribution is necessary for the relevance condition to hold. Shock orthogonality and homogeneous idiosyncratic shock variances are necessary for the instrumental variables exclusion restriction to hold. Moreover, the exclusion restriction fails under heterogeneity of spillover coefficients across countries. As discussed later, there are GIV extensions that individually accommodate spillover coefficient heterogeneity and heterogeneous (and unknown) shock variances, but these require additional conditions.

Take the setting of Section \ref{ill:general_setting} and further assume homogeneous spillover coefficients  ($\phi_1=\phi_2=\phi_3 = \phi$), homogeneous shock variances ($\sigma^2_1= \sigma^2_2 = \sigma^2_3 = \sigma^2$), and skewed unit sizes (ruling out $S_1 = S_2 = S_3$). Just as in Section \ref{ill:general_setting}, granularity induces endogeneity bias in the regression of $r_{it}$ on $r_{St}$ for regression coefficient $\widehat{\phi}^i_{OLS}$
\begin{align*}
	\widehat{\phi}^{i}_{OLS} - \phi \xrightarrow{p} \frac{\EE(r_{St} r_{it})}{\EE(r_{St}^2)} - \phi = \frac{S_i \sigma^2}{\EE(r_{St}^2)[1- \phi]} > 0.
\end{align*} 

\cite{Gabaix2024} propose estimating $\phi$ using instrumental variables. An equal-weighted aggregation of the country-level yield spread growth gives rise to the following IV regression model
\begin{align*}
 r_{Et} &= \phi r_{St} + u_{Et}.
\end{align*}
The equal-weighted yield spread growth responds to the size-weighted yield spread growth $r_{St}$ (through spillover coefficient $\phi$) and equal-weighted idiosyncratic shocks $u_{Et}$. For the IV regression model in the above display, ``granular instrument'' $z_t$ is constructed as the difference between the size- and equal-weighted yield spread growth
\begin{align*}
	z_t := r_{St} - r_{Et} = \big(\phi r_{St} + u_{St}\big) - \big( \phi r_{St} + u_{Et}\big) =u_{St} - u_{Et}.
\end{align*}
From the homogeneous spillovers assumption, the endogeneity from the size-weighted yield spread growth $r_{St}$ is differenced away.

Skewness in country sizes is crucial for the instrumental variables relevance condition to hold and $z_t$ reflects variation coming from idiosyncratic shocks to relatively large countries. To illustrate, first consider the extreme case of no skewness ($S_1 = S_2 = S_3$). Then the size- and equal-weighted yield spread growths are identical, producing a granular instrument that is identically zero $z_t = 0$. More generally, the instrumental variables relevance condition can be explicitly computed:
\begin{align*}
		\EE\left[z_t r_{S t}\right] & =\frac{1}{1-\phi} \mathbb{E}\left[u_{S t}\left(u_{S t}-u_{E t}\right)\right] \\
		&=\frac{\sigma^2}{1-\phi}\left[S_1\left(S_1-\frac{1}{3}\right)+S_2\left(S_2-\frac{1}{3}\right)+S_3\left(S_3-\frac{1}{3}\right)\right].
\end{align*}
When the size of country 1 is large relative to other countries (when $S_1 \gg S_2,S_3$), the relevance condition is further from zero reflected by the $S_1(S_1 - 1/3)$ term. 

The shock orthogonality and homogeneous (or known) shock variance assumptions are crucial for the exclusion restriction to hold. Explicitly, the covariance between the instrument and IV regression model error $u_{Et}$ is
\begin{align*}
\EE[u_{Et} (u_{St} - u_{Et})] &= \frac{1}{3}\big[S_1 \sigma^2 + S_2 \sigma^2 + S_3 \sigma^2\big] - \frac{1}{3}\big[\frac{1}{3} \sigma^2 + \frac{1}{3} \sigma^2 + \frac{1}{3} \sigma^2\big] \\
&= \frac{1}{3} \sigma^2 - \frac{1}{3} \sigma^2 = 0.
\end{align*}
In the first line of the above display, the uncorrelated shock assumption ensures that covariance terms $\EE(u_{it} u_{jt})$ are zero. The homogeneous shock variance assumption ensures that the difference between the first line's bracketed terms is zero. Note that any instrument of the form $z_t' = W_1 r_{1t} + W_2 r_{2t} + W_3 r_{3t} - \frac{1}{3}r_{1t} - \frac{1}{3}r_{2t}  - \frac{1}{3} r_{3t}$ where $W_1 + W_2+W_3= 1$ satisfies the instrumental variables exclusion restriction. \cite{Gabaix2024} show that weighting by size gives a variance-minimizing estimator for the IV regression model. The spillover coefficient can also be consistently estimated when a time fixed effect is included, as the time fixed effect is differenced away in the construction of $z_t$. 

The exclusion restriction argument described in the preceding paragraph can easily be extended to settings in which the shock variances are known to the econometrician rather than being homogeneous across units. \cite{Gabaix2024} show that identical computations hold after replacing the equal weights in $z_t = r_{St} - r_{Et}$ with inverse variance weights.

Spillover coefficient heterogeneity leads to a failure in the instrumental variables exclusion restriction. Allowing for unit-specific spillover coefficients, the granular instrument contains variation from $r_{St}$:
\begin{align}
	z_t &= r_{St} - r_{Et} = (\phi_S - \phi_E) r_{St} + u_{St} - u_{Et}. \label{equation:GK_contamination} 
\end{align}
From the $(\phi_S - \phi_E)$ term, the magnitude of the exclusion restriction's violation is greater when spillover coefficients vary systematically by size. Moreover, I show later in Proposition \ref{prop:GIV_het_elasticities} that there is no guarantee that the GIV spillover coefficient estimand is a non-negative weighted average of spillover coefficients. In this sense, the GIV spillover coefficient estimand could be far from the potentially heterogeneous true spillover coefficients.

\subsection{RGIV illustration: Environment as GMM moment conditions}
This subsection introduces the estimator proposed in this paper, robust granular instrumental variables (RGIV). RGIV exploits the uncorrelatedness of idiosyncratic shocks through the generalized method of moments. The estimator's identifying variation comes from individual country-level idiosyncratic shocks.

Returning to the baseline environment described in Section \ref{ill:general_setting}---which again features unit-specific spillover coefficients $\phi_i$ and unit-specific shock variances $\sigma^2_i$---the RGIV estimator encodes the uncorrelatedness of idiosyncratic shocks through GMM moment conditions. Taking sizes $S_1, S_2, S_3$ as given, store data in the vector $\V{r}_t = [r_{1t}, r_{2t}, r_{3t}]'$ and parameters in $\boldsymbol{\phi} = [\phi_1, \phi_2, \phi_3]'$. Letting $u_{i}(\V{r}_t, \boldsymbol{\phi}) = r_{it} - \phi_i r_{St}$, moment function $g(\V{r}_t, \boldsymbol{ \phi})$ encodes the condition that idiosyncratic shocks are uncorrelated
\begin{align*}
	g(\V{r}_t, \boldsymbol{\phi}) &= \begin{bmatrix}
		u_{1}(\V{r}_t, \boldsymbol{\phi}) u_{2}(\V{r}_t, \boldsymbol{\phi})& u_{1}(\V{r}_t, \boldsymbol{\phi}) u_{3}(\V{r}_t, \boldsymbol{\phi})& u_{2}(\V{r}_t, \boldsymbol{\phi}) u_{3}(\V{r}_t, \boldsymbol{\phi})
	\end{bmatrix}' 
\end{align*}
where 	$\EE[g(\V{r}_t, \boldsymbol{ \phi}_0)] = 0$ for true parameter $\boldsymbol{ \phi}_0$. The parameter-dependent weight matrix is $\widehat{W}(\boldsymbol{\phi}) = \textrm{diag}\big(\frac{1}{\widehat{\sigma}^2_1(\boldsymbol{\phi})\widehat{\sigma}^2_2(\boldsymbol{\phi})}, \frac{1}{\widehat{\sigma}^2_1(\boldsymbol{\phi})\widehat{\sigma}^2_3(\boldsymbol{\phi})}, \frac{1}{\widehat{\sigma}^2_2(\boldsymbol{\phi})\widehat{\sigma}^2_3(\boldsymbol{\phi})}\big)$ for $\widehat{\sigma}^2_i(\boldsymbol{\phi}) = \frac{1}{T} \sum_{t=1}^T u_{i}(\V{r}_t, \boldsymbol{\phi}) ^2$. Then, for parameter space $\boldsymbol{\Phi}$ where $\phi_S < 1$ for $\phi \in \Phi$, the robust granular instrumental variables (RGIV) estimator is defined as a continuously updating GMM estimator:
\begin{align}
	\widehat{\boldsymbol{ \phi}}^{RGIV} &= \arg\min_{\boldsymbol{\phi} \in \boldsymbol{\Phi}} \Big(\frac{1}{T} \sum^T_{t=1} g(\V{r}_t, \boldsymbol{ \phi})\Big)' \widehat{W}(\boldsymbol{ \phi}) \Big(\frac{1}{T} \sum^T_{t=1} g(\V{r}_t, \boldsymbol{ \phi})\Big). \label{ill:gmm}
\end{align}
$\widehat{W}(\boldsymbol{ \phi})$ is an efficient GMM weight matrix when idiosyncratic shocks are independent. Avoiding inversion of a potentially non-diagonal matrix, $\widehat{W}(\boldsymbol{ \phi})$ also ensures numerical stability if the initialization of  $\boldsymbol{\phi}$ is far from the GMM objective function's minimum.

Equivalently, the continuously updating GMM formulation minimizes the average squared correlation coefficients between pairs of estimated idiosyncratic shocks. To see this, define the estimated correlation coefficient of idiosyncratic shocks as $\widehat{\rho}_{ij}(\boldsymbol{\phi}) = \frac{\frac{1}{T} \sum_{t=1}^T u_{i}(\V{r}_t, \boldsymbol{\phi})u_{j}(\V{r}_t, \boldsymbol{\phi}) }{\sqrt{ \widehat{\sigma}_i^2(\boldsymbol{\phi})\widehat{\sigma}_j^2(\boldsymbol{\phi})} }$. Then, the RGIV estimator is
\begin{align*}
	\widehat{\boldsymbol{\phi}}^{RGIV} &= \arg \min_{\boldsymbol{\phi} \in \boldsymbol{\Phi}}  \frac{1}{3}\left[\widehat{\rho}_{12}(\boldsymbol{\phi})^2 + \widehat{\rho}_{13}(\boldsymbol{\phi})^2 + \widehat{\rho}_{23}(\boldsymbol{\phi})^2 \right]
\end{align*}
after multiplying the objective function in Equation \ref{ill:gmm} by $\frac{1}{3}$. Intuitively, RGIV chooses the spillover coefficient vector $\boldsymbol{\phi}$ that makes the estimated shocks the least correlated.

RGIV also admits an instrumental variables interpretation, as a country's spillover coefficient is estimated using information from internally estimated idiosyncratic shocks to other countries. The asymptotic variance matrix of the RGIV estimator  is $V = (G' \Sigma^{-1}G)^{-1}$  for Jacobian matrix $G = \EE[\nabla_{\phi} g(\V{r}_t, \boldsymbol{ \phi}_0)]$ and moment covariance matrix $\Sigma = \EE[g(\V{r}_t, \boldsymbol{ \phi}_0) g(\V{r}_t, \boldsymbol{ \phi}_0)']$. Then the diagonal entries of $V$ are
\begin{align*}
\operatorname{Avar}\left(\widehat{\phi}_i^{RGIV}\right)=\frac{\sigma_i^2}{\prod_{j \neq i}\left(S_j^2 \sigma_j^2\right)} \frac{\left(1-\phi_S\right)^2\left(S_1^2 \sigma_1^2+S_2^2 \sigma_2^2+S_3^2 \sigma_3^2\right)}{4} .
\end{align*}
The above display illustrates that the identifying variation of the RGIV estimator doesn't require skewness in the unit size distribution, as the identifying variation comes from individual idiosyncratic shocks; estimates for $\widehat{\phi}_i^{RGIV}$ are more precise when idiosyncratic shocks to countries $j \neq i$ have higher variance \alert{and are relatively large---giving the analogous ``relevance condition.''} \alert{Intuitively, identification can be viewed sequentially where shock uncorrelatedness can be viewed as the analogous ``exclusion restriction.''} Guessing the spillover coefficient of country 1, the resulting estimated idiosyncratic shock to country 1 can be used as an instrument for the estimation of the spillover coefficients for countries 2 and 3. If one or more pairs of estimated idiosyncratic shocks are too correlated, the procedure is repeated for a new spillover coefficient. \alert{In practice, the GMM objective function is minimized jointly.} 

\section{General model and main results}
\label{sec:general}
This section describes the assumptions needed for the robust GIV estimator for $n<\infty$ countries.

\subsection{Assumptions}

\begin{assumption}{(Baseline model)}
	 	
\label{ass:general}
\begin{enumerate}[label = (\roman*), ref= \theassumption(\roman*)]
\item \textbf{Model}: For fixed $n  \geq 3$ units, let known sizes $S_i \in (0,1)$ sum to 1.  Outcome $\V{r}_t =[r_{1t},\dots,r_{nt}]'$ responds to the size-aggregated outcome $r_{St}$ according to spillover coefficient $\phi_i$ and unobserved shocks $\V{u}_{t} = [u_{1t}, u_{2t}, \dots,u_{nt}]'$ \label{ass:general_reducedform}
	\begin{align*}
		r_{it} &= \phi_i r_{St} + u_{it}, \quad \forall i = 1,...,n, \quad \phi_S < 1.
	\end{align*}

\item \textbf{Shock moments}: For $\sigma^2_{i} > 0$, shocks $\V{u}_t$ are i.i.d. with moments $\EE(\V{u}_t) = 0$, $\EE(\V{u}_t \V{u}_t') =  \mathrm{diag}(\sigma^2_1, \dots,\sigma^2_n)$,  and $\EE(\| \V{u}_t\|^4) < \infty.$ Moreover, $u_{it}$ is independent of $u_{jt}$ for $i \neq j$.  \label{ass:general_shockmoments}

\item \textbf{Parameter space}: 
For spillover coefficient $\boldsymbol{\phi} = [\phi_1,\dots,\phi_n]'$,  the true parameter $\boldsymbol{\phi}_0$ is in the interior of parameter space $\boldsymbol{\Phi}$.  $\boldsymbol{\Phi}$ is compact and for any $\boldsymbol{\widetilde{\phi}}  \in \boldsymbol{\Phi}$,  $\widetilde{\phi}_S   < 1$. \label{ass:general_parameterspace}
\end{enumerate}
\end{assumption}

In Assumption \ref{ass:general_reducedform}, the outcome variable $r_{it}$ responds to the size-aggregated outcome $r_{St}$ according to spillover coefficient $\phi_i$ and idiosyncratic shock $u_{it}$. Through $\phi_S< 1$, positive idiosyncratic shocks increase the size-weighted outcome $r_{St}$. As will be outlined below, the RGIV estimator is just-identified for $n=3$ and over-identified for $n>3$. Moreover, size $S_i$ is taken to be \textit{known} by the econometrician. Hence, the model presented in \ref{ass:general_reducedform} can be modified to allow for time-varying size (for $S_{it} \in (0,1)$ and $\sum_{i = 1}^n S_{it}=1$) without changing the proofs to follow.  In Assumption \ref{ass:general_parameterspace}, the parameter space is assumed to be compact and restricted to encode the sign restriction of $\phi_S< 1$. Compactness is a standard technical assumption for the consistency of extremum estimators \citep{Newey1994}.

\subsection{RGIV estimator}
\label{sec:general_RGIVestimator}

The GMM estimator established in Definition \ref{defn:naive} below (called RGIV) exploits the uncorrelated unit-specific shock condition established in Assumption \ref{ass:general}(ii). RGIV is constructed as a continuously updating GMM estimator \citep{Hansen1996}.  The moment function $g(\V{r}_t, \boldsymbol{ \phi})$ contains the pairwise products of each of the estimated shocks, and the weight matrix inversely weights each moment by the product of the respective estimated shock variances. Exploiting shock uncorrelatedness can be best understood as being consistent with the tradition of exploiting second moment information in the traditional SVAR identification setting \citep{Kilian2017,Leeper1996}. 

\begin{definition}
	\label{defn:naive}
For  outcome variable $\V{r}_t = [r_{1t}, \dots,r_{nt}]'$, let $u_{i}(\V{r}_t, \boldsymbol{\phi}) = r_{it}  - \phi_i r_{St} $ for $i=1,...,n$.
The moment function $g(\V{r}_t, \boldsymbol{\phi})$  is
{\small
\begin{align*}
		g(\V{r}_t, \boldsymbol{\phi}) &= [u_{1}(\V{r}_t, \boldsymbol{\phi}) u_{2}(\V{r}_t, \boldsymbol{\phi}), \dots,\,u_{1}(\V{r}_t, \boldsymbol{\phi}) u_{n}(\V{r}_t, \boldsymbol{\phi}),\,u_{2}(\V{r}_t, \boldsymbol{\phi}) u_{3}(\V{r}_t, \boldsymbol{\phi}), \dots ,\text{}u_{n-1}(\V{r}_t, \boldsymbol{\phi}) u_{n}(\V{r}_t, \boldsymbol{\phi})]'.
\end{align*}
}
For $\widehat{\sigma}^2_i(\boldsymbol{\phi}) = \frac{1}{T} \sum_{t=1}^T u_{i}(\V{r}_t, \boldsymbol{\phi}) ^2$, define the sample weight matrix 
{ \small \begin{align*}
\widehat{W}(\boldsymbol{ \phi}) = \mathrm{diag}\Big(&\frac{1}{\widehat{\sigma}^2_1(\boldsymbol{\phi})\widehat{\sigma}^2_2(\boldsymbol{\phi})}, \dots,\,\frac{1}{\widehat{\sigma}^2_1(\boldsymbol{\phi}) \widehat{\sigma}^2_n(\boldsymbol{\phi})},\,\frac{1}{\widehat{\sigma}^2_2(\boldsymbol{\phi}) \widehat{\sigma}^2_3(\boldsymbol{\phi})}, \dots, \, \frac{1}{\widehat{\sigma}^2_{n-1}(\boldsymbol{\phi}) \widehat{\sigma}^2_n(\boldsymbol{\phi})} \Big).
\end{align*}}

Then, for GMM objective function $\widehat{Q}_T(\boldsymbol{\phi})=  \Big[\frac{1}{T} \sum_{t=1}^T g(\V{r}_t, \boldsymbol{\phi}) \Big]' \widehat{W}(\boldsymbol{\phi}) \Big[\frac{1}{T} \sum_{t=1}^T g(\V{r}_t, \boldsymbol{\phi})\Big]$, the \textbf{robust granular instrumental variables (RGIV)} estimator is \\$
\widehat{\boldsymbol{\phi}}^{RGIV} = \arg\min_{\boldsymbol{\phi} \in \boldsymbol{\Phi}} \widehat{Q}_T(\boldsymbol{\phi} )$.
\end{definition}

The RGIV estimator is robust in that the estimator allows for unit-specific spillover coefficient heterogeneity while also allowing for unknown shock variances and equal unit sizes. When the researcher is  \textit{a priori} certain of homogeneous spillover coefficients $\phi_i = \phi$, the RGIV moment vector can accommodate this parameter restriction by restricting the elements of parameter vector $\boldsymbol{ \phi}$ to be homogeneous across units yielding possible gains in efficiency. The assumption of parameter homogeneity is also testable as will be discussed in Section \ref{sec:testing}.

\begin{lemma}[Identification]
	\label{lemma:gmm_id}
	Impose Assumption \ref{ass:general}. For $g_0(\boldsymbol{ \phi}) = \EE[g(\V{r}, \boldsymbol{ \phi})]$, $g_0(\boldsymbol{ \phi}_0)= 0$ for the true parameter $\boldsymbol{\phi}_0$ and $g_0(\widetilde{\boldsymbol{ \phi}})  \neq 0$ for $\widetilde{\boldsymbol{ \phi}} \in \boldsymbol{\Phi}$ such that $\widetilde{\boldsymbol{ \phi}} \neq \boldsymbol{ \phi}_0$.
\end{lemma}
\begin{proof}
	See Appendix \ref{lemma:gmm_id_proof}.
\end{proof}

For Lemma \ref{lemma:gmm_id}, the parameter space restriction $\widetilde{\phi}_S  < 1$ in Assumption \ref{ass:general_parameterspace} rules out the false solution to the population moment condition. The proof of Lemma \ref{lemma:gmm_id} shows that there is a second parameter $\check{\boldsymbol{ \phi}}$ such that $g_0(\check{\boldsymbol{ \phi}}) = 0$. However, $\check{\boldsymbol{ \phi}}$ is not in the parameter space because $\check{\phi}_S = 1 + (1-\phi_S) > 1$ and thus is not a candidate solution.  In economic terms, Assumption \ref{ass:general_parameterspace} represents the researcher's knowledge of the sign of an idiosyncratic shock's effect on the size-weighted outcome $r_{St}$. This knowledge can come from an application's institutional details; for the case of the running example of Euro area yield spread spillovers, a positive idiosyncratic shock to one country gives rise to an increase in aggregated yield spreads since losses from the default of government debt are partially shared.  The identification lemma can also be adapted for the opposite case $\phi_S>1$ by reversing the inequality specified in Assumption \ref{ass:general_parameterspace} to $\widetilde{\phi}_S > 1$. 

Mimicking Assumption 1 of \cite{Gabaix2024}, the Lemma is derived under the weaker condition that the correlation structure of shocks is known (see Condition (ii') in Appendix \ref{lemma:gmm_id_proof} for details). Theorems \ref{thm:gmm_consistency} and \ref{thm:gmm_asymptotic_normality} (below), which establish consistency and asymptotic normality of the RGIV estimator, are shown under the assumption of uncorrelated shocks but can analogously be adapted to the known shock correlation case. These results follow from a standard application of arguments provided in \cite{Pakes1989}. 

\begin{theorem}[Consistency of RGIV]
	\label{thm:gmm_consistency}
	Impose Assumption \ref{ass:general}. The RGIV estimator is consistent $\widehat{\boldsymbol{\phi}}^{RGIV} \xrightarrow{p} \boldsymbol{\phi}_0$ for the true parameter $\boldsymbol{\phi}_0$ as $T \to \infty$.
\end{theorem}
\begin{proof}
	See Appendix \ref{thm:gmm_consistency_proof}.
\end{proof}

\begin{theorem}[Asymptotic normality of RGIV]
	\label{thm:gmm_asymptotic_normality}
	Impose Assumption \ref{ass:general}. The RGIV estimator is asymptotically normal 
	\begin{align*}
		\sqrt{T}(\widehat{\boldsymbol{ \phi}}^{RGIV} - \boldsymbol{ \phi}_0) \xrightarrow{d} \nn(0, (G'WG)^{-1}G'W \Sigma W G (G'WG)^{-1} )
	\end{align*}
	for RGIV population weight matrix $W =  \mathrm{diag}(\frac{1}{\sigma_1^2 \sigma_2^2},\dots, \frac{1}{\sigma_1^2 \sigma_n^2}, \frac{1}{\sigma_2 ^2\sigma_3^2},\dots, \frac{1}{\sigma_{n-1}^2 \sigma_n^2})$, moment covariance matrix $\Sigma = \EE[g(\V{r}_t, \boldsymbol{ \phi}_0)g(\V{r}_t, \boldsymbol{ \phi}_0)']$ and $G = \EE[\nabla_{\boldsymbol{\phi}} g(\V{r}_t, \boldsymbol{ \phi}_0)]$ as $T\to\infty$.
\end{theorem}
\begin{proof}
	See Appendix \ref{thm:gmm_asymptotic_normality_proof}.
\end{proof}

\alert{Granularity ($S_i > 0$) and non-zero shock variances can loosely be viewed as the corresponding ``relevance conditions'' for the RGIV setting. These two conditions help guarantee that  $G'WG$ is full rank, ensuring that idiosyncratic shocks contribute to fluctuations in the endogenous variable $r_{St}$ and can be exploited as identifying variation.} 

\alert{The diagonal weight matrix permits the RGIV estimator to be interpreted as the $\boldsymbol{\phi}$ that minimizes the average squared correlation coefficient between shocks. In addition, $\widehat{W}(\widehat{\boldsymbol{ \phi}}^{RGIV})$ is an efficient sample weight matrix under shock independence\footnote{Under shock independence, the off-diagonal elements of $W$ are zero and the on-diagonal elements are multiplicatively separable i.e. that $\EE(u_{it}^2 u_{jt}^2) = \EE(u_{it}^2)\EE( u_{jt}^2)$ for $i \neq j$.} and helps ensure the stability of the GMM objective function when $T$ is relatively small. In practice, the procedure is implemented in MATLAB using the constrained nonlinear optimizer \texttt{fmincon()}. We recommend comparing objective function values from multiple starting points to help verify convergence to a global optimum.} \alert{Under the weaker condition of cross-sectional shock uncorrelatedness rather than shock independence, the RGIV estimator is consistent and asymptotically normal. However, the GMM weight matrix  $\widehat{W}(\widehat{\boldsymbol{ \phi}}^{RGIV})$ would no longer converge in probability to the efficient one since higher order dependence (like that arising from a shared volatility term) is permissible. Here, $\widehat{\boldsymbol{\phi}}^{RGIV}$ can still be used as a plug-in estimator for the efficient GMM weight matrix.\footnote{\label{footnote:higher_order_efficient}The following gives an example of one such procedure. First, RGIV can be used to obtain a preliminary (consistent) estimate of the spillover coefficient. Second, GMM can be used to estimate $\widehat{\boldsymbol{ \phi}}^\text{Step 2}$ with sample GMM weight matrix $\widehat{W}^\text{Step 2} =  [\frac{1}{T} \sum_{t=1}^T g(\V{r}_t, \widehat{\boldsymbol{ \phi}}^{RGIV})g(\V{r}_t, \widehat{\boldsymbol{ \phi}}^{RGIV})']^{-1}$.}\footnote{Serial correlation is another form of dependence of interest for applied work. Here, Theorem \ref{thm:gmm_asymptotic_normality} can be adapted to serially correlated shocks by exchanging the currently used central limit theorem for independent and identically distributed data for one applicable to data with dependence, like those found in \cite{davidson1994stochastic}.}}

\subsection{GIV under spillover coefficient heterogeneity}

In general, the baseline \cite{Gabaix2024} GIV estimand $\phi^{GIV}$ under spillover coefficient heterogeneity cannot be interpreted as a positive-weighted average of unit-specific spillover coefficients. From Equation  \ref{equation:GK_contamination}, the spillover coefficient homogeneity assumption is necessary for the instrumental variables exclusion restriction to hold, as it ensures the  endogenous term $r_{St}$ is differenced away. When shock variances are homogeneous across countries, Proposition \ref{prop:GIV_het_elasticities} (below) decomposes  the GIV estimand into an equal-weighted spillover coefficient term and a term that depends on $(\phi_S - \phi_E)$. The bias is larger when the spillover coefficient varies systematically with the size distribution, giving rise to a larger gap between $\phi_S$ and $\phi_E$. Moreover, as also discussed in Appendix D.8 of \cite{Gabaix2024}, the bias is smaller when the number of units is large.

\begin{proposition}
	\label{prop:GIV_het_elasticities}
	Impose Assumption \ref{ass:general} and  shock variance homogeneity ($\sigma^2_i = \sigma^2$). Then the GIV estimand is $\phi^{GIV}= \frac{\EE(z_t r_{Et})}{\EE(z_t r_{St})} = \phi_E + \frac{\phi_S - \phi_E}{n} \cdot \frac{1}{\frac{\phi_S - \phi_E}{1- \phi_S}[\sum_{i=1}^n S_i^2] -\frac{1}{n} + \sum_{i=1}^n S_i^2}.$

\end{proposition}
\begin{proof}
	See Appendix \ref{prop:GIV_het_elasticities_proof}.
\end{proof}

Proposition \ref{prop:GIV_het_elasticities}  implies that the GIV estimand doesn't admit a weighted average interpretation of unit-specific spillover coefficients. To see this, consider the following example. Suppose $n=3, \, \V{S} = \begin{bmatrix}
	0.2 & 0.3 & 0.5
\end{bmatrix}',$ and $\boldsymbol{ \phi} = \begin{bmatrix}
0.6 & 0.3 & 0.3 
\end{bmatrix}'$. Applying Proposition \ref{prop:GIV_het_elasticities},  $ \phi^{GIV} =-0.18 \not\in [0.3,0.6]$, so $\phi^{GIV}$ is not a positive-weighted average of individual spillover coefficients. Section \ref{sec:simulation} investigates the practical implications of GIV under spillover coefficient heterogeneity using an empirically relevant DGP.

While Proposition \ref{prop:GIV_het_elasticities} highlights the potential pitfalls from mistakenly applying the baseline GIV estimator, \cite{Gabaix2024} also gives guidance for estimating unit-specific spillover coefficients under additional restrictions. Their Proposition 6 and Appendix D.9 give alternative procedures that require heterogeneity to depend on observables and for the shock variance to be known, respectively. These restrictions would be unnecessary for RGIV.

\section{Testing}
\label{sec:testing}
In this section, I propose two tests derived from standard GMM results \citep{Newey1994}: a test of over-identifying restrictions that evaluates the uncorrelatedness of idiosyncratic shocks and a test that evaluates the homogeneous spillover coefficient condition of \cite{Gabaix2024}. For what follows, impose Assumption \ref{ass:general}.

\subsection{RGIV specification test}

The uncorrelatedness of  idiosyncratic shocks is testable when there are four or more units, as the number of moment conditions exceeds the number of unit-specific spillover coefficients. Recall that the RGIV estimator is a GMM estimator for moment function $g(\V{r}_t, \boldsymbol{ \phi})$, which encodes the pairwise uncorrelatedness of  idiosyncratic shocks. For $n\geq 4$ countries, the number of moments exceeds the number of estimated parameters allowing for the use of the Sargan--Hansen test. The null hypothesis $H_0{:}\, \, \EE[g(\V{r}_t, \boldsymbol{ \phi}_0)] = \V{0}$ for true parameter $\boldsymbol{\phi}_0$ is rejected for large values of the $J$-statistic $J_T =T\cdot  \widehat{Q}_T(\widehat{\boldsymbol{ \phi}}^{RGIV})$. Intuitively, $J_T$ is large when one or more pairs of estimated idiosyncratic shocks are correlated. \alert{The Sargan--Hansen test is subject to the usual caveat: failure to reject should not be interpreted as proof of correct specification.}

\subsection{Spillover coefficient homogeneity test}

Spillover coefficient homogeneity across units is not only a subject of potential substantive interest---particularly given its role in the \cite{Gabaix2024} GIV estimator---but also can be formally evaluated within the framework of RGIV. Since the null hypothesis of coefficient homogeneity $H_0{:}\,\, \phi_1=\phi_2=...=\phi_n$ is a special case of RGIV's unit-specific spillover coefficients, $H_0$ can be tested with the distance metric test. Here, estimator $\overline{\boldsymbol{\phi}}$ minimizes  $\widehat{Q}_T(\boldsymbol{\phi})$ subject to the constraints of null hypothesis $H_0$. Then, spillover coefficient homogeneity is rejected when the distance metric test statistic $DM_T = T(\widehat{Q}_T(\overline{\boldsymbol{\phi}}) - \widehat{Q}_T(\widehat{\boldsymbol{\phi}}^{RGIV}) )$ is large since $DM_T \xrightarrow{d} \chi^2_{n-1}$ under the null hypothesis. A large value of $DM_T$ indicates that the constraints of null hypothesis $H_0$ bind.

\section{Extensions to additional explanatory variables}
\label{sec:extensions}

Motivated by the requirements of empirical applications, this section discusses two extensions that relax the condition of uncorrelated idiosyncratic shocks discussed in Section \ref{sec:general} through observed and unobserved explanatory variables. I show that when observable time-varying explanatory variables determine the shocks' correlation structure, RGIV can be used after residualizing the outcome variables with respect to these observables. When the correlation structure is instead determined by unobserved factors with unknown loadings, the global identification condition fails.

\subsection{RGIV with observed explanatory variables}
\label{sec:extensions_controls}

In practice, observable characteristics can determine the correlation structure among shocks as described by the following two situations. First, the outcome variable could have unit-specific exposures to a particular observed variable---take country-specific exposures to the USD-EUR exchange rate in the Euro area sovereign yields example. Second, observable characteristics could also be used to account for a correlation structure driven by unobserved explanatory variables---like country-specific exposures to a ``global financial conditions'' factor---so long as such unobserved factors are in the span of the observed explanatory variables.

Observed variable $\V{x}_{it}$ ($k \times 1$) affects outcome variable $r_{it}$ through a direct effect and an indirect effect. Modifying Assumption \ref{ass:general_reducedform}, unit-specific coefficients $\boldsymbol{\beta}_i$ determine the cross-sectional correlation of $v_{it}$ 
\begin{align}
	r_{it} &= \phi_i r_{St} + \underbrace{\boldsymbol{\beta}_i' \V{x}_{it} + u_{it}}_{v_{it}}, \quad \V{x}_{it} \indep u_{it}, \quad i = 1,\dots,n
	\label{eq:RGIV_controls_model}
\end{align}
where $u_{it}$ is still idiosyncratic in the sense that $u_{it}$ is independent of $u_{jt}$ for $i \neq j$. The orthogonality condition $u_{it} \indep \V{x}_{it}$ can be interpreted as a ``selection-on-observables'' assumption; the cross-sectional correlation of $v_{it}$ is entirely determined by the observed explanatory variables. Holding $r_{St}$ and $u_{it}$ constant, $\boldsymbol{{\beta}}_i$ can be interpreted as the direct effect of $\V{x}_{it}$ on $r_{it}$. Additionally, mediated through changes in $r_{St} = \frac{1}{1-\phi_S}[\sum_{i=1}^n S_i \boldsymbol{\beta}'_i \V{x}_{it} + u_{St}]$,  there also exists an indirect effect of $\V{x}_{it}$ on $r_{it}$. \alert{The assumption of unit-specific observables $\V{x}_{it}$ nests two notable cases. The first is of a common observed factor ($\V{x}_{it} = \V{x}_t$) where the cross-sectional correlation in $v_{it}$ is induced by the unit-specific loadings $\boldsymbol{\beta}_i$. The second is of an observed group-specific factor ($\V{x}_{it} = \V{x}_{g(i)t}$) where $v_{it}$ is correlated only within groups $g(i)$.}

The spillover coefficients in Equation \ref{eq:RGIV_controls_model} can be estimated by using RGIV after residualizing the outcome variable $r_{it}$ with respect to observed variables $\V{x}_{it}$. Residualizing purges $r_{it}$ of the variation induced by the direct and indirect effects of $\V{x}_{it}$ on $r_{it}$. Concretely, the procedure has two steps: (1) Compute the residual $\dot{r}_{it}$ of the regression of $r_{it}$ on $\V{x}_{it}$, (2) treating $\dot{r}_{it}$ as data, estimate $\boldsymbol{\phi}$ using the RGIV estimator described in Definition \ref{defn:naive}. Conveniently, estimation uncertainty of the first step's regression coefficients has no effect on the asymptotic variance of the RGIV estimator in the second step. Thus, treating $\dot{r}_{it}$ as data in the second step produces valid standard errors for $\widehat{\boldsymbol{\phi}}^{RGIV}$. See Appendix \ref{sec:RGIV_controls} for formal results.

\subsection{Non-identification when explanatory variables are unobserved}
\label{sec:extensions_factor}

In this section, I describe the tradeoff in assumptions between allowing spillover coefficient heterogeneity/unknown shock variances and a factor structure for the shocks. I show that global identification of unit-specific spillover coefficients is lost when a single unobserved factor is included in the error term. Such a case is empirically relevant because there is typically no a priori reason to believe that spillover coefficients are homogeneous across units and because estimated latent factors are commonly used as control variables in the applied GIV literature.  

 GIV regressions with estimated latent factors as control variables are ubiquitous in the applied GIV literature (\cite{Flynn2022, Gabaix2024,Gabaix2023, Camanho2022,Baumeister2023a,Adrian2022} among others). Typically, latent factors are estimated using principal components on the demeaned outcome variable before being included as control variables in the GIV regression. Such an approach is attractive because it enables practitioners to apply GIV to applications where the correlations between unit shocks are driven by a small number of latent factors. These estimated factors, however, are subject to measurement error, so their inclusion as control variables in subsequent GIV regressions gives rise to attenuation bias.  

\cite{Banafti2022} addresses this concern by extending GIV with homogeneous spillover coefficients to a large time and panel dimension framework. When the size distribution of units is very skewed (more skewed than Zipf's law), the sampling uncertainty arising from latent factors and loadings is negligible under their procedure. Homogeneity of spillover coefficients across units is crucial for the procedure's validity. When spillover coefficients are heterogeneous, cross-sectionally demeaning each unit no longer differences away the (no longer constant) contribution of spillovers. As a result, PCA estimates of the latent factors are polluted by the presence of spillovers.

Given the empirical relevance of unit-level heterogeneity and latent factors, I consider a heterogeneous spillover coefficient latent factor model. Taking $n$ to be fixed, restrictions on the skewness of the unit size distribution are unneeded. Investigating identification, I augment Assumption \ref{ass:general_reducedform} to include a single latent factor $f_t$ with unknown unit-specific loading $\lambda_i$ 
\begin{align}
	r_{it} = \phi_i r_{St} + \lambda_i f_t + u_{it},\quad \lambda_i \neq 0, \quad f_t \indep u_{it}, \quad n \geq 5
	\label{eq:singlefactor}
\end{align}
where the number of units $n$ is fixed and time $T\to\infty$. Factor $f_t$ is normalized so that $\EE(f_t) = 0$, $\EE(f_t^2) =1$, and $\lambda_1 > 0$. Shocks $u_{it}$ are idiosyncratic in that they are independent of latent factor $f_t$ and $u_{it} \indep u_{jt}$ for $i \neq j$. Then, extending the logic of the RGIV estimator to the single factor case, the moment function between units $i \neq j$ is 
\begin{align*}
	g_{ij}^\text{factor}(\V{r}_t, \boldsymbol{\theta}) = (r_{it} - \phi_i r_{St})(r_{jt} - \phi_j r_{St}) - \lambda_i \lambda_j
\end{align*}
where  $\boldsymbol{\theta} = [\boldsymbol{ \phi}', \boldsymbol{\lambda}']'$ for $\boldsymbol{\lambda} = [\lambda_1,\dots,\lambda_n]'$. $g_{ij}^\text{factor}(\V{r}_t, \boldsymbol{\theta})$ can then be stored in moment vector $g^\text{factor}(\V{r}_t, \boldsymbol{\theta})$. When $n \geq 5$, the number of moments (a total of $n(n-1)/2$) is greater than or equal to the number of parameters to be estimated (a total of $2n$).

Lemma \ref{lemma:notid} (below) however shows that the population moment condition $g_0^\text{factor}(\boldsymbol{\theta}) = \EE[g^{\mathrm{factor}}(\V{r}, \boldsymbol{\theta})]$ has multiple roots, so the model described in Equation \ref{eq:singlefactor} is not identified.

\begin{lemma}
	Consider Equation \ref{eq:singlefactor} and let $\boldsymbol{\theta}_0$ be the true parameter. There exists $\widetilde{\boldsymbol{\theta}} \neq \boldsymbol{\theta}_0$ such that $g_0^\mathrm{factor}(\widetilde{\boldsymbol{\theta}}) = 0$.
	\label{lemma:notid}
\end{lemma}
\begin{proof}
	See Appendix \ref{lemma:notid_proof}.
\end{proof}

The failure in the global identification condition comes from the assumptions of unknown shock variances and unknown factor loadings. If instead shock variances and factor loadings were taken to be known---mirroring Assumption 1 of \cite{Gabaix2024} and as is the case for the proof of Lemma \ref{lemma:gmm_id}---then the unit-specific spillover coefficients are identified.\footnote{In a related case, Appendix \ref{sec:factor_known_observables} shows that differencing RGIV moments can account for factor loadings determined by unit-specific observables.} In the proof of Lemma \ref{lemma:notid}, I show that factor loadings can compensate for incorrect guesses for the spillover coefficient. To see this, let $\widetilde{\boldsymbol{\theta}} = [\widetilde{\boldsymbol{\phi}}', \, \widetilde{\boldsymbol{\lambda}}']'$ be a candidate root to the population moment condition. In the proof, I consider the set of solutions where the first unit's spillover coefficient and loading equal their true values ($\widetilde{\phi}_1 = \phi_1$ and $\widetilde{\lambda}_1 = \lambda_1$). I then show that a subset of the population moment conditions implies that $ \widetilde{\phi}_k-	\phi_k   =- \frac{\lambda_1(1-\phi_S)}{  \lambda_1 \lambda_S + S_1 \sigma^2_1 }(\widetilde{\lambda}_k - \lambda_k   )$ for $k \geq 2$. In words, $\widetilde{\lambda}_k$ can compensate for an incorrect spillover coefficient $\widetilde{\phi}_k \neq \phi_k$. Formally, the proof shows that there is at least one $\widetilde{\boldsymbol{\theta}}\neq \boldsymbol{\theta}_0$ such that $g_0^\textrm{factor}(\widetilde{\boldsymbol\theta})=0$. Moreover,  $\widetilde{\boldsymbol{\theta}}$ need not be ``close'' to the true parameter, as there is no guarantee that $\max_i \widetilde{\phi}_i \geq \min_i \phi_i$ or $\min_i \widetilde{\phi}_i \leq \max_i \phi_i$. In this sense, $\widetilde{\boldsymbol{\phi}}$ is potentially far from the true spillover coefficient $\boldsymbol{\phi}_0$.

Lemma \ref{lemma:notid} also implies a tradeoff between modeling unit-level heterogeneity and allowing for correlated shocks. Recall that Proposition 7 of \cite{Gabaix2024} shows that a \textit{homogeneous} spillover coefficient is identified when shocks admit a factor structure with unknown loadings and if shock variances are homogeneous across units. In contrast, Lemma \ref{lemma:notid} shows that global identification is lost under unrestricted heterogeneity on spillover coefficients and shock variances. Taken together, these results suggest that practitioners face a tradeoff between two empirically-relevant models.

\alert{At the same time, if the true model contained a single latent factor, but the researcher mistakenly applied the (over-identified) baseline RGIV moment condition of Definition \ref{defn:naive}, the population moment condition wouldn't generically hold:}

\begin{lemma}
	\alert{Consider Equation \ref{eq:singlefactor} with $n \geq 4$. Suppose there is no constant $c \in \RR$ such that $S_i\sigma^2_i/\lambda_i = c$ for all but at most one $i \in \{1,\ldots,n\}$. Then, there exists no $\widetilde{\boldsymbol{\phi}} \in \Phi$ such that $g_0(\widetilde{\boldsymbol{\phi}}) = 0.$}
	\label{lemma:no_id_factor}
\end{lemma}

\begin{proof}
	\alert{See Appendix \ref{app:no_id_factor}.}
\end{proof}

\alert{Thus, with a large enough sample, the result  suggests that the $J$ test of Section~\ref{sec:testing} would generically detect misspecification arising from a missing latent factor. Note that the result includes a condition that rules out a knife-edge case on homogeneity across sizes, factor loadings, and shock variances. Since the proof reduces to a nonlinear system with more equations than unknowns, the condition ensures no solution exists for such a system.}

\section{Applications}
\label{sec:application}

I illustrate the usefulness of RGIV in two settings. In Section \ref{subsec:sovereign_yields}, I apply RGIV to estimate sovereign yield spillovers. In Section \ref{subsec:inelastic_markets}, I adapt RGIV to a demand system to study investor-level heterogeneity in the ``inelastic markets hypothesis.''

\subsection{Sovereign yield spillovers}
\label{subsec:sovereign_yields}
Applying the robust granular instrumental variables (RGIV) methodology to the Euro area sovereign yield spillovers application of a working paper version of \cite{Gabaix2024}, this section finds strong evidence of country-level heterogeneity in the spillovers of idiosyncratic shocks. 

Just as in \cite{Gabaix2024}, the sample consists of daily data on 10-year zero coupon yields from Bloomberg from September 1, 2009 to May 31, 2018 giving a total of 2283 observations. The included countries are Austria, Belgium, Finland, France, Germany, Greece, Ireland, Italy, Netherlands, Portugal, Slovenia, and Spain. For the yield spread of country $i$ (relative to Germany) $y_{it}$, the outcome variable $r_{it}$ is defined as $r_{it} = \frac{y_{it} - y_{it-1}}{0.01 + y_{i,t-1}}$ just as in \cite{Gabaix2024}. Since shocks are linearly related to the outcome variable as outlined in Assumption \ref{ass:general_reducedform}, $r_{it}$ is winsorized (over time) at the 0.5th and 99.5th percentiles. Following \cite{Gabaix2024}, size is time-varying and computed as ``debt-at-risk'' $S_{i,t-1} = \frac{B_{i,t-1} y_{i,t-1}}{\sum_j B_{j,t-1} y_{j,t-1}}$ where $B_{i,t-1}$ is the outstanding government debt of country $i$.

Observed explanatory variables are included to account for unit-specific exposures to latent aggregate shocks. In particular, I include the STOXX 50 Volatility Index (differences), the STOXX Europe 600 Index (growth), the EUR-USD exchange rate (growth), the United States 10 Year Treasury yield (growth), BBB/Baa-10Y spread (differences), and the European Fama-French 5 factors \citep{Fama2015}. These observed explanatory variables account for differential exposure of countries to uncertainty, exchange rates, equity prices, and risk. All data are downloaded from Bloomberg.

Even with observed explanatory variables, well-documented regional correlations among Europe's core and periphery countries likely drive correlations between shocks \citep{Bayoumi1992}. To address this concern, I size-aggregate $r_{it}$ to form larger country blocks; even if shocks \emph{within} blocks are correlated, RGIV is still valid so long as shocks \emph{between} blocks are uncorrelated. Hence, uncorrelatedness between blocks is a weaker condition than uncorrelatedness between countries. Specifically, I consider the following four blocks:
\begin{itemize}
	\item Block 1 (\textit{Core}):  Austria, Belgium, Finland, France, Netherlands;
	\item Block 2 (\textit{Western periphery}): Ireland, Portugal, and Spain;
	\item Block 3 (\textit{Eastern periphery}): Greece and Italy;
	\item Block 4: Slovenia.
\end{itemize} 
Block 1 includes countries typically classified as being members of the EU ``core.'' Blocks 2 and 3 contain countries typically classified as being members of the EU ``periphery.'' Block 4 contains Slovenia, which is typically left uncategorized on account of its distinct institutional structure as a former member of Yugoslavia. 

\alert{Blockwise uncorrelatedness also allows for this paper's conventional, fixed-dimensional moment condition framework to be applied. Without aggregation, the number of moments would grow with the cross-section. Then, exact uncorrelatedness of the idiosyncratic shocks would become increasingly restrictive (analogous to an exact factor structure in the approximate factor model literature e.g., \cite{Chamberlain1983}) and standard GMM inference would be subject to distortions (familiar from the many-moments literature e.g., \cite{Han2006}).}

I report RGIV point estimates and standard errors for individual spillover coefficients.\footnote{\alert{The estimator is initialized at $(0.5, 0.5, 0.5, 0.5)'$. Checking convergence, I estimate the preferred specification from 2,000 starting points drawn uniformly over $[0, 0.99]^4$. Of these, 90.4\% of the draws converge to a spillover coefficient vector that is within 0.001 (by Euclidean distance) of the reported estimate. The remaining all obtain objective function values that are greater than the reported optimum.}} Conservatively, standard errors are computed using a HAC GMM weight matrix to account for the possible serial correlation of idiosyncratic shocks.\footnote{Specifically, a Newey-West kernel with the \cite{Lazarus2018} truncation parameter rule of $1.3 \sqrt{T}$.}   Size-weighted spillover coefficients are constructed using the Delta method, using the average block size over the estimation sample.

I also present GIV results based on the shock variance approximation and factor estimation procedures detailed in Section 5.3 of the July 2021 working paper version of \cite{Gabaix2024}. I include this procedure for comparison, as it is used in the original application. This procedure approximates the variance of the shocks with $\Var(r_{it})$, valid when spillovers are small. As a diagnostic for GIV instrument strength, the first stage $F$-statistic is also reported.

I find strong evidence of spillovers in the aggregate and of spillover heterogeneity across countries. In the preferred specification (Column 1 of Table \ref{table:application_main}), \alert{we fail to reject the null hypothesis of correct specification} at conventional significance levels. I find evidence of spillovers in the aggregate. The size-weighted spillover coefficient is 0.54 (with a standard error of 0.08). Moreover, with a $p$ value of $<0.001$ for the spillover coefficient homogeneity test,  the null hypothesis of spillover coefficient homogeneity is rejected at conventional significance levels. 

\alert{The core countries exhibit the lowest sensitivity to aggregate yields (0.39), consistent with their perceived soundness over the sample. The block consisting of Greece and Italy exhibits a moderate spillover coefficient (0.44). In contrast, the block consisting of Ireland, Portugal, and Spain exhibits the highest size-weighted spillover coefficient (0.82). One possible explanation is that for much of the sample, Greece’s passthrough was buffered by its involvement in international bailout programs despite its severe distress.  Ireland, Portugal, and Spain, however, spent substantial portions of the sample either in pre-program market distress or in post-program recovery.}

\alert{The RGIV model also provides a lens through which to view the historical narrative analogous to the ``narrative'' identification approach of a working paper version of \citet{Gabaix2024}. Recall that the contribution of an idiosyncratic shock to the size-weighted average yield is given by $S_i u_{it}/(1-\phi_S)$. Table \ref{table:top10_narrative} lists the dates that correspond to the ten largest (by absolute value) idiosyncratic contributions to movements in the size-weighted average yield. The table also lists the relevant events that occurred around each date as reported by Bloomberg News. All ten dates can be matched to events that correspond to shocks to the Italy/Greece periphery block (five each). For example, the largest contribution occurred on July 6, 2015. July 6 corresponds to the first trading day after the July 5 Greek bailout referendum where 61\% of Greek voters voted ``No'' on the terms of the bailout. The contribution's positive sign (+0.068) is consistent with the historical narrative of an elevated risk of Greece's exit from the Euro area.}

The qualitative features of the preferred specification are robust to omitting observed explanatory variables and using an estimator that is efficient under higher-order dependence among shocks. Omitting explanatory variables, Column 2 of Table \ref{table:application_main}  shows spillover coefficient heterogeneity across blocks. Relative to the preferred specification, the estimated size-weighted spillover coefficient is slightly larger (at 0.63). Recall that the RGIV estimator is efficient under the independence of idiosyncratic shocks, but isn't guaranteed to be efficient under the weaker condition of idiosyncratic shock uncorrelatedness.\footnote{For example, a shared volatility term would induce dependence of higher order moments.} To address this concern, Column 3 shows the results of a ``higher-order efficient'' estimator, which uses the procedure outlined in Footnote \ref{footnote:higher_order_efficient} of Section \ref{sec:general_RGIVestimator}. The point estimates and standard errors are nearly identical, suggesting that higher order dependence of idiosyncratic shocks plays a minor role in this application. The spillover coefficient estimated with the GIV procedure featured in a working paper version of \cite{Gabaix2024} is comparable to the size-weighted spillover coefficient computed with RGIV, but is unable to speak to country-level heterogeneity. Column 4 of Table \ref{table:application_main} presents results of the baseline GIV methodology applied to this section's core-periphery-aggregated panel (distinct from the country-level panel of \cite{Gabaix2024}). The spillover coefficient of 0.56 is close to the size-aggregated spillover coefficient computed in the preferred RGIV specification.

\begin{table}[tbp]
	\centering
	\small
	\begin{tabular}{l|cccc}
		\toprule
		& Preferred & No controls & Higher-order efficient & 1-factor GIV \\
		\hline
		RGIV results: \\
		\quad$\phi_S$ & 0.54 & 0.63 & 0.54 & \\
		& (0.08) & (0.07) & (0.08) & \\
		\quad$\phi_\text{IRL, PRT, ESP}$ & 0.82 & 0.86 & 0.83 & \\
		& (0.06) & (0.05) & (0.06) & \\
		\quad$\phi_\text{GRC, ITA}$ & 0.44 & 0.56 & 0.44 & \\
		& (0.16) & (0.13) & (0.16) & \\
		\quad$\phi_\text{Core}$ & 0.39 & 0.47 & 0.40 & \\
		& (0.03) & (0.02) & (0.03) & \\
		\quad $\phi_\text{SVN}$ & 0.33 & 0.49 & 0.33 & \\
		& (0.06) & (0.05) & (0.05) & \\
		$\phi^{GIV}$ & & & & 0.56 \\
		& & & & (0.02) \\
		\hline
		Tests ($p$-values): \\
		\quad Specification & 0.951 & 0.923 & 0.937 & \\
		\quad Homogeneity & <0.001 & <0.001 & <0.001 & \\
		First stage $F$-stat.\! & & & & 592 \\
		\bottomrule
	\end{tabular}
	\caption{Spillover coefficient estimation results. RGIV coefficient estimates for $\phi_S$, $\phi_\text{IRL, PRT, ESP}$, $\phi_\text{GRC, ITA}$, $\phi_\text{Core}$, and $\phi_\text{SVN}$ are listed above standard errors, which are provided in parentheses. $p$ values are provided in the bottom section of the table for the specification test and parameter homogeneity tests. Estimates for the GIV spillover coefficient and standard error are listed in the $\phi^{GIV}$ row, and the instrumental variables first stage $F$-statistic is given in the table's last row. See the main text for details on the table's columns.}
	\label{table:application_main}
\end{table}

Section \ref{sec:application_robustness} gives additional RGIV results under alternative winsorizations, the more widely-used \cite{Andrews1991} truncation parameter, omitting the Fama-French factors as observed explanatory variables, and alternative choices of country blocking. \alert{In particular, the section considers a specification where each country represents its own block---here, the $J$ test's null hypothesis of correct specification is rejected at conventional significance levels, evidence of the correlation of idiosyncratic shocks.} Section \ref{sec:application_robustness} also contains results for the 0-factor and 2-factor GIV specifications.

Summarizing, RGIV finds strong evidence of country-level differences in sovereign yield spillovers. In response to a 1\% increase in the size-weighted relative yield spread, the relative yield spread of ``core'' countries increases by 0.4\% compared to an increase of 0.8\% for countries in the western ``periphery.'' Substantively, these estimates point to the importance of understanding the role of country-level characteristics in the heterogeneous propagation of idiosyncratic shocks during sovereign debt crises.

\subsection{Inelastic markets hypothesis}
\label{subsec:inelastic_markets}

		\alert{In this section, I use RGIV to investigate the inelastic markets hypothesis of \cite{Gabaix2023}. They find that when the price of the equity market portfolio increases by 5\%, the quantity demanded falls by 1\% (equivalently an elasticity of $-0.2$). This finding contrasts with the traditional view of investors as rational agents, which would predict an elasticity that would be substantially larger in magnitude  (e.g., an elasticity of $-20$ in Figure 3 of \citet{Gabaix2023}).} \alert{Using RGIV, I also find evidence supporting the inelastic markets hypothesis with a point estimate of $-0.05$, with evidence of investor-level heterogeneity.  Econometrically, the application demonstrates how the theory from Section \ref{sec:general} can be adapted to a specialized application.}

		\alert{I replicate the data used in \cite{Gabaix2023} for my analyses, which uses sector-level data from the Flow of Funds on equity holdings---both levels and flows. See Appendix C of \cite{Gabaix2023} for details on the construction of the dataset. The dataset used for the analyses includes 12 sectors described in Table \ref{table:ffunds_block_def} from 1993:Q1 to 2018:Q4. For sector $i$ and time $t$, the growth in investors' equity holdings is $\Delta q_{it}$ and the fraction of the market held is $S_{it-1}$. My measure of prices is the growth in the value-weighted return ex-dividends from the Center for Research in Security Prices (CRSP). Quarterly returns are noted as $\Delta p_t$.} 

		\alert{For demand elasticities $\phi_1, \dots, \phi_n$ and supply elasticity $\gamma$, I adapt the RGIV framework for spillovers to the following demand system
		\begin{align}
			\Delta q_{it} &=  \phi_i \Delta p_t + \boldsymbol{\beta}_i' \mathbf{x}_{t} + u_{it} 			\label{eq:demand_system_flowfunds}\\
			\Delta q_{St} &= \gamma \Delta p_t + \boldsymbol{\beta}_0' \mathbf{x}_{t} + \varepsilon_{t} 			\label{eq:supply_system_flowfunds}
		\end{align}
		 where $\mathbf{x}_t$ is a vector of observed explanatory variables and $(\varepsilon_{t}, u_{1t}, u_{2t}, \dots, u_{nt})'$ are \! cross-sectionally uncorrelated. Equation \ref{eq:demand_system_flowfunds} describes the demand for each sector while Equation \ref{eq:supply_system_flowfunds} describes the supply. Relative to \citet{Gabaix2023}, demand elasticities $\phi_i$ are permitted to be unit-specific and $\gamma$ is estimated. Analogous to the restriction of the size weighted spillover coefficient described in Section \ref{sec:general}, the size-weighted demand elasticity is assumed to be less than the supply elasticity ($\phi_S < \gamma$).}
		
		\alert{Endogeneity is induced by each unit's contribution to movements in the price. Setting the size-weighted sum of Equation \ref{eq:demand_system_flowfunds} equal to Equation \ref{eq:supply_system_flowfunds} implies that the price is a function of observables, the size-weighted idiosyncratic demand shocks, and the supply shock}
\begin{align}
	\label{eq:price_demand_controls}
	\Delta p_t = \frac{1}{\phi_S-\gamma}(\varepsilon_t - u_{St} + (\boldsymbol{\beta}_0 - \boldsymbol{\beta}_S)' \mathbf{x}_t ).
\end{align}		
	\alert{Thus, a regression of $\Delta q_{it}$ on $\Delta p_t$ would give rise to endogeneity bias in the estimation of $\phi_i$. Here, the RGIV estimator exploits the uncorrelatedness of the supply and demand shocks $(\varepsilon_{t}, u_{1t}, u_{2t}, \dots, u_{nt})'$ giving $n(n+1)/2$ moment conditions. See Appendix \ref{sec:demand_framework} for the full environment and identification argument.}
		
	\alert{I use a wide set of observed explanatory variables to control for investor-specific exposures to aggregate shocks. Following \cite{Gabaix2023}, I include GDP growth, log market cap, log book-to-market ratio, and  momentum. The latter three variables are constructed using the data of \citet{Jensen2023}. In addition, I include the 2-year Treasury, the 5-year Treasury, the 10-year Treasury, the 30-year Treasury, TED spread, University of Michigan Consumer Sentiment Index, the Major Currencies Dollar Index (Goods only),  the credit spread of \citet{Gilchrist2012}, and the WTI spot crude oil price. These observed variables help capture differential exposure to monetary policy, credit conditions, and sentiment.}

\alert{From cross-sector operational similarity, I consider the weaker assumption of blockwise uncorrelatedness. Like the application to sovereign yield spillovers, the assumption permits the correlatedness of shocks within blocks but requires that shocks are uncorrelated across blocks. Specifically, I consider four blocks: Retirement \& Insurance (Block I), Intermediaries and Funds (Block II), Foreign (Block III), and Real Domestic (Block IV). See Table \ref{table:ffunds_block_def} for details. The weaker assumption of blockwise uncorrelatedness permits for the correlation of shocks not fully accounted for by the observed explanatory variables. Examples include the reaction to catastrophic events or changes in pension regulation in Block I, client pass-through for Block II, and exposure to the real economy in Block IV.}

\begin{table}[t]
	\centering
\begin{tabular}{c|c|c|l|c}
    \hline
    Block & Label & Size & \multicolumn{1}{c|}{Sector} & Relative size \\
    \hline
    \multirow{5}{*}{I} & \multirow{5}{*}{\makecell{Retirement \\ \& Insurance}} & \multirow{5}{*}{0.21} & Private pension funds                   & 0.46 \\
                       &                                                         &                     & Federal government retirement funds     & 0.02 \\
                       &                                                         &                     & State and local govt. retirement funds  & 0.37 \\
                       &                                                         &                     & Life insurance companies                & 0.09 \\
                       &                                                         &                     & Property-casualty insurance companies   & 0.06 \\
    \hline
    \multirow{3}{*}{II} & \multirow{3}{*}{\makecell{Intermediaries \\ \& Funds}} & \multirow{3}{*}{0.25} & ETFs and mutual funds                  & 0.96 \\
                        &                                                         &                     & Closed-end funds                       & 0.02 \\
                        &                                                         &                     & Brokers and dealers                    & 0.02 \\
    \hline
    III & Foreign & 0.11 & Rest of the world & 1.00 \\
    \hline
    \multirow{3}{*}{IV} & \multirow{3}{*}{Real Domestic} & \multirow{3}{*}{0.42} & Household sector                        & 0.98 \\
                        &                                &                     & U.S.-chartered depository institutions  & 0.01 \\
                        &                                &                     & State and local governments             & 0.01 \\
    \hline
\end{tabular}
\caption{Summary of blocks for the inelastic markets hypothesis. The ``Size'' column refers to the average block size over the sample of interest (1993:Q1 to 2018:Q4). The ``Relative size'' column refers to the average size of each sector relative to the total block size over the sample of interest (1993:Q1 to 2018:Q4).}
\label{table:ffunds_block_def}
\end{table}

\alert{I consider the two-step GMM estimator described in Footnote \ref{footnote:higher_order_efficient} as a baseline for efficiency under higher order dependence of idiosyncratic shocks. Like the previous application, standard errors are conservatively computed using a HAC GMM weight matrix to account for potential serial correlation of the idiosyncratic shocks. Size-weighted elasticities are again constructed using the Delta method using the average block size over the estimation sample.}

\alert{Under RGIV, I estimate an inelastic size-weighted price elasticity of demand. From the first column of Table \ref{table:ffunds_results}, the size-weighted demand elasticity for the preferred specification is $-0.05$. For comparison, the demand elasticity computed using the 1-factor GIV estimator under the same assumption of blockwise uncorrelatedness\footnote{As in \citet{Gabaix2023}, the specification includes the following control variables: GDP growth, log market cap, the log book-to-market ratio, and a constant. The point estimate differs from the original reference (of approximately $-0.20$) due in part to the weaker assumption of blockwise uncorrelatedness.} is broadly comparable, with a point estimate of $-0.15$ and standard error of 0.08 (final column of Table \ref{table:ffunds_results}). At the same time, I find evidence of investor-level heterogeneity using RGIV; the null hypothesis of elasticity homogeneity is rejected at conventional significance levels for the preferred specification.}

\alert{Investigating heterogeneity, the RGIV results rule out large elasticities for each of the four considered blocks. Table \ref{table:ffunds_results} shows that the elasticity for the Retirement and Insurance block is $-0.20$. The category contains a mixture of investors with fixed glide path defined contribution plans (with elasticities near zero) and plans with fixed-share pension mandates (with more moderate elasticities). Next, the elasticity for the Intermediaries and Funds block is 0.09. The positive elasticity is consistent with the pro-cyclicality of the flow-performance relationship. The elasticity of the Foreign block is nearly zero, consistent with offsetting elasticities of heterogeneous investors (like foreign wealth funds, central banks, and retail investors) and would benefit from further study with more granular data. The elasticity for the Real Domestic block is $-0.08$, which is dominated by the Household sector (See Table \ref{table:ffunds_block_def}). Finally, the supply elasticity is computed to be 0.20. To interpret, as in \citet{Gabaix2023}, recall that five sectors are dropped from the demand system due to poor data quality.\footnote{These sectors include Nonfinancial Corporate Business, Foreign Banking Offices, Federal Government, Monetary Authority, and Funding Corporations.} Thus, the parameter should be viewed as a function of the omitted-sector residual supply and total market supply.}  

\input{out/ffunds_results.tex}

\alert{The qualitative findings of the preferred specification are robust to the alternative 
specifications featured in Columns 2--5 of Table \ref{table:ffunds_results}. Columns 2--4 reassign sectors across blocks to help address threats to the assumption of blockwise uncorrelatedness. Specifically, Column 2 moves Brokers \& dealers to the Real Domestic block (addressing the concern that its shocks primarily reflect funding liquidity rather than passthrough), Column 3 moves State \& local governments to the Retirement and Insurance block (addressing the concern that it behaves as an institutional investor), and Column 4 moves the Closed-end funds to the Real Domestic block (addressing the concern that closed-end funds are in large part held by retail investors). Column 5 displays results for the CUE estimator. The size-weighted aggregate elasticity and point estimates are qualitatively similar to the preferred specification across the considered alternative specifications. The mean squared off-diagonal correlation among the moment conditions for the preferred GMM estimator is 0.23, consistent with the presence of higher-order dependence of the idiosyncratic shocks. While RGIV inference under CUE would still be valid under higher-order dependence, the results on hypothesis testing should be interpreted with caution as the diagonal CUE weight matrix would no longer be efficient.}

\alert{Summarizing, RGIV finds evidence supporting the inelastic markets hypothesis of \citet{Gabaix2023} in the aggregate and evidence of investor-level heterogeneity.}

\section{Simulation study}

\label{sec:simulation}

I close by showing that RGIV has good finite sample performance in simulation using six DGPs based on the preferred specification studied in Section \ref{sec:application}. In contrast, I show that two versions of the GIV procedures---the ``feasible'' GIV procedure (featured in this paper's application) and ``oracle'' GIV procedure that takes idiosyncratic shock variances $\sigma^2_i$ to be known---can severely under-cover under spillover coefficient heterogeneity. \\

\noindent\textsc{Setup.}{ }{ }  The ``homogeneous spillovers'' DGP is loosely based on the preferred specification of Section \ref{sec:application}. Under spillover coefficient homogeneity across units, the DGP serves as a baseline, as both the oracle GIV and RGIV estimators are correctly specified. For $n=4$ units and a sample length of $T = 2283$, the model parameters are below:
\begin{align*}
	\boldsymbol{\phi}&= 0.54 \cdot \mathbf{1}_{4,1}\text{, } \boldsymbol{\sigma} = 0.014 \cdot \mathbf{1}_{4,1}\text{, } 
	 \text{ and } \V{S} = \begin{bmatrix}
	0.29 & 0.56 & 0.14 & 0.01
	\end{bmatrix}'.
\end{align*}

Having established an environment  for which both the oracle GIV and RGIV estimators are valid, I study five additional DGPs that deviate from the homogeneous spillovers DGP. These illustrate the finite sample properties of RGIV and GIV. First, the ``coefficient outlier'' DGP takes the homogeneous spillovers specification and sets the spillover coefficient of the fourth unit to  $0.75$. Second, the ``shock variance outlier'' DGP takes the homogeneous spillovers specification and sets the shock standard deviation of the first unit to $0.03$. \alert{Third, the ``short $T$'' DGP takes the homogeneous spillovers specification and instead considers $T=100$ observations. Fourth, the ``near-homogeneous size'' specification takes the homogeneous spillovers specification and modifies the size vector to be $[0.250, 0.253, 0.249, 0.248]'$.} Fifth, the ``application'' DGP allows for heterogeneous spillover coefficients and shock variances by assigning them to be the estimated values from the application's preferred specification. 

For each DGP, I consider results from the RGIV estimator, ``feasible'' GIV estimator, and ``oracle'' estimator. For the RGIV estimator, I compute confidence intervals for the spillover coefficients of individual units, size-weighted spillover coefficients, and equal-weighted spillover coefficients.  These confidence intervals highlight how RGIV can be used for researchers interested in individual spillover coefficients and for those interested in a single, aggregated parameter.  The ``feasible'' GIV estimator is computed using the $\Var(r_{it}) \approx \Var(u_{it})$ approximation (valid when spillovers are small) while the ``oracle'' GIV estimator is computed as if the true idiosyncratic shock variance $\Var(u_{it})$ were known. Both GIV estimators are included to show the practical implications of mistakenly assuming spillover coefficient homogeneity across units. The feasible GIV estimator, in particular, is included to evaluate the approximation $\Var(r_{it}) \approx \Var(u_{it})$ and for completeness (it is featured in the application section).

A conservative notion of coverage is used for the GIV estimators. Since the GIV estimators require spillover coefficient homogeneity across units, I report the proportion of GIV confidence intervals that contain \textit{any} positive-weighted average of potentially heterogeneous spillover coefficients---measured as the proportion of GIV confidence intervals with a nonempty intersection with the closed interval $[\min_{1\leq i \leq n} \phi_i, \max_{1\leq i \leq n} \phi_i]$. \\

\noindent\textsc{Results.}{ }{ } While the empirical coverage is near the nominal level for RGIV \alert{for the DGPs with $T = 2283$}, feasible GIV can severely under-cover. Table \ref{table:simulation_main} shows that nearly 95\% of the confidence intervals for $\widehat{\phi}_S$ and $\widehat{\phi}_E$ contain the true estimands $\phi_S$ and $\phi_E$ for these DGPs. In contrast, the true spillover coefficient is contained in none of the feasible GIV confidence intervals in the homogeneous spillovers DGP. Considering that the true coefficient is contained in 95\% of the oracle estimator confidence intervals, together these results suggest that the approximation $\Var(r_{it}) \approx \Var(u_{it})$ is inappropriate for this DGP. 

Under spillover coefficient heterogeneity, the GIV estimators can substantially under-cover---even with the conservative notion of coverage featured in this simulation study. For the coefficient outlier DGP, none of the feasible GIV confidence intervals contain \textit{any} positive-weighted average of heterogeneous spillovers, compared to 15\% for the oracle estimator. Hence, for certain DGPs, mistakenly assuming spillover coefficient homogeneity and applying versions of the baseline GIV estimator can result in unreliable inference.

\input{out/main_T=2300.tex}

Evaluating the properties of the testing procedures featured in Section \ref{sec:testing}, the empirical size of the RGIV specification and coefficient homogeneity tests are near their nominal levels. \alert{Comparing DGPs where the null hypothesis of spillover coefficient homogeneity is true, the empirical rejection rate of the coefficient homogeneity test is, at worst, 6.1\% for the $T=100$ DGP.} When spillover coefficients are not homogeneous, as is the case under the elasticity outlier and application DGPs, the null hypothesis of homogeneous coefficients is rejected in more than 99\% of Monte Carlo replications. Turning to the specification test, note that the RGIV framework is correctly specified in all six DGPs. At worst, the null hypothesis of idiosyncratic shock uncorrelatedness is rejected in 6.4\% of Monte Carlo replications under the application DGP.

\alert{For the $T=2283$ DGPs,} the RGIV estimator presents a modest (if any) power tradeoff relative to the oracle GIV estimator, and unit-specific RGIV spillover coefficients have good coverage properties despite applying to a more general environment. In the homogeneous spillovers application, the equal-weighted RGIV spillover coefficient has a comparable median confidence interval length (0.038) to that of GIV (0.046).\footnote{In general, researchers interested in \textit{any} positive-weighted average of potentially heterogeneous spillover coefficients could estimate unit-specific spillover coefficients using RGIV and compute the weights that minimize the confidence interval length of the resulting aggregated estimator.} Turning to unit-specific spillover coefficients, Table \ref{table:simulation_spillovercoeff} shows that the empirical coverage is near $0.95$ for the five $T=2283$ DGPs.

\begin{table}[t!]
	\centering
	\small
	\begin{tabular}{l|cccc}
		\toprule 
		\multicolumn{1}{c|}{DGP}	& $\phi_1$&$\phi_2$&$\phi_3$&$\phi_4$\\
		\hline
		Homogeneous spillovers & 0.96 & 0.95 & 0.95 & 0.95 \\
		& (0.16) & (0.3) & (0.075) & (0.058) \\
		Coefficient outlier & 0.95 & 0.95 & 0.95 & 0.94 \\
		& (0.16) & (0.3) & (0.075) & (0.058) \\
		Variance outlier & 0.95 & 0.96 & 0.95 & 0.95 \\
		& (0.43) & (0.18) & (0.046) & (0.044) \\
		Application & 0.96 & 0.95 & 0.96 & 0.95 \\
		& (0.12) & (0.32) & (0.091) & (0.14) \\
Short $T$ & 1.0 & 0.90 & 0.97 & 0.94 \\
 & (0.84) & (1.5) & (0.37) & (0.28) \\
Near-homogeneous size & 0.95 & 0.95 & 0.95 & 0.95 \\
 & (0.098) & (0.098) & (0.097) & (0.097) \\		
		\bottomrule
	\end{tabular}
	\caption{Simulation results for the RGIV unit-specific spillover coefficients (5000 simulations, $T = 2283$, 95\% confidence intervals).  Empirical coverage probabilities are listed above median confidence interval lengths, which are provided in parentheses. }
	\label{table:simulation_spillovercoeff}
\end{table}

\alert{I find that coverage worsens for RGIV for small sample sizes and for Oracle GIV when size is insufficiently skewed. For the ``Short $T$'' DGP, the empirical coverage of the size-weighted RGIV spillover coefficient is 90\% compared to 94\% for Oracle GIV. These results suggest caution in applying RGIV when the number of moment conditions is large relative to the sample size. For the ``Near-homogeneous size'' DGP, Oracle GIV severely under-covers with an empirical coverage rate of 63\% compared to 95\% for size-weighted RGIV. These results suggest that when all other requirements are satisfied, the researcher should be cautious when using the GIV estimator if size is insufficiently skewed relative to the sample size.}

\section{Discussion}
\label{sec:discussion}
\cite{Gabaix2024} introduces granularity-based identification, a substantial step forward for credible spillover estimates in macroeconomics and finance. Its baseline granular instrumental variables procedure, however, requires strong assumptions---namely homogeneous spillovers across units, homogeneous shock variances, skewed unit size, and idiosyncratic shocks (after accounting for common factors with known loadings). I build on this innovative approach by showing that unit-specific spillover coefficient heterogeneity with heterogeneous (and unknown) shock variances can jointly be accounted for without further restrictions. My estimator, called robust granular instrumental variables, also allows for homogeneous unit sizes unlike GIV.  Intuitively, my approach uses internally estimated individual idiosyncratic shocks as instruments. I give results on identification and inference, showing that the required GIV assumption of coefficient homogeneity is directly testable.  Relaxing the idiosyncratic shock assumption, I highlight a tradeoff: practitioners must choose between allowing unrestricted unit-level heterogeneity and a general shock covariance structure. Studying Euro area sovereign yields, I find strong evidence of country-level heterogeneity in spillovers. I find that my proposed estimator has good finite sample properties through simulation.

There are several directions for future work. First, my approach builds on the baseline framework proposed by \cite{Gabaix2024}. Here, an individual unit's outcome is determined by the size-weighted aggregate outcome. A structure where spillover responses are allowed to differ by the shock's source could be of interest for empirical work, \alert{like in explicitly allowing for differential spillovers for shocks originating from core versus periphery countries in the sovereign yields application}. \alert{Second, while the empirical applications featured in this paper have natural bases for grouping units (like the core--periphery for the sovereign yields application and operational similarity for the inelastic markets hypothesis application), the grouping of units or the choice of moments to use could be automated (for instance see \citet{Donald2001} and \citet{Cheng2015}). Doing so would ensure scalability as the number of units increases.} Third, as discussed in Section \ref{sec:extensions_factor}, practitioners face a tradeoff between unrestricted unit-level heterogeneity in spillover coefficients and correlated shocks. Further work on alternative economically-motivated conditions that preserve global identification would be of considerable interest to applied users.

\appendix
\section{Main proofs}

\subsection{Proof of Lemma \ref{lemma:gmm_id}}
	\label{lemma:gmm_id_proof}
		
	I prove Lemma \ref{lemma:gmm_id} under the weaker condition that the covariance of shocks $i$ and $j$ is known. The below proof will use Assumption \ref{ass:general} after replacing Condition (ii) with Condition (ii'):

	\begin{manualcondition}{(ii')}
 For $\sigma^2_{i} > 0$, shocks $\V{u}_t$ are i.i.d. with moments $\EE(\V{u}_t) = 0$, $\EE(u_{it}^2) = \sigma^2_i$,  and $\EE(\| \V{u}_t\|^4) < \infty.$ Moreover, $\EE(u_{it} u_{jt}) = \mu_{ij}$ for known $\mu_{ij}$. 		
	\end{manualcondition}

	 Take $\widetilde{\boldsymbol{ \phi}} \in \boldsymbol{\Phi}$. The moment condition for the idiosyncratic shocks of units $i$ and $j$ 
{\small	\begin{align}
		\mu_{ij} &= \EE[(r_{it} - \widetilde{\phi}_i r_{St})(r_{jt} - \widetilde{\phi}_j r_{St})] = \EE[((\phi_i - \widetilde{\phi}_i)r_{St} + u_{it})((\phi_j - \widetilde{\phi}_j)r_{St} + u_{jt})] \tag{Assumption \ref{ass:general_reducedform}} \\
		&= (\phi_i - \widetilde{\phi}_i)\EE[r_{St} u_{jt}] + (\phi_j - \widetilde{\phi}_j) \EE[r_{St} u_{it}] + (\phi_i - \widetilde{\phi}_i)(\phi_j - \widetilde{\phi}_j) \EE[r_{St}^2]  + \mu_{ij}. \label{eq:general_shockcorr}
	\end{align}}
	Rearranging, the above display implies 
{\small	\begin{align}
		\phi_i - \widetilde{\phi}_i &= - \frac{(\phi_j - \widetilde{\phi}_j) \EE[r_{St} u_{it}]}{\EE[r_{St} u_{jt}] + (\phi_j - \widetilde{\phi}_j) \EE[r_{St}^2]} \text{ and } 
		 		\phi_k - \widetilde{\phi}_k = - \frac{(\phi_1 - \widetilde{\phi}_1) \EE[r_{St} u_{kt}]}{\EE[r_{St} u_{1t}] + (\phi_1 - \widetilde{\phi}_1) \EE[r_{St}^2]} \label{eq:general_phitilde}
	\end{align}}
	for $k>1$. Substituting Equation \ref{eq:general_phitilde} into Equation \ref{eq:general_shockcorr} and rearranging gives
	\begin{align}
		0 &= \frac{(\phi_1 - \widetilde{\phi}_1) \EE[r_{St} u_{it}] \EE[r_{St} u_{jt}] }{\EE[r_{St} u_{1t} ]+ (\phi_1 - \widetilde{\phi}_1) \EE[r_{St}^2]} \left\{ -2 +  \frac{(\phi_1 - \widetilde{\phi}_1) \EE[r_{St}^2] }{\EE[r_{St} u_{1t} ]+ (\phi_1 - \widetilde{\phi}_1) \EE[r_{St}^2]}\right\}. \label{eq:general_momcond}
	\end{align}
	
	Equation \ref{eq:general_momcond} equals zero when either the first or second terms equal zero. Beginning with the first case, the first multiplicative term equals zero when $\widetilde{\phi}_1 = \phi_1$. Applying Equation \ref{eq:general_phitilde} for $k > 1$, $\widetilde{\phi}_k = \phi_k$. Thus, the first case corresponds to the true solution $\widetilde{\boldsymbol{\phi}} = \boldsymbol{\phi}$.
	
	Focusing on the second case, the second multiplicative term of Equation \ref{eq:general_momcond} equals zero when 		$0=-2 +  \frac{(\phi_1 - \widetilde{\phi}_1) \EE[r_{St}^2] }{\EE[r_{St} u_{1t} ]+ (\phi_1 - \widetilde{\phi}_1) \EE[r_{St}^2]}$ or equivalently $\EE[r_{St} u_{1t}] = -\frac{1}{2} (\phi_1 - \widetilde{\phi}_1) \EE[r_{St}^2]$.	Substituting to Equation \ref{eq:general_phitilde} implies 
	\begin{align}
\phi_k - \widetilde{\phi}_k = -2 \frac{\EE[r_{St} u_{kt}]}{\EE[r_{St}^2]} \text{ for } k\geq1.		
\label{eq:general_falsesolution}
	\end{align}

	After taking a size-weighted sum of $\widetilde{\phi}_k$, $\sum_{i=1}^n S_i \widetilde{\phi}_i = \phi_S + 2(1-\phi_S) = 1 + (1-\phi_S) > 1.$
		The final inequality follows from the $\phi_S< 1$ condition from Assumption \ref{ass:general_reducedform}.	In words, the second multiplicative term of Equation \ref{eq:general_momcond} equals zero for a value of $\widetilde{\boldsymbol{\phi}}$ that is outside the parameter space.		Hence, the only solution to $\EE[g(\V{r}_t, \widetilde{\boldsymbol{ \phi}})] = 0$ is $\boldsymbol{ \phi} = \boldsymbol{ \phi}_0$.

\subsection{Proof of Theorem \ref{thm:gmm_consistency}}
\label{thm:gmm_consistency_proof}

I will begin by showing that the hypotheses of Theorem 3.1 of \cite{Pakes1989} are satisfied by Assumption \ref{ass:general}. Doing so gives consistency of a GMM estimator with moment function $g(\V{r}_t, \boldsymbol{\phi})$ and the weight matrix taken to be the identity matrix.  To do so, I proceed by verifying the (more stringent) conditions of \cite{Newey1994} Theorem 2.6, taking the weight matrix as the identity matrix. I close by showing that the constructed weight matrix $\widehat{W}(\boldsymbol{\phi})$ satisfies Lemma 3.4 of \cite{Pakes1989}, implying consistency of the RGIV estimator.

Below, I verify the conditions of Theorem 2.6 of \cite{Newey1994} for an identity weight matrix $I_{n(n-1)/2}$ (preserving the same numbering as the original reference): 
\begin{enumerate}[label=\roman*.]
	\item 	Since the weight matrix $\widehat{W} = I_{n(n-1)/2}$ is taken to be the identity matrix, trivially $\widehat{W} \xrightarrow{p} I_{n(n-1)/2}.$ Lemma \ref{lemma:gmm_id} shows $\EE[g(\V{r}_t, \boldsymbol{\phi})] = 0$  if and only if $\boldsymbol{\phi} = \boldsymbol{\phi}_0$.
	\item Parameter space $\boldsymbol{\Phi}$ is compact from Assumption \ref{ass:general_parameterspace}.
	\item By inspection, $g(\V{r}_t, \boldsymbol{\phi})$ is continuous at each $\boldsymbol{\phi}$ with probability 1.
	\item Take $\widetilde{\boldsymbol{ \phi}} \in \boldsymbol{\Phi}$. Then, consider the element of $g(\V{r}_t, \widetilde{\boldsymbol{\phi}})$ corresponding to the orthogonality of shocks to units $i$ and $j$. Then, Assumption \ref{ass:general_reducedform} implies
$(r_{it} - \widetilde{\phi}_i r_{St})(r_{jt} - \widetilde{\phi}_j r_{St}) = [u_{it} + (\phi_i - \widetilde{\phi}_i)r_{St}][u_{jt} + (\phi_j - \widetilde{\phi}_j)r_{St}] = [u_{it} + (\phi_i - \widetilde{\phi}_i)\frac{u_{St}}{1-\phi_S}][u_{jt} + (\phi_j - \widetilde{\phi}_j)\frac{u_{St}}{1-\phi_S}].$
Thus, the moment condition is quadratic in idiosyncratic shocks. Since $\EE[\|\V{u}_t\|^2] < \infty$, $\EE[\sup_{\phi \in \boldsymbol{\Phi}} \| g(\V{r}, \boldsymbol{ \phi})\|] < \infty$.			
\end{enumerate}		

Since the more stringent conditions of \cite{Newey1994} Theorem 2.6 are satisfied, the conditions of Theorem 3.1 of \cite{Pakes1989} are also satisfied. 

$\widehat{W}(\boldsymbol\phi)$ is positive definite, so it can be decomposed as $\widehat{W}(\boldsymbol{ \phi}) = \widehat{A}(\boldsymbol{ \phi})'\widehat{A}(\boldsymbol{ \phi})$ where $\widehat{A}(\boldsymbol{\phi})=	\sqrt{\widehat{W}(\boldsymbol\phi)}$ where $\sqrt{\cdot}$ represents the element-wise square root operator. Next, proceed by verifying the conditions of Lemma 3.4 of \cite{Pakes1989} (preserving the same numbering as the original reference). Conditions (i) and (ii) follow from the law of large numbers, continuous mapping theorem, and $0<\sigma_i^2 < \infty$. Hence, RGIV is consistent.

\bibliographystyle{apalike}
\bibliography{references}

\newpage
\section{Online appendix}
			\label{sec:additional_results}

\subsection{Proof of Theorem \ref{thm:gmm_asymptotic_normality}}
\label{thm:gmm_asymptotic_normality_proof}
		
I will first show that the hypotheses of Theorem 3.3 of \cite{Pakes1989} are satisfied by Assumption \ref{ass:general}. Doing so gives asymptotic normality of a GMM estimator with moment function $g(\V{r}_t, \boldsymbol{ \phi})$ and the weight matrix taken to be the identity matrix. I begin by verifying the (more stringent) conditions of \cite{Newey1994} Theorem 3.4 where the weight matrix is taken to be the identity matrix. Then, I close by showing that the estimator with weight matrix $\widehat{W}(\boldsymbol{ \phi})$ satisfies Lemma 3.5 of \cite{Pakes1989}, implying that the RGIV estimator is asymptotically normal with asymptotic variance $(G'WG)^{-1}G'W \Sigma W G (G'WG)^{-1}$.
		
Below, I verify the conditions of \cite{Newey1994} Theorem 3.4. Before doing so, I precompute $\nabla_\theta g(\V{r}_t, \widetilde{\boldsymbol{ \phi}})$ for $\widetilde{\boldsymbol{ \phi}} \in \boldsymbol{\Phi}$
{\small\begin{align*}
	\nabla_\phi g(\V{r}_t, \widetilde{\boldsymbol{ \phi}}) &= \begin{bmatrix}
		M_2 \\
		M_3 \\
		\vdots \\
		M_n
	\end{bmatrix} 
	M_i = \begin{bmatrix}
		0_{(n-i+1) \times (i-2) } & \begin{matrix}
			-r_{St} (r_{it} - \widetilde{\phi}_i r_{St}) \\
			-r_{St} (r_{(i+1)t} - \widetilde{\phi}_{i+1} r_{St}) \\				
			\vdots \\
			-r_{St} (r_{nt} - \widetilde{\phi}_n r_{St}) \\				
		\end{matrix} & -I_{n-i+1} \cdot r_{St} (r_{(i-1)t} - \widetilde{\phi}_{i-1}r_{St})
	\end{bmatrix}.
\end{align*}}
Also, for more compact notation, let $K = \sum_{i=1}^n S_i ^2 \sigma_i^2$ where $\sigma_i^2$ is given by Assumption \ref{ass:general_shockmoments}. Moreover, the hypotheses of \cite{Newey1994} Theorem 2.6 are satisfied in  Theorem \ref{thm:gmm_consistency}. I check the hypotheses of Theorem 3.4 of \cite{Newey1994} for identity weight matrix $\widehat{W} = I_{n(n-1)/2}$ (preserving the numbering of the original reference):

\begin{enumerate}[label=\roman*., leftmargin=*]
	\item By Assumption \ref{ass:general_parameterspace}, $\boldsymbol{ \phi}_0 \in \mathrm{interior}(\boldsymbol{\Phi})$.
	\item Inspecting the functional form of $\nabla_{\theta} g(\V{r}, \widetilde{\boldsymbol{ \phi}})$, $g(\V{r}, \widetilde{\boldsymbol{ \phi}})$ is continuously differentiable in a neighborhood $\nn$ of $\boldsymbol{ \phi}_0$ with probability approaching 1.

	\item $\EE[g(\V{r}, \boldsymbol{ \phi}_0)]=0$ is shown in Lemma \ref{lemma:gmm_id}.  The square of the shock orthogonality condition between countries $i$ and $j$ is bounded $\EE[u_{it}^2u_{jt}^2]^2 \leq \EE[u_{it}^4] \EE[u_{jt}^4] < \infty.$ Thus, $\EE[\| g(\V{r}, \boldsymbol{ \phi}_0) \|^2]$	is finite.

	\item Fix $i \in 1,...,n$ and $\widetilde{\phi} \in \boldsymbol{\Phi}$. Then
	\begin{align*}
		\| -r_{St} &(r_{it} - \widetilde{\phi}_i r_{St}) \| = \| -r_{St} \big(u_{it} + (\phi_i - \widetilde{\phi}_i)r_{St}\big) \| \tag{Assumption \ref{ass:general_reducedform}} \\
		&\leq \| r_{St} u_{it} \| + \|\phi_i - \widetilde{\phi}_i\| \cdot \| r_{St}^2\| =  \| \frac{u_{St}}{1-\phi_S} u_{it} \| + \|\phi_i - \widetilde{\phi}_i\| \cdot \| \big(\frac{u_{St}}{1-\phi_S} \big)^2\|. 
	\end{align*}
	
	The above expression is quadratic in idiosyncratic shocks $u_{it}$. From Assumption \ref{ass:general_shockmoments}, $\EE[\V{u}_t \V{u}_t']$ exists and $\EE[\sup_{\boldsymbol{ \phi} \in \boldsymbol{\Phi}} \| -r_{St} (r_{it} - \widetilde{\phi}_i r_{St}) \| ] < \infty$, so $\EE[\sup_{\boldsymbol{ \phi}\in \boldsymbol{\Phi}} \| \nabla_\phi g(r, \boldsymbol{ \phi})\|] < \infty$.

	\item  To show $G'G$ is full rank, it suffices to show that $\textrm{rank}(G) = n$. Recall $G = \EE[\nabla_{\phi} g(\V{r}_t, \boldsymbol{ \phi}_0)]$. Since $\EE[-r_{St}(r_{it} - \phi_i r_{St})] = \EE[r_{St} u _{it}] = -\EE\big[\frac{u_{St} u_{it}}{1-\phi_S}\big] = -\frac{S_i \sigma^2_i}{1-\phi_S}$, $G$ can be explicitly computed:
{\small	\begin{align*}
			G& = \begin{bmatrix}
	\EE(M_2) \\
	\EE(M_3) \\
	\vdots \\
	\EE(M_n)
\end{bmatrix} , \quad \EE(M_i) = -\frac{1}{1-\phi_S} \begin{bmatrix}
	0_{(n-i+1) \times (i-2) } & \begin{matrix}
		S_i \sigma^2_i	\\
		S_{i+1} \sigma^2_{i+1}\\				
		\vdots \\
		S_n \sigma^2_n\\				
	\end{matrix} & I_{n-i+1} \cdot S_{i-1}\sigma^2_{i-1}
\end{bmatrix}.
	\end{align*}}

To show $G$ is full rank, consider the $n \times n$ submatrix $G_{1:n, 1:n}$ formed by taking the first $n$ rows of $G$. Then, the determinant of $G'_{1:n, 1:n}$ can be computed by applying elementary row operations:
\begingroup
\allowdisplaybreaks
{\scriptsize\begin{align*}
	&\det(G_{1:n, 1:n}') = \det\left(-\frac{1}{1-\phi_S} \begin{bmatrix}
		\begin{matrix}
			S_2 \sigma^2_2 & S_3 \sigma^2_3 & \dots & S_n \sigma^2_n 
		\end{matrix} & 0 \\
		S_1 \sigma^2_1 I_{n-1} & \begin{matrix}
			S_3 \sigma^2_3 \\
			S_2  \sigma^2_2 \\
			0 \\
			\vdots\\
			0
		\end{matrix}
	\end{bmatrix}\right) \\
	&= (-1)^{n-1}\big(-\frac{1}{1-\phi_S}\big)^n \big[S_1 \sigma^2_1\big]^{n-1} \big[ -\frac{S_3 \sigma^2_3}{S_1 \sigma^2_1} S_2 \sigma^2_2 -\frac{S_2 \sigma^2_2}{S_1 \sigma^2_1} S_3 \sigma^2_3  \big] = 2 \Big(\frac{1}{1-\phi_S}\Big)^n (S_1 \sigma^2_1)^{n-2} S_2 S_3 \sigma^2_2 \sigma_3^2 \neq 0.
\end{align*}}
\endgroup

The determinant in line 2 is the product of the matrix's diagonal elements because the matrix is upper triangular. Since $1-\phi_S \neq 0$, $S_i > 0$, and  $\sigma^2_i >0$, the determinant is not zero. Hence, submatrix $G'_{1:n, 1:n}$ is full rank, which immediately implies $G$ is also full rank	$n = \textrm{rank}(G'_{1:n, 1:n}) \leq \textrm{rank}(G) \leq \min(\frac{n(n-1)}{2},n) = n.$

\end{enumerate}

Since the more stringent conditions of \cite{Newey1994} Theorem 3.4(i)-(v) are satisfied, the conditions of Theorem 3.3 of \cite{Pakes1989} are also satisfied. Continue by verifying the conditions of Lemma 3.5 of \cite{Pakes1989}. Recall from the proof of Theorem \ref{thm:gmm_consistency_proof} that the weight matrix can be decomposed as
$\widehat{W}(\boldsymbol{ \phi}) = \widehat{A}(\boldsymbol{ \phi})'\widehat{A}(\boldsymbol{ \phi})$ where $\widehat{A}(\boldsymbol{\phi})=	\mathrm{diag}\Big(\sqrt{\frac{1}{T} \sum_{t=1}^T g(\V{r}_t, \boldsymbol{\phi})^{\circ 2}} \Big)^{-1}$ for element-wise square root operator $\sqrt{\cdot}$. From the continuous mapping theorem and law of large numbers,	$\widehat{A}(\boldsymbol{ \phi}_0) \xrightarrow{p} \mathrm{diag}(1/(\sigma_1 \sigma_2),\dots, 1/(\sigma_1 \sigma_n), 1/(\sigma_2 \sigma_3),\dots, 1/(\sigma_{n-1} \sigma_n))=A.$
Let $\{\delta_n\}$ be a sequence of positive numbers that converges to zero. Then,	$\sup_{\| \boldsymbol{ \phi}- \boldsymbol{ \phi}_0 \| < \delta_n} \| \widehat{A}(\boldsymbol{ \phi}) - A\| = o_p(1)$ from continuity of the map from $\boldsymbol{ \phi}$ to $\widehat{A}(\boldsymbol{ \phi})$ at $\boldsymbol{ \phi} = \boldsymbol{ \phi}_0$. Hence, the RGIV estimator is asymptotically normal with asymptotic variance $(G'WG)^{-1}G'W \Sigma W G (G'WG)^{-1}$ for $W = A(\boldsymbol{ \phi}_0)'A(\boldsymbol{ \phi}_0) = \mathrm{diag}(1/(\sigma_1^2 \sigma_2^2),\dots, 1/(\sigma_1^2 \sigma_n^2), 1/(\sigma_2^2\sigma_3^2),\dots, 1/(\sigma_{n-1}^2 \sigma_n^2)).$

			\subsection{Proof of Proposition \ref{prop:GIV_het_elasticities}}
\label{prop:GIV_het_elasticities_proof}

Before computing the GIV estimand, $\EE[(u_{St} - u_{Et})u_{Et}]$, $\EE[z_t u_{Et}]$, and $\EE[(r_{St} - r_{Et})r_{St}]$ are precomputed below:
{\small
\begin{align*}
	\EE[(u_{St} - u_{Et})u_{Et}] &=  \EE\Big[ \big(\sum_{i=1}^n (S_i - \frac{1}{n})u_{it}\big)\frac{1}{n}\sum_{i=1}^n u_{it}  \Big] = \frac{1}{n} \sum_{i=1}^n (S_i - \frac{1}{n})\sigma^2 = 0. \\
	\EE[z_t u_{Et}] &= \EE[(r_{St} - r_{Et})u_{Et}] = \EE[(\phi_S - \phi_E) r_{St}u_{Et} + (u_{St} - u_{Et})u_{Et}] \\
	&= (\phi_S - \phi_E)\EE[r_{St} u_{Et}]  = \frac{\phi_S - \phi_E}{1-\phi_S} \EE[u_{St} u_{Et}]= \frac{\phi_S - \phi_E}{1-\phi_S} \frac{1}{n} \sigma^2. \\
\EE[(r_{St} - r_{Et})r_{St}]&= (\phi_S - \phi_E) \EE[r_{St}^2] + \EE[r_{St}(u_{St} - u_{Et})] = \frac{(\phi_S - \phi_E) \sigma^2 \sum_{i=1}^n S_i^2}{(1-\phi_S)^2} + \frac{\sigma^2}{1-\phi_S}[\sum_{i=1}^nS_i^2 - \frac{1}{n} ]. 
\end{align*}}

Combining, the GIV estimand is $\frac{\EE[z_t r_{Et}]}{\EE[z_t r_{St}]} = \frac{\EE[z_t (\phi_E r_{St} + u_{Et})]}{\EE[z_t r_{St}]} =  \phi_E + \frac{\phi_S - \phi_E}{n} \frac{1}{\frac{\phi_S - \phi_E}{1-\phi_S} \sum_{i=1}^n S_i^2 -\frac{1}{n} + \sum_{i=1}^n S_i^2}.$

\subsection{Proof of Lemma \ref{lemma:notid}}
\label{lemma:notid_proof}

	The size-weighted outcome variable can be written as idiosyncratic shocks and the common factor	$r_{St} = \phi_S r_{St} + \lambda_S f_t + u_{St} = \frac{1}{1-\phi_S}[\lambda_S f_t + u_{St}].$	For true parameter values $\bm{\theta} = [\bm{\phi}', \bm{\lambda}']'$. We will show that the system of equations given by the moment conditions doesn't admit a unique solution; it needn't  be true that $\widetilde{\bm\theta} = \bm\theta$.

	For units $i\neq j$, the moment condition can be written as
{\small	\begin{align}
		\widetilde{\lambda}_i \widetilde{\lambda}_j &= \EE\big\{ [r_{it} - \widetilde{\phi}_i r_{St}] [r_{jt} - \widetilde{\phi}_j r_{St}]   \big\} = \EE \big\{ [(\phi_i - \widetilde{\phi}_i ) r_{St} + \lambda_if_t + u_{it}][(\phi_j - \widetilde{\phi}_j ) r_{St} + \lambda_jf_t + u_{jt}] \big\} \nonumber\\
		&=(\phi_i - \widetilde{\phi}_i)(\phi_j - \widetilde{\phi}_j) \Big[\frac{\lambda_S^2 + \sum_{k=1}^n S_k^2 \sigma^2_k }{(1-\phi_S)^2}\Big] + (\phi_i - \widetilde{\phi}_i) \lambda_j \frac{\lambda_S}{1-\phi_S} \nonumber \\
		&\quad + (\phi_j - \widetilde{\phi}_j) \lambda_i \frac{\lambda_S}{1-\phi_S} + (\phi_i - \widetilde{\phi}_i) \frac{S_j \sigma^2_j}{1-\phi_S} + (\phi_j - \widetilde{\phi}_j) \frac{S_i \sigma^2_i}{1-\phi_S} + \lambda_i \lambda_j \label{eq:ortho}
	\end{align}
}One solution to above is $\widetilde{\bm{\theta}} = \bm{\theta}$. 	Next, guess that another solution satisfies the restrictions $\widetilde{\phi}_1 = \phi_1$ and $\widetilde{\lambda}_1 = \lambda_1$. From this guess, $\phi_k$ for $k>1$ can be expressed in terms of $\phi_1$
	\begin{align}
		\phi_k - \widetilde{\phi}_k &= \frac{\widetilde{\lambda}_k \widetilde{\lambda}_1 - \lambda_k \lambda_1 - (\phi_1 - \widetilde{\phi}_1)\Big[\frac{\lambda_k \lambda_S + S_k \sigma^2_k}{1-\phi_S}\Big] }{\frac{\phi_1 - \widetilde{\phi}_1}{(1- \phi_S)^2}K + \frac{\lambda_1 \lambda_S + S_1 \sigma^2_1}{1-\phi_S} } = \lambda_1(1-\phi_S) \frac{\widetilde{\lambda}_k  - \lambda_k }{  \lambda_1 \lambda_S + S_1 \sigma^2_1 }. \label{eq:phik_unit1}		
	\end{align}
	The final equality uses the restrictions $\widetilde{\phi}_1 = \phi_1$ and $\widetilde{\lambda}_1 = \lambda_1$. Storing constants as  $	K_j = \frac{\lambda_1 (\lambda_j \lambda_S + S_j \sigma^2_j)}{\lambda_1 \lambda_S + S_1 \sigma^2_1}$ for $j > 1$ and $K_0 = \frac{\lambda_1^2 K}{(\lambda_1 \lambda_S + S_1 \sigma^2_1)^2}$, substituting Equation \ref{eq:phik_unit1} into  Equation \ref{eq:ortho} gives
	\begin{align}
		\widetilde{\lambda}_i \widetilde{\lambda}_j &= K_0 (\widetilde{\lambda}_i - \lambda_i)(\widetilde{\lambda}_j - \lambda_j) + K_j(\widetilde{\lambda}_i - \lambda_i) + K_i(\widetilde{\lambda}_j - \lambda_j) + \lambda_i \lambda_j.
		\label{eq:ortho_2}
	\end{align}
	
	Consider the guess $\widetilde{\lambda}_j = \frac{\lambda_j - 2K_j + K_0 \lambda_j}{K_0 - 1}$. Continue by verifying that the guess satisfies Equation \ref{eq:ortho_2}. Expanding, the left-hand side can be written as
	{\small
		\begin{align*}
			LHS = \widetilde{\lambda}_i \widetilde{\lambda}_j = \frac{2K_0 \lambda_i \lambda_j - 2K_0 K_j \lambda_i - 2K_0 K_i \lambda_j + 4 K_i K_j - 2K_j \lambda_i - 2 K_i \lambda_j + \lambda_i \lambda_j K_0^2 + \lambda_i \lambda_j}{(K_0 - 1)^2}.
		\end{align*}
	}	
	
	Noting that $\widetilde{\lambda}_j - \lambda_j = 2\frac{\lambda_j - K_j}{K_0 - 1}$, the RHS of Equation \ref{eq:ortho_2} can be written as 
{	\small
	\begin{align*}
		RHS &= \frac{1}{(K_0 - 1)^2}\{4 K_0 (\lambda_i - K_i)(\lambda_j - K_j) + 2K_j (\lambda_i - K_i)(K_0 - 1) \\
		&\quad +2K_i (\lambda_j - K_j) (K_0 - 1) + \lambda_i \lambda_j (K_0-1)^2 \} = LHS.
	\end{align*}}

			\subsection{Proof of Lemma \ref{lemma:no_id_factor}}
				\label{app:no_id_factor}
			Let $\boldsymbol{\phi}$ be the true spillover coefficient and let $\check{\boldsymbol{\phi}} = \boldsymbol{\phi} - \widetilde{\boldsymbol{\phi}}$. Let 
\begin{align*}
	K_0 &= \EE[r_{St}^2] = \Big(\frac{1}{1-\phi_S}\Big)^2 \Big[\lambda_S^2 + \sum_{i=1}^n S_i^2 \sigma^2_i\Big] \\
	K_i &= \frac{\lambda_i \lambda_S + S_i \sigma^2_i}{1-\phi_S} \text{ for } i = 1,\ldots,n.
\end{align*}	
Then the element of $g_0(\widetilde{\boldsymbol{\phi}})$ corresponding to shocks $i$ and $j$ is $	\check{\phi}_i \check{\phi}_j K_0 + \check{\phi}_i K_j + \check{\phi}_j K_i + \lambda_i \lambda_j.$

	Now prove the contrapositive. Suppose there exists $\widetilde{\boldsymbol{\phi}}$ (and corresponding $\check{\boldsymbol{\phi}}$) such that
\begin{align}
	0 &= \check{\phi}_i \check{\phi}_j K_0 + \check{\phi}_i K_j + \check{\phi}_j K_i + \lambda_i \lambda_j 	\label{eq:moment_ij_singlefactor}
\end{align}	
	 for all distinct pairs $i,j$. We want to show that this implies there exists $c \in \mathbb{R}$ such that $S_i\sigma^2_i/\lambda_i = c$ for all but at most one $i \in \{1,\ldots,n\}$.  
	
	Define $A_{mn} \equiv K_m K_n - \lambda_m \lambda_n K_0$. Multiplying Equation \ref{eq:moment_ij_singlefactor} by $K_0$ and factoring:
\begin{equation}
    (\check{\phi}_i K_0 + K_i)(\check{\phi}_j K_0 + K_j) = A_{ij} \quad \text{for all } 1 \leq i < j \leq n. \label{eq:factored}
\end{equation}
Then, for any four distinct indices $i, j, k, l$, the following equality must hold
\begin{equation}
    (\check{\phi}_i K_0 + K_i)(\check{\phi}_j K_0 + K_j)(\check{\phi}_k K_0 + K_k)(\check{\phi}_l K_0 + K_l) = A_{ij}A_{kl} = A_{ik}A_{jl} = A_{il}A_{jk}. \label{eq:pairwise}
\end{equation}

Now let's derive the condition. Consider the second to last inequality of Equation \ref{eq:pairwise}
\begin{align}
    0 &= A_{ij} A_{kl} - A_{ik}A_{jl} = K_0 (\lambda_l K_i - \lambda_i K_l)(\lambda_j K_k - \lambda_k K_j) \nonumber\\
	&= \frac{K_0}{(1-\phi_S)^2}(\lambda_l S_i \sigma^2_i - \lambda_i S_l \sigma^2_l)(\lambda_j S_k \sigma^2_k - \lambda_k S_j \sigma^2_j).
	\label{eq:diff1}
\end{align}
Then Equation \ref{eq:diff1} holds if and only if  $S_i\sigma^2_i/\lambda_i = S_l\sigma^2_l/\lambda_l$ or $S_j\sigma^2_j/\lambda_j = S_k\sigma^2_k/\lambda_k$. An identical calculation shows that $A_{ik}A_{jl} = A_{il}A_{jk}$ (last equality of Equation \ref{eq:pairwise})  if and only if $S_i\sigma^2_i/\lambda_i = S_j\sigma^2_j/\lambda_j$ or $S_k\sigma^2_k/\lambda_k = S_l\sigma^2_l/\lambda_l$. Putting both conditions together, at least three of $\{S_i\sigma^2_i/\lambda_i, S_j\sigma^2_j/\lambda_j, S_k\sigma^2_k/\lambda_k, S_l\sigma^2_l/\lambda_l\}$ must be equal. The result for $n=4$ immediately follows.

 For $n > 4$, now let $i,j,k,l$ be an arbitrary set of four distinct indices. Use contradiction: suppose $i, j$ satisfy $S_i\sigma^2_i/\lambda_i \neq c$ and $S_j\sigma^2_j/\lambda_j \neq c$. We have a contradiction since it is not true that at least three of $\{S_i\sigma^2_i/\lambda_i, S_j\sigma^2_j/\lambda_j, S_k\sigma^2_k/\lambda_k, S_l\sigma^2_l/\lambda_l\}$ are equal.
			
		\subsection{RGIV with observed explanatory variables}
\label{sec:RGIV_controls}

This section contains formal results for the RGIV model with observed explanatory variables. 

Assumption \ref{ass:general_controls} (below) establishes conditions necessary for identification, consistency, and asymptotic normality. Assumption \ref{ass:general_controls} extends Assumption \ref{ass:general} to include control variables $\V{x}_{it}$ that are independent of idiosyncratic shocks $u_{it}$. In the derivations that follow, it will be convenient to reparametrize the coefficient on the control variables as $\boldsymbol{\psi}_i = \boldsymbol{{\beta}}_i + \frac{S_i \phi_i}{1-\phi_S} \boldsymbol{{\beta}}_i$. Then, store the re-parametrized coefficients in vector $\boldsymbol{\psi} = [\boldsymbol{\psi}_1',\dots \boldsymbol{\psi}_n']'$ in vector $\boldsymbol{\theta} = [\boldsymbol{ \phi}', \boldsymbol{\psi}']'$ and its corresponding parameter space as $\boldsymbol{\Theta}$.

\begin{assumption}{(Baseline model)}
	
	\label{ass:general_controls}
	\begin{enumerate}[label = (\roman*), ref= \theassumption(\roman*)]
		\item \textbf{Model}: For $n  \geq 3$ units, let fixed sizes $S_i \in (0,1)$ sum to 1.  Outcome $\V{r}_t =[r_{1t},...,r_{nt}]'$ responds to the size-aggregated outcome $r_{St}$ according to spillover coefficient $\phi_i$, coefficients $\boldsymbol{ \beta}_i$ ($k \times 1$), observed control variables $\V{x}_{it}$ ($k \times 1$), and unobserved shocks $\V{u}_{t} = [u_{1t}, u_{2t}, ...,u_{nt}]'$, \label{ass:general_reducedform_controls} $r_{it}  = \phi_i r_{St} +  \boldsymbol{ \beta}_i' \V{x}_{it} + u_{it}$ for all $i = 1,...,n$ where $\phi_S < 1.$

		\item \textbf{Moments}: For $\sigma^2_{i} > 0$, shocks $\V{u}_t$ are i.i.d. with moments $\EE(\V{u}_t) = 0, \quad \EE(\V{u}_t \V{u}_t') =  \mathrm{diag}(\sigma^2_1,...,\sigma^2_n), \quad \EE(\| \V{u}_t\|^4) < \infty.$ Moreover, $u_{it} \indep u_{jt}$ for $i \neq j$. $\EE(\V{x}_{it}\V{x}_{it}') = \Sigma_{XX}^i$ for positive definite $\Sigma_{XX}^i$ and $\V{x}_{it} \indep u_{it}$. \label{ass:general_shockmoments_controls}

		\item \textbf{Parameter space}: 
		For spillover coefficient $\boldsymbol{\phi} = [\phi_1,\dots,\phi_n]'$,  store parameters in $\check{\boldsymbol{ \theta}} = [\boldsymbol{ \phi}', \boldsymbol{ \beta}_1',\dots, \boldsymbol{ \beta}_n']'$. Then, the true parameter $\check{\boldsymbol{\theta}}_0$ is in the interior of parameter space $\check{\boldsymbol{\Theta}}$.  $\check{\boldsymbol{\Theta}}$ is compact and for any $\check{\boldsymbol{\widetilde{\theta}}}  \in \check{\boldsymbol{\Theta}}$, $\widetilde{\phi}_S   < 1$. \label{ass:general_parameterspace_controls}
	\end{enumerate}
\end{assumption}

Definition \ref{defn:naive_controls} (below) extends the RGIV estimator to include observed explanatory variables. Step 1 estimates the total effect of control variables $\V{x}_{it}$ on the outcome $r_{it}$. Step 2 writes the moment uncorrelatedness condition after subtracting the variation induced by the observed explanatory variables. See Section \ref{sec:extensions_controls} for intuition.

\begin{definition}
	\label{defn:naive_controls}
	Store data $\V{z}_t = [\V{r}_t', \V{x}_{1t}', \dots, \V{x}_{nt}']'$ for  outcome variable $\V{r}_t = [r_{1t}, \dots,r_{nt}]'$. Let $u_{i}(\V{z}_t, \boldsymbol{\phi}; \boldsymbol{\psi}) = (r_{it} - \boldsymbol{ \psi}_i' \V{x}_{it})  - \phi_i \sum_{\ell=1}^n S_\ell (r_{\ell t} - \boldsymbol{ \psi}_\ell' \V{x}_{\ell t})$ for $i=1,...,n$.
	The moment function $g^c(\V{z}_t, \boldsymbol{\phi}; \boldsymbol{\psi})$  is
{\small
	\begin{align*}
		g^c(\V{z}_t, \boldsymbol{\phi}; \boldsymbol{\psi}) &= [u_{1}(\V{z}_t, \boldsymbol{\phi}; \boldsymbol{\psi}) u_{2}(\V{z}_t, \boldsymbol{\phi}; \boldsymbol{\psi}), \dots, \, u_{1}(\V{z}_t, \boldsymbol{\phi}; \boldsymbol{\psi}) u_{n}(\V{z}_t, \boldsymbol{\phi}; \boldsymbol{\psi}),\\
		&\quad \, \text{ }u_{2}(\V{z}_t, \boldsymbol{\phi}; \boldsymbol{\psi}) u_{3}(\V{z}_t, \boldsymbol{\phi}; \boldsymbol{\psi}), \dots ,\text{}u_{n-1}(\V{z}_t, \boldsymbol{\phi}; \boldsymbol{\psi}) u_{n}(\V{z}_t, \boldsymbol{\phi}; \boldsymbol{\psi})]' .
	\end{align*}}
	
	Let the sample weight matrix $\widehat{W}$ be defined as $\widehat{W} = \mathrm{diag}\Big(\frac{1}{\widehat{\sigma}_1^2 \widehat{\sigma}_2^2}, \dots, \,\frac{1}{\widehat{\sigma}_1^2 \widehat{\sigma}_n^2}, \frac{1}{\widehat{\sigma}_2^2 \widehat{\sigma}_3^2}, \dots, \frac{1}{\widehat{\sigma}_{n-1}^2 \widehat{\sigma}_n^2}\Big)$ where $\widehat{\sigma}_i^2$ is a consistent estimator of idiosyncratic shock variances $\sigma^2_i$ for $i=1,\dots,n$.	Then the \textbf{robust granular instrumental variables estimator with observed explanatory variables} is characterized as the following two-step estimator:
	\begin{enumerate}
		\item For each $i=1,\dots,n$, $\widehat{\boldsymbol{\psi}}^\text{Step 1}_i = (\frac{1}{T} \sum_{t=1}^T \V{x}_{it} \V{x}_{it}')^{-1} (\frac{1}{T} \sum_{t=1}^T \V{x}_{it} r_{it})$.
		\item For $\widehat{Q}_T(\boldsymbol{\phi}; \widehat{\boldsymbol{\psi}}^\text{Step 1})=  \Big[\frac{1}{T} \sum_{t=1}^T g^c(\V{z}_t, \boldsymbol{\phi};\widehat{ \boldsymbol{\psi}}^\text{Step 1}) \Big]' \widehat{W} \Big[\frac{1}{T} \sum_{t=1}^T g^c(\V{z}_t, \boldsymbol{\phi};\widehat{ \boldsymbol{\psi}}^\text{Step 1})\Big]$, 
		$\widehat{\boldsymbol{\phi}}^{RGIV,c} = \arg\min_{\boldsymbol{\phi} \in \boldsymbol{\Phi}} \widehat{Q}_T(\boldsymbol{\phi}; \widehat{\boldsymbol{\psi}}^\text{Step 1} ).$

	\end{enumerate}

\end{definition}

For brevity, the formal results rely on simple extensions of the \cite{Newey1994} GMM identification, consistency, and asymptotic normality results as detailed in Appendix \ref{sec:NM1994_extend} to allow for the first step estimation of a nuisance parameter. Note that the two-step estimation results of \cite{Newey1994} (Section 6) cannot be directly applied since there are potentially more moments than estimated parameters. Alternatively, the estimator described in Definition \ref{defn:naive_controls} could be defined using a continuously updating GMM objective function in the second step. The procedure could be formalized by modifying the more general results of \cite{Pakes1989} to accommodate a first-step estimator.

Lemma \ref{lemma:rgiv_controls_ols} (below) establishes that the first step estimator $\widehat{\boldsymbol{\psi}}^\text{Step 1}$ is consistent and $O_p(1)$. Lemma \ref{lemma:gmm_id_controls} establishes identification of spillover coefficients $\boldsymbol{\phi}$. Theorems \ref{thm:gmm_consistency_controls} and
\ref{thm:gmm_asymptotic_normality_controls} give consistency and asymptotic normality of  $\widehat{\boldsymbol{\phi}}^\text{RGIV,c}$. 

\begin{lemma}
	\label{lemma:rgiv_controls_ols}
	Impose Assumption \ref{ass:general_controls}. As $T \to \infty$, $\widehat{\boldsymbol{\psi}}^\text{Step 1}_i \xrightarrow{p} \boldsymbol{\psi}_i$  and $\sqrt{T}(\widehat{\boldsymbol{\psi}}^\text{Step 1}_i - \boldsymbol{\psi}_i) = O_p(1)$.
\end{lemma}
\begin{proof}
	Since $0 = \EE [\V{x}_{it} (r_{it} - \boldsymbol{\psi}_i' \V{x}_{it})] = \EE[\V{x}_{it} u_{it}]$ and $\EE[\V{x}_{it} \V{x}_{it}']$ is full rank, the OLS estimator $\widehat{\boldsymbol{\psi}}^\text{Step 1}_i$ is consistent. Assumption \ref{ass:general_shockmoments_controls} (that $\EE(\V{u}_t \V{u}_t') = \mathrm{diag}(\sigma^2_1,\dots,\sigma^2_n)$ and $\EE(\V{x}_{it} \V{x}_{it}')$ full rank) implies $\widehat{\boldsymbol{\psi}}^\text{Step 1}_i$ is asymptotically normal so $\sqrt{T}(\widehat{\boldsymbol{\psi}}^\text{Step 1}_i - \boldsymbol{\psi}_i) = O_p(1)$.
\end{proof}

\begin{lemma}[Identification with observed explanatory variables]
	\label{lemma:gmm_id_controls}
	Impose Assumption \ref{ass:general_controls}. For $g_0^c(\boldsymbol{ \phi}) = \EE[g^c(\V{z}_t, \boldsymbol{\phi}; \boldsymbol{\psi}_0)]$, $g_0^c(\boldsymbol{ \phi}_0)= 0$ for the true parameter $\boldsymbol{\phi}_0$ and $g_0^c(\widetilde{\boldsymbol{ \phi}})  \neq 0$ for $\widetilde{\boldsymbol{ \phi}} \in \boldsymbol{\Phi}$ such that $\widetilde{\boldsymbol{ \phi}} \neq \boldsymbol{ \phi}_0$.
\end{lemma}
\begin{proof}
	The result immediately follows after applying Lemma \ref{lemma:gmm_id} to $\dot{r}_{it} = r_{it} - \boldsymbol{ \psi}_i' \V{x}_{it}$.
\end{proof}

\begin{theorem}[Consistency of RGIV with observed explanatory variables]
	\label{thm:gmm_consistency_controls}
	Impose Assumption \ref{ass:general_controls}. The RGIV estimator with observed explanatory variables is consistent $\widehat{\boldsymbol{\phi}}^{RGIV,c} \xrightarrow{p} \boldsymbol{\phi}_0$ for the true parameter $\boldsymbol{\phi}_0$ as $T \to \infty$.
\end{theorem}
\begin{proof}	
	Verify the conditions of Theorem \ref{thm:NM_consistency_modified}. Condition (i) follows from Lemma \ref{lemma:gmm_id_controls} and Lemma \ref{lemma:NM_id_modified}. Condition (ii) follows from Assumption \ref{ass:general_parameterspace_controls}. Condition (iii) follows by inspection of function $g^c(\V{z}_t, \boldsymbol{\phi}; \boldsymbol{\psi})$. Condition (iv) follows from Assumption \ref{ass:general_shockmoments_controls}.
\end{proof}

\begin{theorem}
	\label{thm:gmm_asymptotic_normality_controls}
	Impose Assumption \ref{ass:general_controls}. The RGIV estimator with observed explanatory variables is asymptotically normal 
	\begin{align*}
		\sqrt{T}(\widehat{\boldsymbol{ \phi}}^{RGIV,c} - \boldsymbol{ \phi}_0) \xrightarrow{d} \nn(0, (G'WG)^{-1}G'W \Sigma W G (G'WG)^{-1} )
	\end{align*}
	for RGIV population weight matrix $W =  \mathrm{diag}(\frac{1}{\sigma_1^2 \sigma_2^2},\dots, \frac{1}{\sigma_1^2 \sigma_n^2},  \frac{1}{\sigma_2 ^2\sigma_3^2},\dots, \frac{1}{\sigma_{n-1}^2 \sigma_n^2})$, moment covariance matrix $\Sigma = \EE[g^c(\V{r}_t, \boldsymbol{ \phi}_0; \boldsymbol{\psi}_0)g^c(\V{r}_t, \boldsymbol{ \phi}_0; \boldsymbol{\psi}_0)']$ and $G = \EE[\nabla_{\phi} g^c(\V{r}_t, \boldsymbol{ \phi}_0; \boldsymbol{\psi}_0)]$ as $T\to\infty$.
\end{theorem}
\begin{proof}
	Verify the conditions of Theorem \ref{thm:NM_asymptoticnormality_modified}. The hypotheses for Theorem \ref{thm:NM_consistency_modified} are satisfied. Lemma \ref{lemma:rgiv_controls_ols} implies $\sqrt{T}(\widehat{\boldsymbol{\psi}}^\text{Step 1} - \boldsymbol{\psi}_0) = O_p(1)$. Next, compute the elements of $\EE[\nabla_{\boldsymbol{\psi}} g^c(\textbf{z}_t, \boldsymbol{\phi}_0; \boldsymbol{\psi}_0)]$. Let $i,j,k \in \{1, \dots,n\}$ and $i \neq j \neq k$. Since $\textbf{x}_{it} \indep u_{it}$,
	\begin{align*}
		\EE\{&\frac{\partial}{\partial \boldsymbol{\psi}_i}[r_{it}  - \boldsymbol{\psi}_i' \textbf{x}_{it} - \phi_i\sum_{\ell=1}^n S_\ell (r_{\ell t} - \boldsymbol{ \psi}_\ell' \V{x}_{\ell t}) ] [r_{jt}  - \boldsymbol{\psi}_j' \textbf{x}_{jt} - \phi_j\sum_{\ell=1}^n S_\ell (r_{\ell t} - \boldsymbol{ \psi}_\ell' \V{x}_{\ell t}) ] \} \\
		&= \EE\{(-\textbf{x}_{it} + \phi_i S_i \textbf{x}_{it})u_{jt}(\textbf{z}_t,\boldsymbol{\phi}_0 ;\boldsymbol{\psi}_0) \} + \EE[\phi_j S_i \textbf{x}_{it} u_{it}(\textbf{z}_t, \boldsymbol{\phi}_0; \boldsymbol{\psi}_0)] = 0 \\
		\EE\{&\frac{\partial}{\partial \boldsymbol{\psi}_k}[r_{it}  - \boldsymbol{\psi}_i' \textbf{x}_{it}  - \phi_i\sum_{\ell=1}^n S_\ell (r_{\ell t} - \boldsymbol{ \psi}_\ell' \V{x}_{\ell t}) ] [r_{jt}  - \boldsymbol{\psi}_j' \textbf{x}_{jt} - \phi_j\sum_{\ell=1}^n S_\ell (r_{\ell t} - \boldsymbol{ \psi}_\ell' \V{x}_{\ell t}) ] \} \\
		&= \EE\{ \phi_i S_k \textbf{x}_{kt} u_{jt}(\textbf{z}_t, \boldsymbol{\phi}_0; {\psi}_0) \} + \EE[\phi_j S_k \textbf{x}_{kt} u_{it}(\textbf{z}_t, \boldsymbol{\phi}_0; \boldsymbol{\psi}_0)] = 0.
	\end{align*}
	Hence, $\EE[\nabla_{\boldsymbol{\psi}} g^c(\textbf{z}_t, \boldsymbol{\phi}_0; \boldsymbol{\psi}_0)]=0$. 	Condition (i) holds from Assumption \ref{ass:general_parameterspace_controls}. Condition (ii) holds by inspection of moment function $g(\V{z}_t, \boldsymbol{\phi}; \boldsymbol{\psi})$. Conditions (iii) and (iv) hold from Assumption \ref{ass:general_shockmoments_controls}. Condition (v) holds from applying the proof of Theorem \ref{thm:gmm_asymptotic_normality}.
\end{proof}

			\subsection{Observed explanatory variables: Second-step GMM estimation with a consistent first-step estimator of a nuisance parameter}
			\label{sec:NM1994_extend}
			For completeness, the below identification, consistency, and asymptotic normality results are  straightforward adaptations of those found in \cite{Newey1994} (abbreviated as NM). These results give conditions for identification and inference for two-step estimators, given a consistent first-step estimator of a nuisance parameter.  Following the notation of NM, the first step nuisance parameter estimator is consistent $\widehat{\boldsymbol{\gamma}} \xrightarrow{p} \boldsymbol{\gamma}_0$ where $\sqrt{n}(\widehat{\boldsymbol{\gamma}} - \boldsymbol{\gamma}_0) = O_p(1)$. In the second step, $\boldsymbol{\theta}$ is the parameter of interest and is estimated using GMM with moment function $g(\V{z}, \boldsymbol{\theta}; \boldsymbol{\gamma})$ and weight matrix estimator $\widehat{W} \xrightarrow{p} W$, where $\widehat{\boldsymbol{\gamma}}$ is an estimator for $\boldsymbol{\gamma}$. Moreover, the condition $\EE[\nabla_{\boldsymbol{\gamma}}g(\V{z}, \boldsymbol{\theta}_0; \boldsymbol{\gamma}_0)] = 0$ (for true parameter values $\boldsymbol{\gamma}_0$ and $\boldsymbol{\theta}_0$) ensures that estimation uncertainty for the first-step estimator doesn't enter the expression for the asymptotic variance of the second-step estimator, as is the case for the RGIV estimator with observed explanatory variables.
			
			The below identification lemma (Lemma \ref{lemma:NM_id_modified}) extends NM Lemma 2.3, the consistency theorem (Theorem \ref{thm:NM_consistency_modified}) extends NM Theorem 2.6, and  the asymptotic normality theorem (Theorem \ref{thm:NM_asymptoticnormality_modified}) extends NM Theorem 3.4. The proofs for these extensions are, for the most part, identical to the original NM results. Notably, care is taken to ensure uniform convergence of the GMM objective function in Theorem \ref{thm:NM_consistency_modified} and the expected Jacobian matrices in Theorem \ref{thm:NM_asymptoticnormality_modified}. 
			
			These results can be applied to estimators where the second step GMM estimator is over-identified, as is  the case for RGIV when $n \geq 4$. In contrast, an application of the results featured in \cite{Newey1994} Chapter 6 requires the number of moment conditions to equal the number of parameters of interest.
			
			\begin{lemma}[Identification, two-step]
				\label{lemma:NM_id_modified}
				If $W$ is positive semi-definite and, for $g_0(\boldsymbol{\theta}) = \EE[g(\V{z}, \boldsymbol{\theta}; \boldsymbol{\gamma}_0)]$, $g_0(\boldsymbol{\theta}_0) = 0$, and $W g_0(\boldsymbol{\theta}) \neq 0$ for $\boldsymbol{\theta} \neq \boldsymbol{\theta}_0$ then $Q_0(\boldsymbol{\theta}) = -g_0(\boldsymbol{\theta})' W g_0(\boldsymbol{\theta})$ has a unique maximum at $\boldsymbol{\theta}_0$.
			\end{lemma}
			\begin{proof}
				Apply an identical argument as the one used for the proof of NM Lemma 2.3. 
				
				Let $R$ be such that $R'R = W$. If $\theta \neq \boldsymbol{\theta}_0$, then $0 \neq W g_0(\boldsymbol{\theta}) = R'R g_0(\boldsymbol{\theta})$ implies $Rg_0(\boldsymbol{\theta}) \neq 0$ and hence $Q_0(\boldsymbol{\theta})= -[R g_0(\boldsymbol{\theta})]' [Rg_0(\boldsymbol{\theta})] < Q_0(\boldsymbol{\theta}_0) = 0$ for $\boldsymbol{\theta} \neq \boldsymbol{\theta}_0$.
			\end{proof}	
			
			\begin{theorem}[Consistency, two-step]
				\label{thm:NM_consistency_modified}
				Suppose that $\V{z}_i, (i=1,2,\dots)$, are i.i.d., $\widehat{W} \xrightarrow{p} W$, $\widehat{\boldsymbol{\gamma}} \xrightarrow{p}\boldsymbol{ \gamma}_0$, and (i) $W$ is positive semidefinite and $W \EE[g(\V{z}, \boldsymbol{\theta}; \boldsymbol{\gamma}_0)] = 0$ only if $\boldsymbol{\theta} = \theta_0$; (ii) $\boldsymbol{\theta}_0 \times \boldsymbol{\gamma}_0 \in \boldsymbol{\Theta} \times \boldsymbol{\Gamma}$ compact; (iii) $g(\V{z}, \boldsymbol{\theta}; \boldsymbol{\gamma})$ is continuous at each $\boldsymbol{\theta} \times \boldsymbol{\gamma} \in \boldsymbol{\Theta} \times \boldsymbol{\Gamma}$ with probability one; (iv) $\EE[\sup_{\boldsymbol{\theta}, \boldsymbol{\gamma}} \| g(\V{z}, \boldsymbol{\theta}; \boldsymbol{\gamma})\|] < \infty$. Then $\widehat{\boldsymbol{\theta}} \xrightarrow{p} \boldsymbol{\theta}_0$.
			\end{theorem}	
			\begin{proof}
				Proceed by verifying the hypotheses of Theorem 2.1 of Newey McFadden (abbreviated as NM2.1):
				\begin{itemize}
					\item NM2.1(i): Follows from Lemma \ref{lemma:NM_id_modified} and Condition (i).
					\item NM2.1(ii): Follows from Condition (ii).
					\item NM2.1(iii): Newey and McFadden Lemma 2.4, Condition (iii), and Condition (iv) imply that $\sup_{ \boldsymbol{\theta}} \| \frac{1}{n}\sum_{i=1}^n g(\V{z}_i, \boldsymbol{\theta}; \widehat{\boldsymbol{\gamma}})  - \EE[g(\V{z}, \boldsymbol{\theta}; \widehat{\boldsymbol{\gamma}})]\| \xrightarrow{p} 0$ and $\EE[g(\V{z}, \boldsymbol{\theta}; \boldsymbol{\gamma}_0)]$ is continuous in $\theta$. Moreover, $\widehat{\boldsymbol{\gamma}} \xrightarrow{p} \boldsymbol{\gamma}_0$ implies  $\sup_{ \boldsymbol{\theta}} \| \frac{1}{n}\sum_{i=1}^n g(\V{z}_i, \boldsymbol{\theta}; \widehat{\boldsymbol{\gamma}})  - \EE[g(\V{z}, \boldsymbol{\theta}; \boldsymbol{\gamma}_0)]\| \xrightarrow{p} 0$. Hence, 2.1(iii) holds because $Q_0(\theta) = - \EE[g(\V{z}, \boldsymbol{\theta}; \boldsymbol{\gamma}_0)]'W\EE[g(\V{z}, \boldsymbol{\theta}; \boldsymbol{\gamma}_0)] $ is continuous. 
					\item NM2.1(iv): For uniform convergence of the objective function $\widehat{Q}_n(\boldsymbol{\theta})$ to $Q_0(\boldsymbol{\theta})$, follow similar computations to Newey and McFadden (1994) Theorem 2.6:
					{\small		\begin{align*}
							|\widehat{Q}_n( \boldsymbol{\theta}) - Q_0(\boldsymbol{ \theta})| &\leq \| \frac{1}{n}\sum_{i=1}^ng(\textbf{z}_i, \boldsymbol{\theta}; \widehat{\boldsymbol{\gamma}}) -  \EE[g(\textbf{z}, \boldsymbol{\theta}; \boldsymbol{\gamma}_0)]\|^2 \| \widehat{W}\| + \\
							&\quad2 \| \frac{1}{n}\sum_{i=1}^ng(\textbf{z}_i, \boldsymbol{\theta}; \widehat{\boldsymbol{\gamma}}) -  \EE[g(\textbf{z}, \boldsymbol{\theta}; \boldsymbol{\gamma}_0)]\| \| \frac{1}{n}\sum_{n=1}^N g(\textbf{z}_t, \boldsymbol{\theta}; \widehat{\boldsymbol{\gamma}}) -  \EE[g(\textbf{z}, \boldsymbol{\theta}; \boldsymbol{\gamma}_0)]\| \|\widehat{W}\| +\\
							&\quad \| \EE[g(\textbf{z}, \boldsymbol{\theta}; \boldsymbol{\gamma}_0)]\|^2 \| |\widehat{W}-W\|^2.
					\end{align*}}
					
				\end{itemize}
			\end{proof}
			
			\begin{theorem}[NM (1994) Theorem 3.4, modified]
				\label{thm:NM_asymptoticnormality_modified}
				Suppose the hypotheses of Theorem \ref{thm:NM_consistency_modified} are satisfied. Also assume $\sqrt{n}(\widehat{\boldsymbol{\gamma}} - \boldsymbol{\gamma}_0) = O_p(1)$, $\EE[\nabla_{\boldsymbol{\gamma}} g(\V{z}, \boldsymbol{\theta_0}; \boldsymbol{\gamma}_0)] = \boldsymbol{0}$, (i) $\boldsymbol{\theta}_0 \times \boldsymbol{\gamma}_0 \in \mathrm{interior}(\boldsymbol{\Theta} \times \boldsymbol{\Gamma})$; (ii) $g(\V{z}, \boldsymbol{\theta}; \boldsymbol{\gamma})$ is continuously differentiable in a neighborhood $\nn$ of $\boldsymbol{\theta}_0 \times \boldsymbol{\gamma}_0$ with probability approaching 1; (iii) $\EE[g(\V{z}, \boldsymbol{\theta}_0; \boldsymbol{\gamma}_0)] = 0$ and $\EE[\|g(\V{z}, \boldsymbol{\theta}_0; \boldsymbol{\gamma}_0) \| ^2] < \infty$, (iv) $\EE[\sup_{\boldsymbol{\theta} \times \boldsymbol{\gamma}} \| \nabla_{\boldsymbol{\theta} \times \boldsymbol{\gamma}} g(\V{z}, \boldsymbol{\theta}, \boldsymbol{\gamma}) \|] < \infty$; (v) $G'WG$ is non-singular for $G = \EE[\nabla_{\boldsymbol{\theta}} g(\V{z}, \boldsymbol{\theta}_0; \boldsymbol{\gamma}_0)]$. Then, for $\Sigma = \EE[g(\V{z}, \boldsymbol{\theta}_0; \boldsymbol{\gamma}_0)g(\V{z}, \boldsymbol{\theta}_0; \boldsymbol{\gamma}_0)']$,
				\begin{align*}
					\sqrt{n} (\widehat{\boldsymbol{\theta}} - \boldsymbol{\theta}_0) \xrightarrow{d} \nn(0, (G'WG)^{-1} G'W \Sigma WG(G'WG)^{-1} ).
				\end{align*}
			\end{theorem}
			\begin{proof}
				Conditions (i), (ii), and (iii) imply that the first order condition is satisfied with probability approaching one: $2 \widehat{G}_n (\widehat{\boldsymbol{\theta}})'\widehat{W}\widehat{g}_n(\widehat{\boldsymbol{\theta}}) = 0$ for $\widehat{G}_n(\widehat{\boldsymbol{\theta}}) = \nabla_{\boldsymbol{\theta}} \frac{1}{n}\sum_{i=1}^n g(\V{z}_i, \widehat{\boldsymbol{\theta}}, \widehat{\boldsymbol{\gamma}})$. Expanding the FOC about $\boldsymbol{\theta}_0$ and $\boldsymbol{\gamma}_0$ and rearranging gives
				\begin{align*}
					\sqrt{n}(\widehat{\boldsymbol{\theta}} - \boldsymbol{\theta}_0) &= -(\widehat{G}_n(\widehat{\boldsymbol{\theta}})' \widehat{W}\widehat{G}_n(\overline{\boldsymbol{\theta}}))^{-1}\widehat{G}_n(\widehat{\boldsymbol{\theta}}) \widehat{W} \frac{\sqrt{n}}{n} \sum_{i=1}^n g(\V{z}_i, \boldsymbol{\theta}_0; \widehat{\boldsymbol{\gamma}}) \\
					&= -(\widehat{G}_n(\widehat{\boldsymbol{\theta}})' \widehat{W}\widehat{G}_n(\overline{\boldsymbol{\theta}}))^{-1}\widehat{G}_n(\widehat{\boldsymbol{\theta}}) \widehat{W} \frac{\sqrt{n}}{n} \sum_{i=1}^n g(\V{z}_i, \boldsymbol{\theta}_0; \boldsymbol{\gamma}_0) \\
					&\quad -(\widehat{G}_n(\widehat{\boldsymbol{\theta}})' \widehat{W}\widehat{G}_n(\overline{\boldsymbol{\theta}}))^{-1}\widehat{G}_n(\widehat{\boldsymbol{\theta}}) \widehat{W} \frac{1}{n} \sum_{i=1}^n \nabla_{\boldsymbol{\gamma}} g(\V{z}_i, \boldsymbol{\theta}_0; \overline{\boldsymbol{\gamma}}) \sqrt{n}(\widehat{\boldsymbol{\gamma}} - \boldsymbol{\gamma}_0) \\
					&\quad\xrightarrow{p} \nn(0, (G'WG)^{-1} G'W \Sigma WG(G'WG)^{-1} )
				\end{align*}
				where $\boldsymbol{\overline{\theta}}$ and $\boldsymbol{\overline{\gamma}}$ are intermediate values. In the above display, $\widehat{\boldsymbol{\gamma}} \xrightarrow{p} \boldsymbol{\gamma}_0$, $\widehat{\boldsymbol{\theta}}\xrightarrow{p} \boldsymbol{\theta}_0$, and condition (iv) imply $\widehat{G}_n(\widehat{\theta}) \xrightarrow{p} G$, $\widehat{G}_n(\overline{\theta}) \xrightarrow{p} G$, and $\frac{1}{n} \sum_{i=1}^n \nabla_{\boldsymbol{\gamma}} g(\V{z}_i, \boldsymbol{\theta}_0; \overline{\boldsymbol{\gamma}}) \xrightarrow{p} \boldsymbol{0}$. The final line follows from the Slutsky theorem and an i.i.d. central limit theorem.
				
			\end{proof}

			\subsection{RGIV when the shock factor is determined by known observables}
			\label{sec:factor_known_observables}
			In this section, I show that differencing RGIV moments can accommodate the case where the shock factor structure is entirely determined by known observables. I also show that the identification reduces to an overdetermined system of nonlinear equations, which admits a unique solution (with the parameter space restriction $\phi_S<1$) outside knife-edge cases.

			\subsubsection{Setup}
			Take Assumption \ref{ass:general} and augment with latent factors $\boldsymbol{\eta}_t$ ($m \times 1$) where $\EE[\boldsymbol{\eta}_t \boldsymbol{\eta}_t']  =\mathbf{I}$ (and finite fourth moments) and loadings $\boldsymbol{\lambda}_i$ ($m \times 1$), which are entirely determined by observables $\mathbf{x}_i$ ($m \times 1$) through $\boldsymbol{\Pi}$ ($m \times m$)
			\begin{align*}
				r_{it} &= \phi_i r_{St} + \underbrace{\boldsymbol{\lambda}_i' \boldsymbol{\eta}_t + u_{it}}_{v_{it}}, \quad	\boldsymbol{\lambda}_i = \boldsymbol{\Pi} \mathbf{x}_i.
			\end{align*}
			Also suppose there are more GMM moment conditions than there are unknowns ($n(n-1)/2-m(m+1)/2 \geq n$). $\boldsymbol{\Pi}$ ($m \times m$) maps the unit-specific observables to the loadings and is unknown to the researcher. In matrix form, for $	\mathbf{X} = \begin{bmatrix}
				\mathbf{x}_1 & \mathbf{x}_2 & \dots &		\mathbf{x}_n
			\end{bmatrix}$ and $ \boldsymbol{\Lambda} = \begin{bmatrix}
				\boldsymbol{\lambda}_1 & \boldsymbol{\lambda}_2 & \dots & \boldsymbol{\lambda}_n
			\end{bmatrix}$, 
			\begin{align*}
				\boldsymbol{\Lambda} = \boldsymbol{\Pi} \mathbf{X} \text{, } \mathbf{r}_t = \boldsymbol{\phi}r_{St} +\underbrace{ \boldsymbol{\Lambda}' \boldsymbol{\eta}_t + \mathbf{u}_t}_{\mathbf{v}_t}, \text{ and } \textrm{rank}(\mathbf{X}) = m.
			\end{align*}

			The shock covariance matrix is $\EE[\mathbf{v}_t\mathbf{v}_t' ] = \mathbf{X}' \boldsymbol{\Pi}' \boldsymbol{\Pi} \mathbf{X} + \mathbf{D}$ for diagonal matrix $\mathbf{D}$ with diagonal entries $\sigma^2_1, \sigma^2_2,\dots, \sigma^2_n$. Vectorizing and selecting the off-diagonal elements of the covariance matrix with selection matrix $\mathbf{S}$ ($n(n-1)/2 \times n^2$), 
			\begin{align*}
				\mathbf{S} \textrm{vec}\left({\EE[\mathbf{v}_t \mathbf{v}_t']}\right) = 	\mathbf{S} \left( \mathbf{X}' \otimes \mathbf{X}' \right) \mathrm{vec}\left( \boldsymbol{\Pi}'\boldsymbol{\Pi}
				\right).
			\end{align*}
			Premultiplying the above display by the annihilator matrix (for arbitrary matrix $W$, defined as $\mathbf{W}$ $M_\mathbf{W} \equiv \mathbf{I} - \mathbf{W}(\mathbf{W}'\mathbf{W})^{-1}\mathbf{W}'$),  
			\begin{align}
				\mathbf{M}_{\mathbf{S} \left( \mathbf{X}' \otimes \mathbf{X}' \right) }\mathbf{S} \textrm{vec}\left({\EE[\mathbf{v}_t \mathbf{v}_t']}\right) = \mathbf{0}. \label{eq:lincomb}
			\end{align}
			Hence, a linear combination of the off-diagonal elements of $\EE[\mathbf{v}_t \mathbf{v}_t']$ equals zero.

			\subsubsection{Identification}
			Equation \ref{eq:lincomb} sets a linear combination of RGIV moments to zero. Below, I show that the spillover coefficients are generally identified outside a set of knife-edge DGPs.  Equation \ref{eq:lincomb} implies a system of (univariate) quadratic equations in the errors of a particular spillover coefficient. For there to exist a solution other than the true solution, these quadratic equations must have identical zeros---which isn't true in general.
			
			For true spillover coefficients $\phi_i$ and proposed solutions $\widetilde{\phi}_i$, let $\check{\phi}_i \equiv \phi_i - \widetilde{\phi}_i$ give the error in spillover coefficient $\phi_i$. The moment condition corresponding to the idiosyncratic shocks of units $i$ and $j$ can be written as 
			\begin{align*}
				\EE\left[(r_{it} - \widetilde{\phi}_i r_{St})(r_{jt} - \widetilde{\phi}_j r_{St}) \right] = \check{\phi}_i \check{\phi}_j k + \check{\phi}_i k_j + \check{\phi}_j k_i + \lambda_i \lambda_j
			\end{align*}
			for $k \equiv\EE[r_{St}^2]$ and $k_i \equiv \EE[r_{St} (u_{it}+ \boldsymbol{\lambda}_i' \boldsymbol{\eta}_t)]$. For $i \neq j$, the above display's moments can be stored in an $n(n-1)/2 \times 1$ vector $\widetilde{\boldsymbol{g}}$ matching the ordering of $\mathbf{S} \textrm{vec}\left({\EE[\mathbf{v}_t \mathbf{v}_t']}\right)$ in Equation \ref{eq:lincomb}. Hence, identification of the moments characterized by Equation \ref{eq:lincomb} reduces to studying the roots of $\mathbf{M}_{\mathbf{S} \left( \mathbf{X}' \otimes \mathbf{X}' \right)} \widetilde{\mathbf{g}} = 0$.

			Since $\textrm{rank}(\mathbf{M}_{\mathbf{S} \left( \mathbf{X}' \otimes \mathbf{X}' \right) }) = n(n-1)/2 - m(m+1)/2$, $\mathbf{M}_{\mathbf{S} \left( \mathbf{X}' \otimes \mathbf{X}' \right) }$ can be expressed in reduced row echelon form (after elementary row operations) 
			\begin{align*}
				\begin{bmatrix}
					\mathbf{I}_{n(n-1)/2 - m^2} & -\mathbf{A} \\
					\mathbf{0}_{m^2 \times (n(n-1)/2 - m^2)} & \mathbf{0}_{m^2 \times m^2}
				\end{bmatrix}.
			\end{align*}
			Letting $\widetilde{\boldsymbol{g}}^I$ be the first $n(n-1)/2-m(m+1)/2$ moments and $\widetilde{\boldsymbol{g}}^{II}$ be the last $m(m+1)/2$ moments, $\mathbf{M}_{\mathbf{S} \left( \mathbf{X}' \otimes \mathbf{X}' \right)} \widetilde{\mathbf{g}} = 0$ can equivalently be expressed as
			\begin{align}
				\widetilde{\mathbf{g}}^I = \mathbf{A} \widetilde{\mathbf{g}}^{II}. \label{eq:lincomb2}
			\end{align}
			The moments $\widetilde{\mathbf{g}}^{I}$ equal some (known) linear combination of $\widetilde{\mathbf{g}}^{II}$ determined by matrix $\mathbf{A}$.
			
			To characterize the solutions, fix $\widetilde{\mathbf{g}}^{II}$ and let $\mathbf{a}\equiv \mathbf{A} \widetilde{\mathbf{g}}^{II}$ where $a_{ij}$ is the entry of vector $\mathbf{a}$ that corresponds to the moment condition for shocks $i\neq j$. Then the row of Equation \ref{eq:lincomb2} corresponding to the shocks of units $i$ and $j$ is $\check{\phi}_i(k_j + 	\check{\phi}_j k) + 	\check{\phi}_j k_i = a_{ij}.$
			
			After rearranging and setting unit 1 as the reference unit, $\check{\phi}_i = \frac{a_{i1} - \check{\phi}_1 k_i}{k_1 + \check{\phi}_1 k} \text{ and } \check{\phi}_j = \frac{a_{j1} - \check{\phi}_1 k_j}{k_1 + \check{\phi}_1 k}.$ Substituting $\check{\phi}_i$ and $\check{\phi}_j$,
			\begin{align*}
				a_{ij}=	 \check{\phi}_i(k_j + 	\check{\phi}_j k) + 	\check{\phi}_j k_i = \frac{a_{i1} - \check{\phi}_1 k_i}{k_1 + \check{\phi}_1 k} k_j +\frac{a_{j1} - \check{\phi}_1 k_j}{k_1 + \check{\phi}_1 k}k_i + \frac{(a_{i1} - \check{\phi}_1 k_i) (a_{j1} - \check{\phi}_1 k_j)  }{(k_1 + \check{\phi}_1 k)^2}k.
			\end{align*}
			
			Rearranging yields a quadratic equation in $\check{\phi}_1$
			\begin{align}
				0=\alpha_{1}^{(ij)} \check{\phi}_1^2 + 	\alpha_{2}^{(ij)} \check{\phi}_1 + 	\alpha_{3}^{(ij)} \label{eq:sys}
			\end{align}
			where $\alpha_1^{(ij)}\equiv a_{ij} k^2 + k_i k_j k$, $\alpha_2^{(ij)}\equiv 2a_{ij}k_1 k + 2 k_ik_jk_1$, and  $\alpha_3^{(ij)}\equiv a_{ij} k_1^2 -(a_{i1} k_j + a_{j1}k_i)k_1 - a_{i1}a_{j1}k$. Hence, the above display gives a system of quadratic equations (in $\check{\phi}_1$) defined by the pairs $i\neq j$ contained in moments $\widetilde{\mathbf{g}}^I$.
			
			If $\widetilde{\mathbf{g}}^{II} = \mathbf{0}$, then $\mathbf{a} = \mathbf{0}$, $\alpha_1^{(ij)} = k_i k_j k$, $\alpha_2^{(ij)} = 2 k_i k_j k$, and $\alpha_3^{(ij)} =0$. Here, $0 = \check{\phi}_1 (k \check{\phi}_1 + 2k_1)$. There are two solutions, which both do not depend on the index $i,j$: (1) the true solution $\check{\phi}_1 = 0$ and (2) a false solution $\check{\phi}_1 = -2\frac{k_1}{k}$ (which can be ruled out from the condition $\phi_S < 1$ as is done in Lemma \ref{lemma:gmm_id}).
			
			Now instead suppose $\widetilde{\mathbf{g}}^{II} \neq \mathbf{0}$. From the quadratic formula, the roots of Equation \ref{eq:sys} (dividing by $\alpha_1^{(ij)}$) are determined by $\frac{\alpha_2^{(ij)}}{\alpha_1^{(ij)}}$ and $\frac{\alpha_3^{(ij)}}{\alpha_1^{(ij)}}$.  $\frac{\alpha_3^{(ij)}}{\alpha_1^{(ij)}}$ in particular depends on the units $i,j$ through terms $k_i$ and $k_j$. Hence, besides knife-edge cases, it is not generally true that the roots for Equation \ref{eq:sys} coincide when $\widetilde{\mathbf{g}}^{II} \neq \mathbf{0}$.
			
			\subsection{Spillovers application}
			\subsubsection{Robustness}
			\label{sec:application_robustness}

			Table \ref{table:application_winsor} reports coefficients and standard errors under alternative winsorization schemes. Columns 1-3 show that the results are qualitatively similar to the preferred specification after winsorizing at 5\%, 0.5\%, and 10\% respectively. The fourth column, however, shows that estimation without winsorization is unreliable. These results show that the qualitative features of the main text's preferred specification aren't specific to the choice of 1\% winsorization and that the linear form of the model described in Assumption \ref{ass:general_reducedform} is sensitive to extreme outliers.
			
			\begin{table}[!htbp]
				\centering
				\begin{tabular}{l|cccc}
					\toprule
					& 5\%  & 0.5\% & 10\%  & 0\% \\
					\hline
					RGIV results: &&&&\\
					\quad $\phi_S$ & 0.57 & 0.54 & 0.58 & -16.57 \\
					& (0.06) & (0.08) & (0.05) & (10275.21) \\
					\quad $\phi_\text{IRL, PRT, ESP}$ & 0.82 & 0.82 & 0.83 & 0.98 \\
					& (0.05) & (0.06) & (0.05) & (2.38) \\
					\quad $\phi_\text{GRC, ITA}$ & 0.49 & 0.45 & 0.5 & -28.08 \\
					& (0.12) & (0.16) & (0.11) & (17084.06) \\
					\quad $\phi_\text{Core}$ & 0.42 & 0.38 & 0.42 & 0.47 \\
					& (0.03) & (0.03) & (0.03) & (2.48) \\
					\quad $\phi_\text{SVN}$ & 0.33 & 0.33 & 0.34 & 0.34 \\
					& (0.05) & (0.06) & (0.05) & (0.39) \\
					\hline
					Tests ($p$ values): &&&&\\
					\quad Specification & 0.817 & 0.872 & 0.679 & 0.616 \\
					\quad Homogeneity & 0 & 0 & 0 & 0.275 \\
					\bottomrule
				\end{tabular}
				\caption{Estimation results. Coefficient estimates for $\phi_S, \phi_\text{IRL,PRT,ESP}, \phi_\text{GRC, ITA}, \phi_\text{Core}$, and $\phi_\text{SVN}$ are listed above standard errors, which are provided in parentheses. $p$ values are provided in the bottom section of the table for the specification test and parameter homogeneity tests. Columns 1--4 give estimation results for winsorization at the 2.5 and 97.5 percentiles, 0.25 and 99.75 percentiles, 5 and 95 percentiles, and no winsorization respectively.}
				\label{table:application_winsor}
			\end{table}

			Table \ref{table:application_altgroups} summarizes RGIV results under alternative blocking schemes. Starting with the preferred specification, the first column shifts Ireland from the Portugal/Spain periphery block to the Greece/Italy periphery block (splitting the periphery countries into an ``Iberia'' block and ``Other'' block). The specification test's $p$ value of 0.95 suggests that the ``Iberia'' specification's moments are consistent with the data. Moreover, the spillover coefficient for the Portugal/Spain block is 0.89 compared to 0.83 for the Portugal/Spain/Ireland block in the preferred specification, suggesting that Ireland's spillover coefficient is lower than the size-weighted combination of Portugal/Spain's. Columns 2 and 3 show that spillover coefficients are nearly unchanged after moving Slovenia to the core and periphery block respectively (no specification test is reported since the model is just-identified). The specification in Column 4 separates France from the core countries and moves it into its own group. The specification test is rejected, giving indirect evidence that French shocks are correlated with those of other Euro area core countries. The final column gives the size-weighted spillover coefficient for a specification where each country is its own block. The specification test is again rejected, suggesting that the estimated idiosyncratic shocks are in fact correlated.

			\begin{table}[!htbp]
				\centering
				\begin{tabular}{l|cccccc}
					\toprule 
					& Iberia & Core SVN & Periphery SVN & Separate FRA& Country-level \\
					\hline
					RGIV results: &&&&\\
					\quad $\phi_S$ & 0.51 & 0.54 & 0.54 & 0.52 & -0.16 \\
					& (0.11) & (0.08) & (0.08) & (0.09) & (2.58)\\
					\quad $\phi_{I}$ & 0.89 & 0.82 & 0.83 & 0.83 &-- \\
					& (0.06) & (0.06) & (0.06) & (0.06)& -- \\
					\quad $\phi_\text{II}$ & 0.39 & 0.45 & 0.44 & 0.42&-- \\
					& (0.2) & (0.16) & (0.16) & (0.18)&-- \\
					\quad $\phi_\text{III}$ & 0.4 & 0.39 & 0.39 & 0.4&-- \\
					& (0.04) & (0.03) & (0.03) & (0.03)&-- \\
					$\quad \phi_\text{IV}$ & 0.33 &  &  & 0.38&-- \\
					& (0.06) & -- & -- & (0.04)&-- \\
					\hline
					Tests ($p$ values): &&&&&\\
					\quad Specification & 0.95 & -- & -- & <0.001 & <0.001 \\
					\quad Homogeneity & <.001 & <.001 & <.001 & <0.001& <0.001\\
					\bottomrule
				\end{tabular}
				\caption{Estimation results for alternative groupings. Coefficient estimates for $\phi_S, \phi_\text{I}, \phi_\text{II}, \phi_\text{III}$, and $\phi_\text{IV}$ are listed above standard errors, which are provided in parentheses. $p$ values are provided in the bottom section of the table for the specification test and parameter homogeneity tests. Groups: ``Iberia'' (I: PRT, ESP; II: GRC, ITA, IRL; III: AUT, BEL, FRA, FIN, NLD; IV: SVN), ``Core SVN'' (I: PRT, ESP; II: GRC, ITA, IRL; III: AUT, BEL, FRA, FIN, NLD, SVN), ``Periphery SVN'' (I: PRT, ESP,  IRL; II: GRC, ITA, SVN; III: AUT, BEL, FRA, FIN, NLD), ``Separate FRA'' (I: PRT, ESP, IRL; II: GRC, ITA; III: AUT, BEL, FIN, NLD, SVN; IV: FRA). ``Country-level'' (each country is its own block).}
				\label{table:application_altgroups}
			\end{table}
			
			\begin{table}[!htbp]
				\centering
				\begin{tabular}{l|cccc}
					\toprule
					& Andrews rule & No Fama-French & 0-factor GIV & 2-factor GIV\\
					\hline
					RGIV results: &&&\\
					\quad$\phi_S$ & 0.54 & 0.54& \\
					& (0.07) & (0.08) &&\\
					\quad$\phi_\text{PRT, ESP, IRL}$ & 0.83 & 0.83 &\\
					& (0.05) & (0.06) &&\\
					\quad$\phi_\text{GRC, ITA}$ & 0.44 & 0.45 &\\
					& (0.14) & (0.17)&& \\
					\quad$\phi_\text{Core}$ & 0.39 & 0.4 &\\
					& (0.03) & (0.03) \\
					\quad $\phi_\text{SVN}$ & 0.33 & 0.34 &\\
					& (0.05) & (0.06) &&\\
					$\phi^{GIV}$ && &0.42&0.31\\
					&&& (0.03) & (0.06)\\
					\hline
					Tests ($p$ values): &&&&\\
					\quad Specification & 0.93 & 0.952 &&\\
					\quad Homogeneity & <.001 & <.001 &&\\
					First stage $F$-statistic &&&411&112\\
					\bottomrule
				\end{tabular}
				\caption{Estimation results. Coefficient estimates for $\phi_S, \phi_\text{IRL,PRT,ESP}, \phi_\text{GRC, ITA}, \phi_\text{Core}$, and $\phi_\text{SVN}$ are listed above standard errors, which are provided in parentheses. $p$ values are provided in the bottom section of the table for the specification test and parameter homogeneity tests.}
			\end{table}
			
\clearpage

\subsubsection{Narrative}

Table \ref{table:top10_narrative} gives the ten largest (by absolute value) idiosyncratic shock contributions to movements in the size-weighted average yield $S_i u_{it}/(1-\phi_S)$. All ten events correspond to shocks to the Greece/Italy periphery block.

{\footnotesize
\makeatletter\let\LT@end@hline\relax\makeatother
\begin{longtable}{%
  l
  r
  >{\RaggedRight\arraybackslash}p{7.5cm}
  >{\RaggedRight\arraybackslash}p{4.2cm}
}
\toprule
\textbf{Date} & \textbf{Value} & \textbf{Description} & \textbf{Sources} \\
\midrule
\endfirsthead
\toprule
\textbf{Date} & \textbf{Value} & \textbf{Description} & \textbf{Sources} \\
\midrule
\endhead
\endfoot
\endlastfoot

Apr.\ 26, 2010 & $+0.054$ &
Three days after Greece requested EU financial assistance, Chancellor Angela Merkel publicly conditioned aid on the release of a ``sustainable, credible'' deficit-reduction plan. &
\href{https://www.bloomberg.com/news/articles/2010-04-26/merkel-says-greece-must-present-sustainable-cuts-plan-to-get-german-aid}%
{\textit{Merkel Says Greece Must Present Sustainable Cuts Plan to Get German Aid}}
(Apr.\ 26, 2010, Bloomberg)
\\[4pt]

Aug.\ 8, 2011 & $-0.065$ &
First trading day after the ECB announced on August 7 an
expansion of the Securities Markets Programme to include Italian and
Spanish debt. &
\href{https://www.bloomberg.com/news/articles/2011-08-08/bonds-rally-in-italy-spain-to-narrow-spread-to-bunds-after-ecb-buy-signal}%
{\textit{Bonds Rally in Italy, Spain to Narrow Spread to Bunds After ECB Buy Signal}}
(Aug.\ 8, 2011, Bloomberg)
\\[4pt]

May 15, 2014 & $+0.053$ &
Greek yields increased as opinion polls suggested declining support for the government-led coalition ahead of the May 2014 European Parliament election. &
\href{https://www.bloomberg.com/news/articles/2014-05-15/euro-area-bonds-gain-as-data-push-yields-to-records}%
{\textit{Euro-Area Bonds Gain as Data Push Yields to Records}}
(May 15, 2014, Bloomberg)
\\[4pt]

Oct.\ 15, 2014 & $+0.055$ &
Greek yields rose sharply on concerns about the Samaras government's plan to exit its bailout programme early without a credible financing arrangement. &
\href{https://www.bloomberg.com/news/articles/2014-10-15/greek-bonds-extend-selloff-pushing-yield-up-most-in-15-months}%
{\textit{Greek Bonds Extend Selloff Pushing Yield Up Most in 15 Months}}
(Oct.\ 15, 2014, Bloomberg)
\\[4pt]

Dec.\ 29, 2014 & $+0.058$ &
The Greek parliament failed to elect a president in the constitutionally mandated third round of voting. This gave rise to mandatory snap elections and created uncertainty over Greece's bailout programme. &
\href{https://www.bloomberg.com/news/articles/2014-12-29/greece-faces-snap-election-as-samaras-fails-to-install-president}%
{\textit{Greece Faces Snap Election as Samaras Fails to Install President}}
(Dec.\ 29, 2014, Bloomberg)
\\[4pt]

Jul.\ 6, 2015 & $+0.068$ &
First trading day following the July 5 Greek bailout referendum. Here 61\% of Greek voters rejected the creditors' terms. &
\href{https://www.bloomberg.com/news/articles/2015-07-05/greece-heads-for-no-vote-raising-risk-of-departure-from-euro}%
{\textit{Greece Heads for `No' Vote Raising Risk of Departure From Euro}}
(Jul.\ 5, 2015, Bloomberg)
\\[4pt]

May 16, 2018 & $+0.063$ &
On the evening of May~15, \textit{Huffington Post Italy} leaked a 39-page draft coalition programme in which the Five Star Movement and the League called for the ECB to cancel 250 billion Euro of Italian sovereign debt and proposed a mechanism for Euro exit.  &
\href{https://www.bloomberg.com/news/articles/2018-05-16/italy-bonds-roiled-as-populists-debate-300-billion-write-down}%
{\textit{Italy Bonds Roiled as Populists Debate 300 Billion Write-Down}}
(May 16, 2018, Bloomberg)
\\[4pt]

May 21, 2018 & $+0.054$ &
The Five Star Movement and the League formally proposed Giuseppe Conte as their candidate for prime minister, signalling that government formation was advancing. &
\href{https://www.bloomberg.com/news/articles/2018-05-21/italian-bonds-extend-slide-amid-concern-over-increased-borrowing}%
{\textit{Italian Bonds Extend Slide Amid Concern Over Increased Borrowing}}
(May 21, 2018, Bloomberg)
\\[4pt]

May 28, 2018 & $+0.054$ &
Over the weekend of May~26--27, President Mattarella vetoed the coalition's proposed appointment of Paolo Savona as economy minister and designated Carlo Cottarelli as prime minister; Luigi Di~Maio of the Five Star Movement called for Mattarella's impeachment. &
\href{https://www.bloomberg.com/news/articles/2018-05-28/italian-markets-jump-as-president-vetoes-savona-as-finance-chief}%
{\textit{Italian Markets Jump as President Vetoes Savona as Finance Chief}}
(May 28, 2018, Bloomberg)
\\[4pt]

May 29, 2018 & $+0.055$ &
The Italian political crisis remained unresolved as interim Prime Minister Cottarelli failed to announce a cabinet. &
\href{https://www.bloomberg.com/news/articles/2018-05-29/italian-two-year-bond-yields-climb-to-highest-level-since-2013}%
{\textit{Italian Two-Year Bond Yields Climb to Highest Level Since 2013}}
(May 29, 2018, Bloomberg)
\\[4pt]

\bottomrule
\caption{Ten largest (by absolute value) idiosyncratic shock contributions to movements in the size-weighted average yield $S_i u_{it}/(1-\phi_S)$. All ten events correspond to shocks to the Greece/Italy periphery block. The first column gives the date of the event. The second column gives the shock contribution. The third column gives a narrative description of events that occurred on or around the given date. The fourth column includes the title, date, and hyperlink for the corresponding Bloomberg News article.}
\label{table:top10_narrative}
\end{longtable}}

\clearpage
\subsection{Additional results on the application to the inelastic markets hypothesis}

\subsubsection{Data appendix}
\label{sec:demand_data}

The following control variables are from FRED (vintage June 6, 2019) with the mnemonics and transformations in parentheses: real GDP (GDPC1, quarterly percent change), 2Y Treasury (DGS2, quarterly first difference), 5Y Treasury (DGS5, quarterly first difference), 10Y Treasury (DGS10, quarterly first difference), 30Y Treasury (DGS30, quarterly first difference), Ted Spread (TEDRATE, quarterly first difference), University of Michigan Consumer Sentiment Index (UMCSENT, quarterly first difference), Nominal Goods-Only Major Currencies U.S. Dollar Index (DTWEXM, quarterly percent change) and WTI Spot Crude Oil Price (WTISPLC, quarterly percent change).

Following \citet{Gabaix2023}, the following control variables are constructed by running Fama-MacBeth regressions on the data of \citet{Jensen2023} and enter in levels: log market cap, log book-to-market ratio, and momentum. 

The credit spread of \citet{Gilchrist2012} enters in levels and is available at the Federal Reserve website\footnote{\url{https://www.federalreserve.gov/econres/notes/feds-notes/updating-the-recession-risk-and-the-excess-bond-premium-20161006.html}}.

\subsubsection{Block and sector size over time}

Figure \ref{fig:size_blocks} gives block size over time for the preferred grouping described in Table \ref{table:ffunds_block_def}. Figure \ref{fig:size_within_blocks} gives the size of the sector relative to block size over time.

\begin{figure}[!htbp]
	\centering
	\includegraphics[width=.6\textwidth]{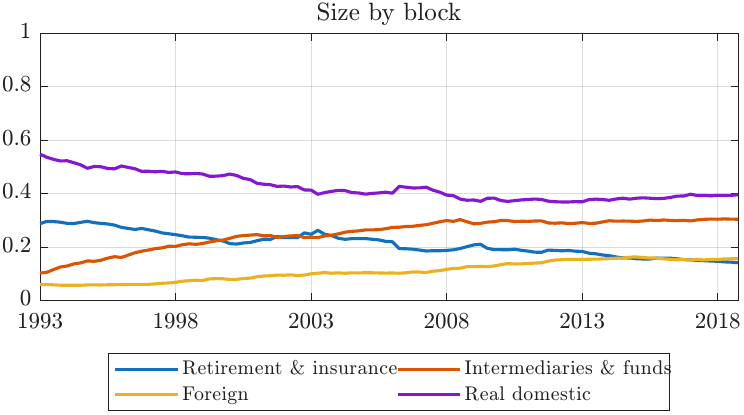}
	\caption{Block size over time. See Table \ref{table:ffunds_block_def} for a complete description of the block grouping.}
	\label{fig:size_blocks}
\end{figure}

\begin{figure}[!htbp]
	\centering
	\includegraphics[width=.75\textwidth]{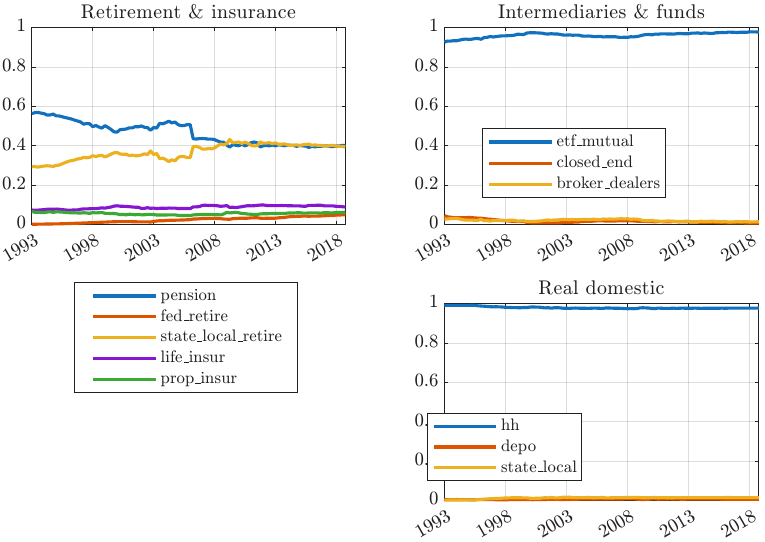}
	\caption{Size of sector relative to block size over time. The foreign block is omitted as the block has a single sector.}
	\label{fig:size_within_blocks}

\end{figure}

\subsubsection{Theoretical framework for the demand system with controls}
\label{sec:demand_framework}

Assumption \ref{ass:demand_controls} (below) establishes conditions for the demand system analog of Assumption \ref{ass:general_controls}. Note that these conditions are written for a generic outcome variable $y_{it}$ and endogenous variable $p_t$. We obtain Equation \ref{eq:demand_system_flowfunds} after substituting $\Delta q_{it}$ for $y_{it}$ and $\Delta p_t$ for $p_t$.

\begin{assumption}{(Demand system with controls)}
	
	\label{ass:demand_controls}
	\begin{enumerate}[label = (\roman*), ref= \theassumption(\roman*)]
		\item \textbf{Model}: For $n  \geq 2$ units, let fixed sizes $S_i \in (0,1)$ sum to 1.  Outcome $\V{y}_t =[y_{1t},...,y_{nt}]'$ responds to $p_t$ with coefficient $\phi_i$, coefficients $\boldsymbol{ \beta}_i$ ($k \times 1$), observed control variables $\V{x}_t$ ($k \times 1$), and unobserved shocks $\V{u}_{t} = [u_{1t}, u_{2t}, ...,u_{nt}]'$, \label{ass:demand_reducedform_controls} $y_{it}  = \phi_i p_t +  \boldsymbol{ \beta}_i' \V{x}_t + u_{it}$ for all $i = 1,...,n$. In addition, $y_{St} = \gamma p_t + \boldsymbol{\beta}_0' \mathbf{x}_t + \varepsilon_t$.

		\item \textbf{Moments}: For $\sigma^2_{i} > 0$ and $\sigma^2_\varepsilon$, shocks $(\V{u}_t', \varepsilon_t)'$ are i.i.d. with moments 
		\begin{align*}
			\EE((\V{u}_t', \varepsilon_t)') = 0, \quad \EE((\V{u}_t', \varepsilon_t) (\V{u}_t', \varepsilon_t)') =  \mathrm{diag}(\sigma^2_1,...,\sigma^2_n, \sigma^2_\varepsilon), \quad \EE(\| (\V{u}_t', \varepsilon_t)'\|^4) < \infty.
		\end{align*}
		Moreover, $u_{it} \indep u_{jt}$ for $i \neq j$, $u_{it} \indep \varepsilon_t$ for all $i$, $\EE(\V{x}_t\V{x}_t') = \Sigma_{XX}$ for positive definite $\Sigma_{XX}$, and $\V{x}_t \indep (\mathbf{u}_t', \varepsilon_t)'$. \label{ass:demand_shockmoments_controls}

		\item \textbf{Parameter space}: 
		For coefficients $\boldsymbol{\phi} = [\phi_1,\dots,\phi_n]'$,  store parameters in $\check{\boldsymbol{ \theta}} = [\boldsymbol{ \phi}', \gamma, \boldsymbol{ \beta}_1',\dots, \boldsymbol{ \beta}_n']'$. Then, the true parameter $\check{\boldsymbol{\theta}}_0$ is in the interior of parameter space $\check{\boldsymbol{\Theta}}$.  $\check{\boldsymbol{\Theta}}$ is compact and for any $\check{\boldsymbol{\widetilde{\theta}}}  \in \check{\boldsymbol{\Theta}}$, $\widetilde{\phi}_S   < \widetilde{\gamma}$. \label{ass:demand_parameterspace_controls}
	\end{enumerate}
\end{assumption}

Definition \ref{defn:demand_naive_controls} (below) defines the RGIV estimator applied to the demand system. 

\begin{definition}
	\label{defn:demand_naive_controls}
	Store data $\V{z}_t = [\V{y}_t', \V{x}_t', p_t]'$ for  outcome variable $\V{y}_t = [y_{1t}, \dots,y_{nt}]'$. Let $u_{i}(\V{z}_t, \boldsymbol{\phi}; \boldsymbol{\psi}) = (y_{it} - \boldsymbol{ \psi}_i' \V{x}_t)  - \phi_i (p_t -\boldsymbol{ \psi}_p' \mathbf{x}_{t} )$ for $i=1,...,n$ and $\varepsilon(\V{z}_t, \gamma; \boldsymbol{\psi}) = (y_{St} - \boldsymbol{ \psi}_S' \V{x}_t)  - \gamma (p_t -\boldsymbol{ \psi}_p' \mathbf{x}_{t} ) $.
	The moment function $g^c_{\mathrm{demand}}(\V{z}_t, \boldsymbol{\phi}; \boldsymbol{\psi})$  is
{\small
	\begin{align*}
		g^c_{\mathrm{demand}}(\V{z}_t, \boldsymbol{\phi}, \gamma; \boldsymbol{\psi}) &= [u_{1}(\V{z}_t, \boldsymbol{\phi}; \boldsymbol{\psi}) u_{2}(\V{z}_t, \boldsymbol{\phi}; \boldsymbol{\psi}), \dots, \, u_{1}(\V{z}_t, \boldsymbol{\phi}; \boldsymbol{\psi}) u_{n}(\V{z}_t, \boldsymbol{\phi}; \boldsymbol{\psi}),\\
		&\quad \, \text{ }u_{2}(\V{z}_t, \boldsymbol{\phi}; \boldsymbol{\psi}) u_{3}(\V{z}_t, \boldsymbol{\phi}; \boldsymbol{\psi}), \dots ,\text{}u_{n-1}(\V{z}_t, \boldsymbol{\phi}; \boldsymbol{\psi}) u_{n}(\V{z}_t, \boldsymbol{\phi}; \boldsymbol{\psi}), \\
		&\quad \, \text{ }u_{1}(\V{z}_t, \boldsymbol{\phi}; \boldsymbol{\psi}) \varepsilon(\V{z}_t, \gamma; \boldsymbol{\psi}), \dots, u_{n}(\V{z}_t, \boldsymbol{\phi}; \boldsymbol{\psi}) \varepsilon(\V{z}_t, \gamma; \boldsymbol{\psi})]' .
	\end{align*}}
	
	Let the sample weight matrix $\widehat{W}$ be defined as 
	\begin{align*}
	\widehat{W} = \mathrm{diag}\Big(\frac{1}{\widehat{\sigma}_1^2 \widehat{\sigma}_2^2}, \dots, \,\frac{1}{\widehat{\sigma}_1^2 \widehat{\sigma}_n^2}, \frac{1}{\widehat{\sigma}_2^2 \widehat{\sigma}_3^2}, \dots, \frac{1}{\widehat{\sigma}_{n-1}^2 \widehat{\sigma}_n^2}, \frac{1 }{\widehat{\sigma}_1^2} \frac{1}{\widehat{\sigma}_\varepsilon^2}, \dots, \frac{1 }{\widehat{\sigma}_n^2} \frac{1}{\widehat{\sigma}_\varepsilon^2}\Big)	
	\end{align*}		
	where $\widehat{\sigma}_i^2$ and $\widehat{\sigma}_\varepsilon^2$ are consistent estimators of idiosyncratic shock variances $\sigma^2_i$ for $i=1,\dots,n$ and $\sigma^2_\varepsilon$ respectively.	Then the \textbf{demand system robust granular instrumental variables estimator with observed explanatory variables} is characterized as the following two-step estimator:
	\begin{enumerate}
		\item For each $i=1,\dots,n$,
		\begin{align*}
		\widehat{\boldsymbol{\psi}}^{\textrm{Step 1}}_i &= (\frac{1}{T} \sum_{t=1}^T \V{x}_t \V{x}_t')^{-1} (\frac{1}{T} \sum_{t=1}^T \V{x}_t y_{it}) \text{ and }	
		\widehat{\boldsymbol{\psi}}_p = (\frac{1}{T} \sum_{t=1}^T \V{x}_t \V{x}_t')^{-1} (\frac{1}{T} \sum_{t=1}^T \V{x}_t p_{t}).
		\end{align*}

		\item For $\widehat{Q}_T(\boldsymbol{\phi}; \widehat{\boldsymbol{\psi}}^{\text{Step 1}})=  \Big[\frac{1}{T} \sum_{t=1}^T g^c_{\mathrm{demand}}(\V{z}_t, \boldsymbol{\phi};\widehat{ \boldsymbol{\psi}}^{\text{Step 1}}) \Big]' \widehat{W} \Big[\frac{1}{T} \sum_{t=1}^T g^c_{\mathrm{demand}}(\V{z}_t, \boldsymbol{\phi};\widehat{ \boldsymbol{\psi}}^{\text{Step 1}})\Big]$, 		$\widehat{\boldsymbol{\phi}}^{RGIV,c} = \arg\min_{\boldsymbol{\phi}_S < \gamma} \widehat{Q}_T(\boldsymbol{\phi}; \widehat{\boldsymbol{\psi}}^{\text{Step 1}} ).$				
	\end{enumerate}

\end{definition}

\begin{lemma}[Identification with observed explanatory variables, demand system]
	\label{lemma:demand_identification_controls}
	Impose Assumption \ref{ass:demand_controls}. For $g_0^c(\boldsymbol{ \phi}) = \EE[g^c_{\mathrm{demand}}(\V{z}_t, \boldsymbol{\phi}; \boldsymbol{\psi}_0)]$, $g_0^c(\boldsymbol{ \phi}_0)= 0$ for the true parameter $\boldsymbol{\phi}_0$ and $g_0^c(\widetilde{\boldsymbol{ \phi}})  \neq 0$ for $\widetilde{\boldsymbol{ \phi}} \in \boldsymbol{\Phi}$ such that $\widetilde{\boldsymbol{ \phi}} \neq \boldsymbol{ \phi}_0$.	 	
\end{lemma}
\begin{proof}
The result immediately follows after applying Lemma \ref{lemma:demand_identification_no_controls} to $\dot{y}_{it} = y_{it} - \boldsymbol{\psi}_i' \mathbf{x}_t$ and $p_t - \boldsymbol{\psi}_p' \mathbf{x}_t$.
\end{proof}

Assumption \ref{ass:demand_controls} nests the case without explanatory variables. For $k=1$, impose $\beta_i = 0$ for $i=0,\dots,n$. For $u_i(\V{z}_t, \boldsymbol{\phi}) = u_i(\V{z}_t, \boldsymbol{\phi}; 0)$ and $\varepsilon(\V{z}_t, \boldsymbol{\phi}) = \varepsilon(\V{z}_t, \boldsymbol{\phi}; 0)$, the moment function absent control variables is 
\begin{align*}
		g_{\mathrm{demand}}(\V{z}_t, \boldsymbol{\phi}, \gamma) &= [u_{1}(\V{z}_t, \boldsymbol{\phi}) u_{2}(\V{z}_t, \boldsymbol{\phi}), \dots, \, u_{1}(\V{z}_t, \boldsymbol{\phi}) u_{n}(\V{z}_t, \boldsymbol{\phi}),\\
		&\quad \, \text{ }u_{2}(\V{z}_t, \boldsymbol{\phi}) u_{3}(\V{z}_t, \boldsymbol{\phi}), \dots ,\text{}u_{n-1}(\V{z}_t, \boldsymbol{\phi}) u_{n}(\V{z}_t, \boldsymbol{\phi}), \\
		&\quad \, \text{ }u_{1}(\V{z}_t, \boldsymbol{\phi}) \varepsilon(\V{z}_t, \gamma), \dots, u_{n}(\V{z}_t, \boldsymbol{\phi}) \varepsilon(\V{z}_t, \gamma)]' .
	\end{align*}

\begin{lemma}[Identification, demand system no explanatory variables]
	\label{lemma:demand_identification_no_controls}
		Impose Assumption \ref{ass:demand_controls}. Also assume $k=1$ and $\beta_i = 0$ for $i=0,\dots,n$. For $g_0([\boldsymbol{ \phi}', \gamma]') = \EE[g_{\mathrm{demand}}(\V{z}_t, \boldsymbol{\phi}, \gamma)]$, $g_0([\boldsymbol{ \phi}', \gamma]')= 0$ for the true parameter $[\boldsymbol{ \phi}'_0, \gamma_0]'$ and $g_0([\widetilde{\boldsymbol{ \phi}}, \widetilde{\gamma}]')  \neq 0$ for any other $[\widetilde{\boldsymbol{ \phi}}', \widetilde{\gamma}]'$ in the parameter space.	 	 
\end{lemma}
\begin{proof}
	Follow the arguments of the proof of Lemma \ref{lemma:gmm_id} after substituting $p_t$ for $r_{St}$,
	\begin{align}
		\phi_k - \widetilde{\phi}_k &= - \frac{(\phi_1 - \widetilde{\phi}_1) \EE[p_t u_{kt}]}{\EE[p_t u_{1t}] + (\phi_1 - \widetilde{\phi}_1) \EE[p_t^2]} \text{ and } \gamma - \widetilde{\gamma} = - \frac{(\phi_1 - \widetilde{\phi}_1) \EE[p_t \varepsilon_t]}{\EE[p_t u_{1t}] + (\phi_1 - \widetilde{\phi}_1) \EE[p_t^2]}. \label{eq:demand_phitilde_gammatilde}
	\end{align}
	
	Then, for units $i\neq j$,
\begin{align}
	0 &= \frac{(\phi_1 - \widetilde{\phi}_1) \EE[p_t u_{it}] \EE[p_t u_{jt}]}{\EE[p_t u_{1t}] + (\phi_1 -\widetilde{\phi}_1) \EE[p_t^2]}\left\{-2 + \frac{(\phi_1 - \widetilde{\phi}_1) \EE(p_t^2)}{\EE(p_t u _{1t}) + (\phi_1 - \widetilde{\phi}_1) \EE(p_t^2)} \right\}.
	\label{eq:demand_twosolutions}
\end{align}
Compare to Equation \ref{eq:general_momcond}. The first term equals zero when $\phi_1 - \widetilde{\phi}_1 = 0$. Plugging into Equation \ref{eq:demand_phitilde_gammatilde}, $\widetilde{\boldsymbol{\phi}} = \boldsymbol{\phi}$ and $\widetilde{\gamma} = \gamma$.

Now focus on the second term of Equation \ref{eq:demand_twosolutions}. The second term equals zero when $\EE[p_t u_{1t}] = -\frac{1}{2} (\phi_1 - \widetilde{\phi}_1) \EE[p_t^2]$. For $k \geq 1$, substituting into Equation \ref{eq:demand_phitilde_gammatilde} gives
\begin{align*}
	\phi_k - \widetilde{\phi}_k &= -2 \frac{\EE[p_t u_{kt}]}{\EE[p_t^2]} \text{ and } \gamma - \widetilde{\gamma} = -2 \frac{\EE[p_t \varepsilon_t]}{\EE[p_t^2]}.
\end{align*}
Compare to Equation \ref{eq:general_falsesolution}. Then, $\widetilde{\phi}_S - \widetilde{\gamma} = (\phi_S - \gamma) + \frac{2\{\EE[p_t u_{St}] - \EE[p_t \varepsilon_t] \}}{\EE[p_t^2]} = -(\phi_S - \gamma) > 0.$
Thus, $(\widetilde{\boldsymbol{\phi}}', \widetilde{\gamma})'$ falls outside the parameter space.

\end{proof}

With identification, consistency and asymptotic normality are straightforward extensions of Theorems \ref{thm:gmm_consistency_controls} and \ref{thm:gmm_asymptotic_normality_controls}.

		\end{document}

%% file: out/ffunds_results.tex
\begin{table}[t]
    \centering
    \small
    \begin{tabular}{l|cccccc}
        \toprule
     & Preferred & \makecell{Brokers \& \\ dealers to  \\Real Domestic} & \makecell{State \& local \\ govts.\! to \\ Retire.\!/Insur.\!} & \makecell{Closed end \\ to Real \\Domestic} & CUE & \makecell{1-factor \\ GIV} \\
     \hline
     RGIV results: &&&&&&\\
\quad$\phi_S$ & $-0.05$ & $-0.05$ & $-0.04$ & $-0.05$ & $-0.06$ & \\
& $(0.01)$ & $(0.01)$ & $(0.01)$ & $(0.01)$ & $(0.02)$ \\
\quad$\phi_\text{Retire.\!/Insur.\!}$ & $-0.20$ & $-0.21$ & $-0.17$ & $-0.19$ & $-0.19$ \\
& $(0.02)$ & $(0.02)$ & $(0.02)$ & $(0.02)$ & $(0.05)$ \\
\quad$\phi_\text{Interm.\!/Funds}$ & $0.09$ & $0.11$ & $0.12$ & $0.09$ & $0.07$ \\
& $(0.03)$ & $(0.03)$ & $(0.03)$ & $(0.03)$ & $(0.04)$ \\
\quad$\phi_\text{Foreign}$ & $0.00$ & $-0.02$ & $0.00$ & $0.00$ & $-0.04$ \\
& $(0.02)$ & $(0.02)$ & $(0.02)$ & $(0.02)$ & $(0.03)$ \\
\quad$\phi_\text{Real Domestic}$ & $-0.08$ & $-0.08$ & $-0.07$ & $-0.07$ & $-0.07$ \\
& $(0.02)$ & $(0.02)$ & $(0.02)$ & $(0.02)$ & $(0.03)$ \\
\quad$\gamma$ & $0.20$ & $0.16$ & $0.22$ & $0.19$ & $0.21$ \\
& $(0.04)$ & $(0.02)$ & $(0.04)$ & $(0.03)$ & $(0.05)$ \\
$\phi^{GIV}$ & &&&&& $-0.15$\\
&&&&&& $(0.08)$\\
\hline
Tests ($p$-values): &&&&&&\\
\quad Specification & $0.258$ & $0.278$ & $0.274$ & $0.278$ & $0.006$ \\
\quad Homogeneity & $<0.001$ & $<0.001$ & $<0.001$ & $<0.001$ & $0.001$ \\
\makecell{First-stage $F$-stat.\!} & &&&&& $7.6$\\
\bottomrule
    \end{tabular}
\caption{Demand system estimation results. RGIV coefficient estimates are listed above their standard errors, which are provided in parentheses. $p$ values are provided in
the bottom section of the table for the specification test and parameter homogeneity tests.}
\label{table:ffunds_results}
\end{table}

%% file: out/main_T=2300.tex
\begin{table}[t!]
	\centering
	\small
	\begin{tabular}{l|cc|cc|cc}
	\toprule
\multicolumn{1}{c}{DGP}		&	\multicolumn{2}{|c|}{RGIV} & \multicolumn{2}{c|}{GIV} & \multicolumn{2}{c}{Testing ($\alpha=0.05$)} \\
		 & $\widehat{\phi_S}$ & $\widehat{\phi_E}$ & Feasible & Oracle & Spec. & Homog. \\
		\hline
Homogeneous spillovers & 0.94  & 0.97  & 0 & 0.95  & 0.054 & 0.042\\
& (0.12) &(0.038)&(0.058)  &(0.046) && \\
Coefficient outlier & 0.94  & 0.97  & 0  & 0.15  & 0.047 & 0.998\\
& (0.12) &(0.037)&(0.078)&(0.071)& & \\
Variance outlier & 0.97  & 0.94  & 0.0068  & 0.95  & 0.045 & 0.052\\
& (0.045) &(0.067)&(0.037)&(0.033) &&\\
Application & 0.95  & 0.97  & 1.00  & 1.00  & 0.064 & 0.996\\
&(0.14) &(0.052)&(0.083)&(0.10)&&\\
Short $T$ & 0.90 & 0.98 & 0.69 & 0.94 & 0.060 & 0.061 \\
&(0.61) & (0.21) & (0.27) & (0.22) &  &  \\
Near-homogeneous size & 0.95 & 0.95 & 0.92 & 0.63 & 0.050 & 0.053 \\
& (0.0308) & (0.0308) & (0.925) & (1.35) &  &  \\		
		\bottomrule
	\end{tabular}

\caption{Simulation results (5000 simulations, $T = 2283$, 95\% confidence intervals).  Empirical coverage probabilities are listed above median confidence interval lengths, which are provided in parentheses. RGIV: Coverage probabilities are computed as the proportion of confidence intervals containing the estimand ($\phi_S$ for $\widehat{\phi}_S$ and $\phi_E$ for $\widehat{\phi}_E$). GIV: Coverage probabilities are computed as the proportion of confidence intervals containing any positive-weighted average of $\phi_i$ for the feasible and oracle GIV estimators. Computed at the 5\% level, the RGIV specification test (Spec.) and coefficient homogeneity test (Homog.) are provided in the final two columns.}
\label{table:simulation_main}
\end{table}